\documentclass[times,sort&compress,3p]{elsarticle}
\usepackage[labelfont=bf]{caption}

\usepackage{amsmath,amsfonts,amssymb,amsthm,booktabs,color,epsfig,graphicx,hyperref,url}
\usepackage{multirow}
\usepackage{algorithmicx,algorithm,algpseudocode}
\usepackage{verbatim}
\usepackage{float}

\usepackage[super]{nth}
\usepackage[all]{xy}
\usepackage{subfigure}
\usepackage{mathrsfs}
\usepackage{bm}
\usepackage{footnote}
\usepackage{ulem}
\pdfoutput=1
\theoremstyle{plain}
\newtheorem{theorem}{Theorem}

\newtheorem{lemma}{Lemma}

\theoremstyle{definition}

\begin{document}

\begin{frontmatter}

\title{Subgroup analysis for the functional linear model}

\author[1]{Yifan Sun\corref{equalcont}}
\author[1]{Ziyi Liu\corref{equalcont}}
\author[1]{Wu Wang\corref{mycorrespondingauthor}}

\address[1]{Center for Applied Statistics, School of Statistics, Renmin University of China, China}

\cortext[equalcont]{equal contributions from these authors}
\cortext[mycorrespondingauthor]{Corresponding author. Email address: \url{wu.wang@ruc.edu.cn}}

\begin{abstract}
Classical functional linear regression models the relationship between a scalar response and a functional covariate, where the coefficient function is assumed to be identical for all subjects. In this paper, the classical model is extended to allow heterogeneous coefficient functions across different subgroups of subjects. The greatest challenge is that the subgroup structure is usually unknown to us. To this end, we develop a penalization-based approach which innovatively applies the penalized fusion technique to simultaneously determine the number and structure of subgroups and coefficient functions within each subgroup. An effective computational algorithm is derived. We also establish the oracle properties and estimation consistency. Extensive numerical simulations demonstrate its superiority compared to several competing methods. The analysis of an air quality dataset leads to interesting findings and improved predictions. 
\end{abstract}

\begin{keyword} 
ADMM algorithm \sep
B-spline basis \sep
Functional linear regression \sep
Minimax concave penalty \sep
Subgroup analysis
\MSC[2020] Primary 62H12 \sep
Secondary 62F12
\end{keyword}

\end{frontmatter}

\section{Introduction\label{sec1}}

With the rapid development of information acquisition and transmission technologies, more and more data are being collected continuously during a time interval. Such data are coined as ``functional data". As a new branch of statistics in recent decades, functional data analysis (FDA) encompasses statistical methods and techniques for such data. See \cite{Ramsay02, Ramsay05} for a comprehensive introduction. Functional linear models (FLMs) are useful for modeling the linear relationship between a scalar response and a functional covariate: 
\begin{equation}\label{homo_model1} 
	y_i=\int_{\mathcal{X}} X_i(t)\beta(t)dt+\epsilon_i, i\in\{1,\ldots, n\},
\end{equation}
where $X_i(t)$ is a square-integrable predictor function defined on a compact support $\mathcal{X}$ of $\mathbb{R}$, $\beta(t)$ is a square-integrable coefficient function defined on $\mathcal{X}$, and the random errors $\epsilon_i$ are uncorrelated with a zero mean and a constant variance. FLMs have been studied extensively in the literature \citep{cardot99, cardot03, Hilgert2013}.

In traditional FLMs, the relationship between the response and functional covariate is often assumed to be identical for all subjects. However, due to latent factors or unobserved covariates, the subjects may belong to different subgroups that have heterogeneous relationships between the response of interest and an observed functional covariate. For example, \cite{yao12} discovered two subgroups of medflies, where the effect patterns of the early fertility process on the lifespan are different across two subgroups. In the study by \cite{preda05}, eighty-four shares quoted at the Paris stock exchange were divided into several subgroups, where the evolution of the index in a past time period has varied effects on the future index for different subgroups. Also, in the empirical analysis conducted in this paper, we find that different regions in China have distinct patterns of association between the NO$_2$ concentration and average PM$_{2.5}$ concentration. 

For the heterogeneous FLMs, some subgroup analysis approaches have been proposed to identify latent group structures. \cite{preda05} developed an approach that integrates the partial least squares and functional principal components analysis, but the statistical properties of the proposed approach have not been discussed. Recently, \cite{yao12} proposed a functional mixture regression (FMR) where the coefficient functions are allowed to vary for different groups of subjects. However, similar to the well-known finite mixture models for traditional linear regressions, FMR also needs to specify a parametric assumption for the conditional density of the scalar responses, often difficult to do in practice. 

In addition, the aforementioned subgroup analysis approaches must specify the number of subgroups, which can be problematic in application. Recently, the penalized fusion technique \citep{Tibshirani05} has attracted extensive attention in subgroup analysis \citep{Ma2017, Ma2020, Hu2021}
and is advantageous in multiple aspects. Specifically, it has a more intuitive definition, determines the number of groups automatically, and can, in principle, accommodate small groups. However, the existing studies focus on linear regression models with scalar covariates, and cannot be directly applied to FLMs. 

In this article, we propose a penalized fusion approach for heterogeneous FLMs that can divide the subjects into disjoint subgroups with different coefficient functions. This approach applies a concave penalty to pairwise differences of B-spline coefficient vectors. Compared to FMR, it does not need to specify the number of subgroups in advance; instead, it automatically determines the number of subgroups in a data-driven manner. This, thus, allows us to identify the subgroup structure and estimate coefficient functions simultaneously without any prior knowledge about the subgroup structure. 

This work is motivated by the air pollution study in China that focuses on how the average $\text{PM}_{2.5}$ relates to $\text{NO}_{2}$, and exploring distinctive relationship patterns between $\text{PM}_{2.5}$ and $\text{NO}_{2}$ in different geographical regions, based on the data collected from hundreds of air monitoring stations across China. As a result, we observe a clear spatial clustering that is highly consistent with geographical regionalization of China, where each cluster corresponds to a coefficient function. Moreover, some geographically nonadjacent stations with similar climate characteristics are also combined in the same cluster. Compared with alternative methods, the proposed approach gives a very different subgrouping results, and has a better prediction performance.

To implement the proposed approach, we develop an alternating direction method of multipliers (ADMM) algorithm, which is widely used in penalization-based methods and has desirable convergence properties \citep{Boyd2011}. Moreover, we establish the consistency property for the proposed approach in identifying the subgroup memberships and estimating the coefficient functions. Specifically, we, firstly, establish the consistency of the oracle estimator with prior knowledge of the true subgroup structure. Then we show that, under mild regularity conditions, the oracle estimator is a local minimizer of the objective function with a high probability. Extensive simulation and real data analysis demonstrate that the proposed approach outperforms several competitors in identifying subgroups and estimating coefficient functions.   

The rest of this paper is organized as follows. We propose the heterogeneous functional linear model, introduce a penalized fusion estimation approach along with the ADMM algorithm, and establish theoretical properties in Section \ref{sec2}. In Section \ref{sec3}, we evaluate the finite sample properties of the proposed approach via simulation studies. We use the proposed approach to analyze an air quality dataset in Section \ref{sec4}. Section \ref{sec5} gives some concluding remarks. The details of the algorithm, proofs of theorems, and additional results of numerical simulations are provided in the supplementary material. 

\section{Methods\label{sec2}}

\subsection{Penalized fusion estimation}\label{sec31}
Consider a dataset with $n$ independent samples, each with a scalar response and a functional covariate. For the $i$th subject, let $y_i$ be the response, and $X_i(\cdot)$ be the functional covariate defined on a compact set $\mathcal{X}=[0,T]$. Assume that the data have been standardized such that there is no intercept. Consider the 
heterogeneous FLM: 
\begin{equation}\label{model_ori}
	y_i=\int_{\mathcal{X}}X_{i}(t)\beta_{i}(t)dt+\epsilon_i, i\in\{1,\ldots, n\},
\end{equation}
where the random errors $\epsilon_i$'s are uncorrelated with mean 0 and variance $\sigma^2$, and $\beta_{i}(\cdot)$ is the unknown coefficient function 
defined on $\mathcal{X}$. Different from the classical FLM, each subject has its own unique coefficient function. 

We assume that $n$ subjects form $K$ disjoint subgroups and the coefficient functions are the same within each subgroup. Specifically, let $\mathcal{G}=\{\mathcal{G}_1,\ldots,\mathcal{G}_K\}$ be a partition of $\{1,\ldots, n\}$, where $K(K\leq n)$ is the number of subgroups. Then $\beta_i=\beta_j$ for any $i,j\in\mathcal{G}_k$. As such, subgroup analysis amounts to determining which $\beta_i$'s are equal to each other. To this end, we adopt a penalized fusion approach, which is more flexible than the FMR and some other techniques. For example, the number of subgroups does not need to be assumed a prior, but rather is determined fully data-driven. In addition, it can potentially accommodate smaller subgroups. 

Specifically, we first leverage the B-spline basis method to approximate the coefficient function $\beta_i(\cdot)$'s. The $q$th-order B-spline basis functions with a set of $m$ internal knots $\{\kappa_l\,|\,\kappa_{-q+1}=\ldots =\kappa_0<\kappa_1<\ldots <\kappa_{m+1}=\ldots =\kappa_{m+q}, \kappa_l=lT/(m+1), l\in\{0,\ldots,m+1\}\}$ are defined recursively \citep{Boor1978} as 
\begin{equation*}
	B_l^1(t)=\left\{
	\begin{aligned} 
		1 & & \kappa_{l-1} \leq x<\kappa_{l}\\
		0 & & \text{otherwise}
	\end{aligned}
	\right.
\end{equation*}
and 
\begin{equation*}
	B_l^q(t)=\frac{t-\kappa_{l-1}}{\kappa_{q+l-2}-\kappa_{l-1}}B_{l-1}^{q-1}(t)+\frac{\kappa_{q+l-1}-t}{\kappa_{q+l-1}-\kappa_{l}}B_{l}^{q-1}(t), \; l\in\{1,\ldots,m+q\},
\end{equation*}
where $t\in [0,1]$. Let $p=m+q$. Then the coefficient function $\beta_i(t)$ can be approximated through its projection on the B-spline space as 
\begin{equation*}
	\beta_i(t)\approx\sum_{l=1}^{p} B_l^q(t)\theta_{il}\equiv \bm{B}\bm \theta_i, 
\end{equation*}
where $\bm{B}=(B_1^q(t),\ldots,B_{p}^q(t))$ is a B-spline basis vector, $\bm{\theta}_i=(\theta_{i1},\cdots,\theta_{ip})^\top=\underset{\bm{\theta_i}\in\mathbb{R}^p}{argmin}\Vert\beta_i-\bm{B}\bm \theta_i\Vert_2$ is a $p$-dimensional coefficient vector, and $\Vert f\Vert_2 = \sqrt{\int_\mathcal{X}f(t)^2\text{d}t}$ as the $L_2$ norm of a function $f$. Our goal is to identify the latent subgroup structures of $n$ subjects in terms of coefficient functions $\beta_i$'s. This is, thus, equivalent to distinguishing between B-spline coefficient vectors $\bm\theta_i$'s. 

Then, we propose the following objective function 
\begin{equation}\label{lossfunc}
	\begin{split}
		Q_n(\bm{\theta};\bm \lambda)&= \frac{1}{2}\sum_{i=1}^n(y_i-\bm H_i\bm \theta_i)^2 +\frac{1}{2}\lambda_1\sum_{i=1}^n\Vert D^2 (\bm B\bm\theta_i)\Vert_2^2\\
		&+\sum_{i<j}\rho (\bm\theta_i-\bm\theta_j,\omega_{ij}\lambda_2), 
	\end{split}
\end{equation}
where $\bm\theta=(\bm\theta_1^\top,\ldots,\bm\theta_n^\top)^\top$, $\bm H_i=(\int_\mathcal{X}X_i(t)B_1(t)dt,\ldots,\int_\mathcal{X}X_i(t)B_p(t)\text{d}t)$, and $\bm\lambda=(\lambda_1, \lambda_2) $ are tuning parameters. In the first penalty term, we denote $D^2$ as the second-order differential operator. 
In the second penalty term, we take $\rho(\cdot,\cdot)$ as the minimax concave penalty \citep[MCP,][]{Zhang2010}, to obtain a nearly unbiased and sparse solution. Specifically, $\rho (\bm\theta_i-\bm\theta_j,\omega_{ij}\lambda_2)=\rho_\tau(\Vert\bm\theta_i-\bm\theta_j\Vert_2,\omega_{ij}\lambda_2)$, where $\rho_{\tau}(t,\gamma)=\gamma\int_0^t(1-\frac{x}{\tau\gamma})_{+}\text{d}x$ and $\omega_{ij}$'s are pre-specified weights. Note that the Smoothly Clipped Absolute Deviation Penalty \citep[SCAD,][]{Fan2001} and some other concave penalties are also applicable. We obtain the B-spline coefficient vector estimator $\hat{\bm \theta}=(\hat{\bm\theta}_1^\top,\ldots,\hat{\bm\theta}_n^\top)^\top$ by minimizing (\ref{lossfunc}) with respect to $\bm\theta$, and the corresponding coefficient function estimator is $\hat{\beta}_i=\bm B\hat{\bm\theta}_i$ for $i\in\{1,\ldots, n\}$. Subjects $i$ and $j$ belong to the same subgroup if and only if $\hat{\beta}_i=\hat{\beta}_j$. Consequently, the prediction of the $i$th response is given by $\hat y_i=\int_\mathcal{X} X_i(t)\hat\beta_i(t)dt$. Let $\bm{\hat y}=(\hat y_1,\cdots,\hat y_n)^\top$.

The objective function takes a penalized fusion form. The first term of the objective function measures the lack of fit. The second term controls the smoothness of the estimated coefficient functions. In particular, we adopt the smooth penalty based on the 2nd-order differential of $\bm B\bm\theta_i$, i.e., $\beta_i$, which is a popular choice in functional data analysis. The third term takes a pairwise fusion form and promotes zero differences of B-spline coefficient vectors and hence grouping. The weights $\omega_{ij}$'s allow flexible strength of penalization. For example, in our real-data analysis in Section \ref{sec4}, we take $\omega_{ij}$ as the inverse of the distance between locations $i$ and $j$ to encourage locations with a smaller distance to fall into the same subgroups.  

\subsection{Computing algorithm and implementations}\label{sec32}
Let $\bm y=(y_1,\ldots, y_n)^\top$, and $\bm H=\text{diag}\{\bm H_1,\ldots,\bm H_n\}$. Some algebra shows that we can rewrite the objective function $Q_n(\bm\theta;\bm\lambda)$ as 
\begin{equation}\label{lossfunc2}
	Q_n(\bm\theta;\bm\lambda)=\frac{1}{2}(\bm{y}-\bm{H}\bm{\theta})^\top(\bm{y}-\bm{H}\bm{\theta})+\frac{1}{2}\lambda_1\bm{\theta}^\top\bm{G}\bm{\theta}+\sum_{i<j}\rho_{\tau}(\Vert \bm\theta_i-\bm\theta_j\Vert_2,\omega_{ij}\lambda_2). 
\end{equation}
Here $\bm G=\bm I_n\otimes \bm G_0$, where $\otimes$ is the Kronecker product, $\bm I_n$ is the $n\times n$ identity matrix, and $\bm G_0$ is the $p\times p$ matrix with element $(\bm G_{0})_{sl} = \int_\mathcal{X} B''_s(t)B''_l(t)\text{d}t$.  

Directly minimizing the objective function (\ref{lossfunc2}) is challenging because the penalty function is not separable in $\bm\theta_i$'s. As such, we reparameterize by introducing a new set of parameters $\bm \eta_{ij}=\bm \theta_i-\bm\theta_j$. Hence minimizing (\ref{lossfunc2}) is equivalent to 
the constrained optimization problem: 
\begin{equation*}\label{lossfunc3}\centering
	\begin{split}
		\min Q_n(\bm{\theta},\bm{\eta};\bm{\lambda})\equiv&\frac{1}{2}(\bm{y}-\bm{H}\bm{\theta})^\top(\bm{y}-\bm{H}\bm{\theta})+\frac{1}{2}\lambda_1\bm{\theta}^\top\bm{G}\bm{\theta}+\sum_{i<j}\rho_{\tau}(\Vert\bm\eta_{ij}\Vert_2,\omega_{ij}\lambda_2),\\
		&\text{subject to}\quad \bm\theta_i-\bm\theta_j-\bm\eta_{ij}=0,
	\end{split}
\end{equation*}
where $\bm{\eta}=\{\bm\eta_{ij}^\top,i<j\}^\top$. The augmented Lagrangian function reads as
\begin{equation}\label{lossfunc5}\centering
	\begin{split}
		L_n(\bm{\theta}, \bm{\eta},\bm{\zeta}; \bm{\lambda})&=\frac{1}{2}(\bm{y}-\bm{H}\bm{\theta})^\top(\bm{y}-\bm{H}\bm{\theta})+\frac{1}{2}\lambda_1\bm{\theta}^\top\bm{G}\bm{\theta}+\sum_{i<j}\rho_{\tau}(\Vert\bm\eta_{ij}\Vert_2,\omega_{ij}\lambda_2)\\
		&+\sum_{i<j}\bm\zeta_{ij}^\top(\bm\eta_{ij}-\bm\theta_i+\bm\theta_j)+\frac{\delta}{2}\sum_{i<j}\Vert\bm\eta_{ij}-\bm\theta_i+\bm\theta_j\Vert_2^2,
	\end{split}
\end{equation}
where $\bm\zeta=\{\bm\zeta_{ij}, i<j\}^\top$ are the Lagrange multipliers and $\delta$ is the penalty parameter. We adopt the alternating direction method of multipliers (ADMM) algorithm to compute the estimate of $(\bm{\theta}, \bm{\eta},\bm{\zeta})$. More details are provided in the supplementary material. 

Given the estimate $(\bm{\theta}^{(s)}, \bm{\eta}^{(s)},\bm{\zeta}^{(s)})$ at the $s$th iteration, the $(s+1)$th iteration estimates are 
\begin{equation}\label{eq:theta}
	\bm{\theta}^{(s+1)}=\left(\bm{H}^\top\bm{H}+\lambda_1\bm{G}+\delta\bm{A}^\top\bm{A}\right)^{-1}\left[\bm{H}^\top\bm{y}+\delta\bm{A}^\top(\bm{\eta}^{(s)}+\frac{1}{\delta}\bm{\zeta}^{(s)})\right], 
\end{equation}
where $\bm A=\Delta \otimes \bm I_p$. Here $\Delta=\{(\bm e_i-\bm e_j),i<j\}^\top$, with $\bm e_i$ being the $n$-dimensional column vector whose $i$th element is 1 and others are 0, $\bm I_p$ is a $p\times p$ identity matrix and $\otimes$ is the Kronecker product. 
\begin{equation}\label{eq:eta}
	\bm\eta_{ij}^{(s+1)}=
	\begin{cases}
		\bm u_{ij}^{(s+1)}\quad& \text{if}\ \Vert \bm u_{ij}^{(s+1)}\Vert_2\geq\tau\omega_{ij}\lambda_2,\\
		\frac{\tau\delta}{\tau\delta-1}\left(1-\frac{\lambda_2}{\delta\Vert \bm u_{ij}^{(s+1)}\Vert_2}\right)_+\bm u_{ij}^{(s+1)} \quad&\text{if}\ \Vert\bm u_{ij}^{(s+1)}\Vert_2<\tau\omega_{ij}\lambda_2,
	\end{cases}
\end{equation}
where $\bm u_{ij}^{(s+1)}=\bm \theta_i^{(s+1)}-\bm \theta_j^{(s+1)}-\frac{\bm\zeta_{ij}^{(s)}}{\delta}$, 

and 
\begin{equation}\label{eq:zeta}
	\bm\zeta^{(s+1)}=\bm\zeta^{(s)}+\delta(\bm A\bm\theta^{(s+1)}-\bm\eta^{(s+1)}).
\end{equation}

The overall algorithm is summarized in Algorithm \ref{alg1}. 
\begin{algorithm}[h]
	\caption{ADMM for minimizing (\ref{lossfunc5})}
	\label{alg1}
	\begin{algorithmic}[1]
		\State Initialize $\bm \theta^{(0)}$, $\bm\eta^{(0)}=\bm A \bm \theta^{(0)}$, and $\bm \zeta^{(0)}=\bm 0$. Set $s=0$.
		\Repeat
		\State Update $ \bm{\theta}^{(s+1)}$ via (\ref{eq:theta});
		\State Update $\bm\eta^{(s+1)}$ via (\ref{eq:eta});
		\State Update $\bm\zeta^{(s+1)}$via (\ref{eq:zeta});
		\State $s=s+1$.
		\Until the convergence criterion is met.
		\State \Return the estimates of $\bm\theta$ and $\bm\eta$ at convergence.
	\end{algorithmic}
\end{algorithm}

\noindent\textbf{Remark 1}. In non-convex optimization, it is important to assign appropriate initial values to obtain a good solution. In the numerical study, we use estimates from a homogeneous functional linear model as the initial $\bm\theta^{(0)}$. Specifically, we assume that all subjects have the same coefficient function, and hence, the same B-spline coefficient vector $\bm{\tilde\theta}$, which is a $p$-dimensional column vector. We apply the penalized B-spline (without fusion penalization) for all $n$ subjects, and obtain $\bm{\tilde\theta}$ by minimizing the following objective function: 
\[(\bm y-\bm{\tilde H}\bm{\theta})^\top (\bm y-\bm{\tilde H}\bm{\theta})+\lambda_1 \bm\theta^\top\bm G_0\bm{\theta},\]
where $\bm{\tilde H}=(\bm{H}_1^\top,\ldots,\bm{H}_n^\top)^\top$, $\lambda_1$ is determined by minimizing GCV, which will be explicitly defined in the following. The solution is $\bm{\tilde\theta}=\left(\bm{\tilde H}^\top\bm{\tilde H}+\lambda_1 \bm G_0\right)^{-1}\bm{\tilde H}\bm y$. Then we assign $\bm\theta^{(0)}=(\underbrace{\bm{\tilde\theta}^\top,\ldots,\bm{\tilde\theta}^\top}_{n})^\top$. 

\noindent\textbf{Remark 2}. As the standard ADMM algorithm, we track the progress of the algorithm based on the primal residual $\bm R^{(s+1)}=\bm A\bm\theta^{(s+1)}-\bm\eta^{(s+1)}$ and dual residual $\bm D^{(s+1)}=\delta \bm A^\top (\bm \eta^{(s+1)}-\bm \eta^{(s)})$. The algorithm is terminated when $\Vert \bm R^{(s+1)} \Vert_2\leq \epsilon^p$ and $\Vert \bm D^{(s+1)} \Vert_2\leq \epsilon^d$, where $\epsilon^p$ and $\epsilon^d$ are some small values according to \citep{Boyd2011}: 
$$\epsilon^p=\sqrt{np}\epsilon^{abs}+\epsilon^{rel}\Vert\bm{A}^\top\bm\zeta^{(s)}\Vert_2,\ \ \epsilon^d=\sqrt{\frac{n(n-1)p}{2}}\epsilon^{abs}+\epsilon^{rel}\underset{ }\max\{\Vert\bm A\bm\theta^{(s)}\Vert_2,\Vert\bm\eta^{(s)}\Vert_2\}$$
Here $\epsilon^{abs}$ and $\epsilon^{rel}$ are pre-specified small constants.

\noindent{\textbf{Remark 3}}. For a sufficient large $\lambda_2$, the MCP penalty can force $\bm\eta_{ij}=0$. We put subjects $i$ and $j$ into the same subgroup if $\bm\eta_{ij}=0$. As a result, we have $\hat{K}$ estimated subgroups $\hat{\mathcal{G}}_1,\ldots,\hat{\mathcal{G}}_{\hat{K}}$, and let the estimated B-spline coefficient function for the $k$th subgroup be $\hat{\bm \alpha}_k=\sum_{i\in \hat{\mathcal{G}}_{\hat{k}}}\hat{\bm\theta_i}/|\hat{\mathcal{G}}_k|$, where $|\hat{\mathcal{G}}_k|$ is the cardinality of $\hat{\mathcal{G}}_k$.  

\noindent{\textbf{Tuning procedure}}. The proposed approach involves two tuning parameters $\lambda_1$ and $\lambda_2$, where $\lambda_1$ controls the smoothness of the estimated coefficient functions and $\lambda_2$ controls the structure of subgroups. To reduce the computational cost, we adopt a two-step procedure \citep{Zhu2018} instead of a traditional grid search for $\lambda_1$ and $\lambda_2$. This approach, firstly, searches for an optimal value of $\lambda_2$ while fixing $\lambda_1$ at a small value (in our numerical analysis, $\lambda_1=0.005$), and then selects $\lambda_1$ given the optimal $\lambda_2$ from the first step. 
Following \cite{Wang2009} and \cite{Ma2017}, we choose $\lambda_2$ in the first step by minimizing a modified BIC
\begin{equation}
	\label{eq:bic}
	\text{BIC}=\log\left(\frac{\Vert\bm{y}-\bm{\hat y}\Vert_2^2}{n}\right)+C_{n,p}\frac{\log(n)}{n}\text{d}f,
\end{equation}  
where $\text{d}f=\hat Kp$ and $C_{n,p}=\log(\log(n+p))$. 
And in the second step, we determine the optimal $\lambda_1$ by minimizing 
\begin{equation}
	\label{eq:gcv}
	\text{GCV}=\frac{\Vert\bm{y}-\bm{\hat y}\Vert_2^2}{\left(1-tr\{\bm{\tilde H}(\bm{\tilde H}^\top\bm{\tilde H}+\lambda_1\bm{G})^{+}\bm{\tilde H}^\top\}/n\right)^2},
\end{equation}
where $(\bm{\tilde H}^T\bm{\tilde H}+\lambda_1\bm{G})^{+}$ is the Moore-Penrose inverse of $\bm{\tilde H}^T\bm{\tilde H}+\lambda_1\bm{G}$ and $\bm{\tilde H}$ is a variant version of $\bm{H}$. More precisely, $\bm{\tilde H}=\text{diag}\{\tilde{H_1},\cdots,\tilde{H}_{\hat K}\}$, where $\tilde{H_i}=\{H_i^\top|i\in\mathcal{G}_{\hat k}\}^\top$. The tuning procedure can be iterated more than one round until satisfactory parameters are found.

\noindent{\bf Realization}
To implement the proposed approach, we have developed an \textsf{R} package and made it publicly available at \url{https://github.com/breakerkun/FLM-clustering}. 
Also, we have provided a demo for two sample datasets with two and three subgroups, respectively. The proposed approach is computationally affordable. For example, for a simulated dataset with $n=200$, $q=4$, $m=8$, and two subgroups, an overall analysis can be accomplished within 75 seconds using a laptop with 2 Intel(R) Core(TM) i5-8250U CPU @ 1.60 GHz CPU cores and 8G RAM. 
To speed up the computation, all the ADMM related procedures are coded in \textsf{C++}.  

\subsection{Statistical properties}\label{sec:theory}
We assume that $n$ subjects form $K$ disjoint subgroups. Denote $\mathcal{G}=(\mathcal{G}_1,\ldots,\mathcal{G}_K)$ as the subgroup set, and $|\mathcal{G}_{\min}|=\min_{1\leq k\leq K}|\mathcal{G}_k|$. Let $\bm{\beta}^0=(\beta_1^0,\ldots,\beta_n^0)^\top$ be the true coefficient functions corresponding to the true subgroup set $\mathcal{G}$, and $\bm\xi^0=(\xi^0_1,\ldots,\xi^0_K)$ be the distinct values of $\bm{\beta}^0$. We assume the following conditions: \\
(C1) The error terms in model (\ref{model_ori}) are uncorrelated with a mean 0 and a variance $\sigma^2>0$.\\
(C2) For each coefficient function $\beta_i$ ($i\in\{1,\ldots, n\}$), $\beta_i\in C^{q-2}(\mathcal{X})$ is a $(q-2)$-th order ($q \geq 4$) continuously differentiable function defined on a compact set $\mathcal{X}=[0,T]$.\\
(C3) $\Vert{X}_i\Vert_2 \leq C_1< \infty$ for $i\in\{1,\ldots, n\}$, where $C_1$ is a constant.\\
(C4) $|\mathcal{G}_{k}|=O(n)$ for $k\in\{1,\ldots, K\}$, where $n=\sum_{k=1}^K|\mathcal{G}_{k}|$.\\
(C5) The penalty function $\rho(x,\lambda)$ is symmetric, non-decreasing and concave in $[0,\infty)$. It is a constant for $x>a\lambda$ for some constant $a>0$, and $\rho(0,\lambda)=0$. $\rho'(x,\lambda)$, the derivative of $\rho(x,\lambda)$, exists and is  continuous except for a finite number of $x$, and $\lambda^{-1}\rho'(0^{+},\lambda)=1$.\\
\indent Conditions (C1) and (C2) are standard assumptions for nonparametric B-spline smoothing functions \citep{Zhou1998, Claeskens2009, Lan2010}. Note that $\beta_i$ is assumed in $C^{q-2}(\mathcal{X})$ so as to be consistent with the fact that B-splines of order $q$ are $(q-2)$-th order continuously differentiable. Condition (C3) assumes that the covariate function is bounded. Similar conditions have been considered by \cite{Bosq2000} and \cite{cardot03} to ensure the identifiability of penalized functional linear model, that is, the existence and unicity of the coefficient functions. 
Condition (C4) gives the diverge rate of the subgroup size, that is, the subgroup size grows as the sample size increases. Condition (C5) is commonly assumed in a high-dimensional model, and some concave penalties such as MCP and SCAD satisfy it. 

If the underlying subgroups $\mathcal{G}_1,\ldots, \mathcal{G}_K$ are known, the oracle estimator of B-spline coefficients, denoted as $\bm{\hat\theta}^{\text{or}}$, can be defined as 
\[\bm{\hat\theta}^{\text{or}}=\underset{\bm{\theta}\in\mathcal{M}_{\mathcal{G}}}{\arg\min}(\bm{y}-\bm{H}\bm{\theta})^\top(\bm{y}-\bm{H}\bm{\theta})+\lambda_1\bm{\theta}^\top\bm{G}\bm{\theta},\]
where $\mathcal{M}_\mathcal{G}=\{\bm\theta\in\mathbb{R}^{np}: \bm\theta_i=\bm\theta_j, \text{for any}\ i,j\in\mathcal{G}_k, 1\leq k\leq K\}$. 
Then we can obtain the oracle estimator of the coefficient functions by $\hat{\bm{\beta}}^{\text{or}}=\bm{\mathcal{B}}\bm{\hat\theta}^{\text{or}}$, where $\bm{\mathcal{B}}=\bm I_n\otimes \bm B$. 
Note that all the expectations, variance, and probabilities mentioned below are conditional on covariate functions $X_1,\ldots, X_n$, and the symbolic notations of condition are omitted for simplicity. 

\begin{theorem}\label{thm1}
	Let $\lambda_1\sim |\mathcal{G}_{\min}|^{-(1-\sigma_0)/2}$ for some $0<\sigma_0\leq \frac{3}{5}$, and suppose that the number of knots $m\to \infty$ as $n\to \infty$ and $m=o(|\mathcal{G}_{\min}|^{(1+\sigma_0)/4})$. Under Conditions (C1)-(C4), for any subgroup $\mathcal{G}_k$ ($1\leq k\leq K$) and any $i\in\mathcal{G}_k$, we have 
	$$\Vert\hat\beta_i^{\text{or}}-\beta_i^0\Vert_2=O_p(\phi_{|\mathcal{G}_{k}|}^{1/2}), $$
	where $\phi_{|\mathcal{G}_{k}|}=m\lambda_1^{-1}|\mathcal{G}_k|^{-1}+m^{-2q}+\lambda_1m^{4-2q}+\lambda_1$. 
\end{theorem}

Theorem \ref{thm1} establishes the convergence rate of the oracle estimators of coefficient functions when the true subgroup information is known. Since $|\mathcal{G}_{\min}|=O(n)$ and $m\to \infty$ as $n\to \infty$, then $\lambda_1\to 0$, $\lambda_1m^{4-2q}\to 0$, $m^{-2q}\to 0$, and $m\lambda_1^{-1}|\mathcal{G}_k|^{-1}\to 0$ as $n\to \infty$, and thus $\phi_{|\mathcal{G}_{k}|}\to 0$ as $n\to \infty$. The proof of Theorem \ref{thm1} is provided in the supplementary material. 

Let $$b=\underset{i\in \mathcal{G}_k, j\in \mathcal{G}_{k'}, k\neq k'}{\min}\Vert \beta_i^0-\beta_j^0\Vert_2= \underset{k\neq k'}{\min}\Vert\xi_{k}^0-\xi_{k'}^0\Vert_2$$
be the minimum distance of the common coefficient functions between two different subgroups. Similar to \cite{Ma2017} and \cite{Zhu2018}, the minimum distance $b$ needs to be lower bounded to ensure the model with heterogeneous subgroups is identifiable.

\begin{theorem}\label{thm2}
	Suppose the conditions in Theorem \ref{thm1} and Condition (C5) hold. If $cb>a\lambda_2$ and $\lambda_2\gg\frac{n^{-(1-\sigma_0)/2} p^{4}}{|\mathcal{G}_{\min}|}$ for some constant $c>0$, where $\sigma_0$ is given in Theorem \ref{thm1}. Then, there exists a local minimizer $\bm{\hat\theta}$ of the objective function $\bm{Q_n}(\bm{\theta};\bm{\lambda})$ such that 
	$$\Pr(\bm{\hat\theta}=\bm{\hat\theta}^{\text{or}})\to 1,$$
	and 
	$$\Pr(\bm{\hat\beta}=\bm{\hat\beta}^{\text{or}}) \to 1,$$
	where $\bm{\hat\beta}=\bm{\mathcal{B}}\bm{\hat\theta}$, and $\hat{\bm{\beta}}^{\text{or}}=\bm{\mathcal{B}}\bm{\hat\theta}^{\text{or}}$. 
\end{theorem}

Theorem \ref{thm2} shows that the oracle estimator $\bm{\hat\theta}^{\text{or}}$ is a strictly local minimizer $\bm{\hat\theta}$ of the objective function $Q_n(\bm{\theta};\bm{\lambda})$ with a high probability, and accordingly, the estimator of coefficient functions $\bm{\hat\beta}$ is the oracle estimator of coefficient functions $\bm{\hat\beta}^{\text{or}}$. Let $\hat{\bm\xi}$ be the distinct values of $\bm{\hat\beta}$ and $\bm{\hat\xi}^{\text{or}}$ be the distinct values of $\bm{\hat{\beta}}^{\text{or}}$. By the oracle property in Theorem \ref{thm2}, we have $\Pr(\bm{\hat\xi}=\bm{\hat\xi}^{\text{or}})\to1$. 
Consequently, the true subgroups can be recovered with the estimated common coefficient function for subgroup $k$ given as $\hat{\xi}_k={\hat\beta}_i^{\text{or}}$ for $i\in \mathcal{G}_k$. This result holds given that $b\gg \frac{n^{-(1-\sigma_0)/2} p^{4}}{|\mathcal{G}_{\min}|}$. Under Condition (C4), we require $b\gg C^{*}n^{-\frac{3-\sigma_0}{2}}p^{4}$ for the constant $0<C^*<\infty$. That is, the L2-distance of the true coefficient functions  between subgroups should not be too small, otherwise, the subgroups can not be recovered. The proof of Theorem \ref{thm2} is given in the supplementary material. With Theorems \ref{thm1} and \ref{thm2}, we have established that $\bm{\hat\beta}$ converges to $\bm{\beta}^{0}$ in probability. 

\section{Numerical simulation}\label{sec3}
In this section, we conduct two simulation studies to illustrate the finite-sample performance of the proposed approach. In the first study, we consider two subgroups, and in the second study, there are three subgroups. In our simulation, we let $\mathcal{X}=[0,1]$ and weights $\omega_{ij}=1$. We fix the B-spline with order $q=4$ and the number of knots $m=8$ for all subjects. Following the suggestion of \cite{Zhu2018}, we fix $\delta=2$ and $\tau=1$. 

Besides the proposed approach, we also consider the following alternatives for comparison: (i) Oracle approach (referred as "Oracle"), under which the subgrouping structure is known, and a penalized B-spline estimator is adopted to estimate the coefficient function in each subgroup.  The objective function of the penalized B-spline is 
\[Q_h(\bm\theta;\lambda)=\frac{1}{n}\sum_{i=1}^n(y_i-\bm H_i\bm\theta)^2+\lambda\Vert D^2(\bm B\bm\theta)\Vert_2^2\]
where $\bm\theta$ is the common $p$-dimensional B-spline coefficient vector for all subjects within the same subgroup. (ii) We consider the finite mixture regression (FIMR) approach of \cite{yao12}, which first treats the heterogeneous coefficient functions $\beta_i$'s as $p$-dimensional vectors $\bm\theta_i$'s, and then applies FIMR to the heterogeneous linear regression model: $y_i=\bm H_i\bm \theta_i+\epsilon_i$. 
It is realized by using the R package \emph{flexmix}. (iii) We consider two Kmeans-based clustering methods. First, we consider Response-based clustering (referred to as ``Resp"), which first clusters the subjects based on response variables via Kmeans, and then applies a penalized B-spline to each subgroup. Second, we consider Residual-based clustering (referred as ``Resi"), which applies penalized B-spline under the homogeneity assumption $y_i=\int_0^1X_i(t)\beta(t)\text{d}t+\epsilon_i$, clusters subjects based on the residuals $y_i-\int_0^1X_i(t)\hat\beta(t)\text{d}t$ by using Kmeans, and then applies the penalized B-spline estimator again to each subgroup. 
The FIMR and two Kmeans-based approaches need to determine the number of subgroups. We set the true number of subgroups as the true value. We note that this is not required with the proposed approach and is often not practical in data analysis. For Resp and Resi approaches, the tuning parameter $\lambda$ is determined by using a grid search with $\lambda\in\{0.0001, 0.001, 0.005, 0.01, 0.025, 0.05, 0.1, 0.5, 1, 5\}$ and GCV criterion.

To assess  the subgrouping performance, we consider two measures, namely the number of identified subgroups ($\hat{K}$) and the Adjusted Rand Index (ARI), where the Adjusted Rand Index
measures the agreement between the structure of the estimated subgroups and that of the true subgroups. To assess the estimation performance, we consider the mean squared error (MSE) of the estimated coefficient functions, which is defined as $\sum_{i=1}^{n}\int_0^1(\beta_i(t)-\hat\beta_i(t))^2\text{d}t/n$. The results are based on 100 independent runs. 

\subsection{Example 1: two subgroups} 	
We simulate the heterogeneous data with $n$ samples where $n=40$ and $200$. The $n$ subjects belong to two subgroups, and two subgroup structures are considered: (a) balanced (referred as ``B"), where the two subgroups have the same sizes; and (b) unbalanced (referred as ``UB"), where the two subgroups have sizes in a ratio of 1:3. The continuous response $y_i$ for subject $i$ from the $k$th subgroup is generated by 
\begin{equation}\label{eq:sim}
	y_i=\int_0^1X_i(t)\xi_k(t)\text{d}t+\epsilon_i,
\end{equation}
where the error terms $\epsilon_i$ are independently generated from the normal distribution with mean 0 and variance 1. Each prediction function $X_i(t)$ is generated as a linear combination of B-spline basis functions. Specifically, $X_i(t)=\sum_{l=1}^{\tilde p}a_{il}\tilde B_l(t)$, where $\tilde p=\tilde m+\tilde q$ with $\tilde m=15$ and $\tilde q=5$. These $a_{il}$'s are independently generated from the following distributions: (1) Norm: $N(2,1)$ and (2) Unif: $U(0,4)$. $\xi_k(t)$ represents the coefficient function of the $k$th subgroup, for which we consider the following two scenarios: \\
\noindent (S1) Nonlinear: $\xi_1(t)=4\sin(\pi t)-1$ and $\xi_2(t)=10(t-0.5)^2-2$. The distance between two functions in $L^2[0,1]$ is 3.36, \\
\noindent (S2) Linear: $\xi_1(t)=3t+2$ and $\xi_2(t)=3t-2$. The distance between two functions in $L^2[0,1]$ is 4. \\

Figure \ref{fig:ex1} and Figure C1 in the supplementary material show the distribution of the number of identified subgroups $\hat{K}$ for Scenarios 1 and 2, respectively. The proposed approach can satisfactorily identify the number of true subgroups for all settings. As $n$ increases, the percentage of correctly determining the number of subgroups becomes larger. Subgrouping and estimation results are summarized in Table \ref{tab: table1} for Scenario 1 and Table C1 in the supplementary material for Scenario 2. Since the Oracle approach assumes the true subgroup structure, the proposed approach behaves inferior to the Oracle as expected. But it has significant advantages over the FIMR, Resp and Resi approaches. Consider for example Scenario 1, $n=40$, balanced structure, and $a_{il}$'s are generated from the normal distribution $N(2,1)$. The proposed approach has (ARI, MSE) equal to (0.960, 1.679), compared to (0.889, 41.156) for the FIMR approach, (0.935, 1.934) for the Resp approach, and (0.861, 2.136) for the Resi approach. We observe that, for 100 independent replicates, FIMR either can accurately recover the subgroup structure with ARI$>0.95$, or completely fail with ARI$\approx 0$, leading to a low average ARI and large standard error of ARI. Meanwhile, once FIMR fails in identifying subgroup structure, a tremendous MSE ($100\sim 400$) of the estimated coefficient functions is generated. Note that the true number of subgroups is known as \emph{a prior} for the FIMR, Resi and Resp approaches, whereas it is unknown for the proposed approach. The values of ARI and MSE increase as $n$ grows for the proposed approach. This supports the estimation  consistency established in Theorem \ref{thm1} and \ref{thm2}. To better visualize the estimation results of the proposed approach, we present estimators of the coefficient functions in Figure \ref{fig:coef1} for Scenario 1 and in Figure C2 in the supplementary material for Scenario 2. It can be observed that the underlying true coefficient function curves can be recovered well. 

\begin{figure}[H]
	\centering  
	\subfigure[]
	{
		\includegraphics[width=2.7cm]{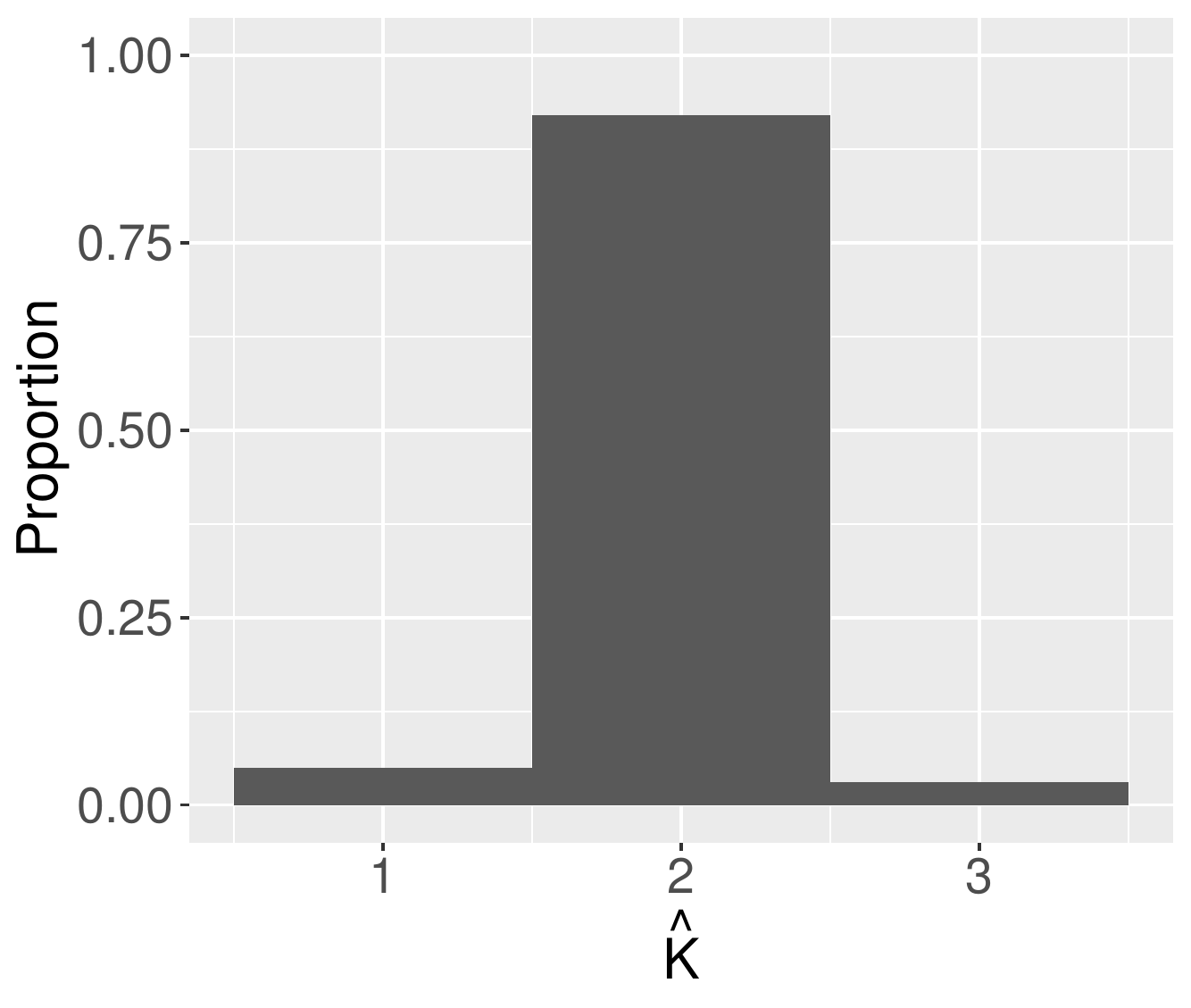}
	}
	\subfigure[]
	{
		\includegraphics[width=2.7cm]{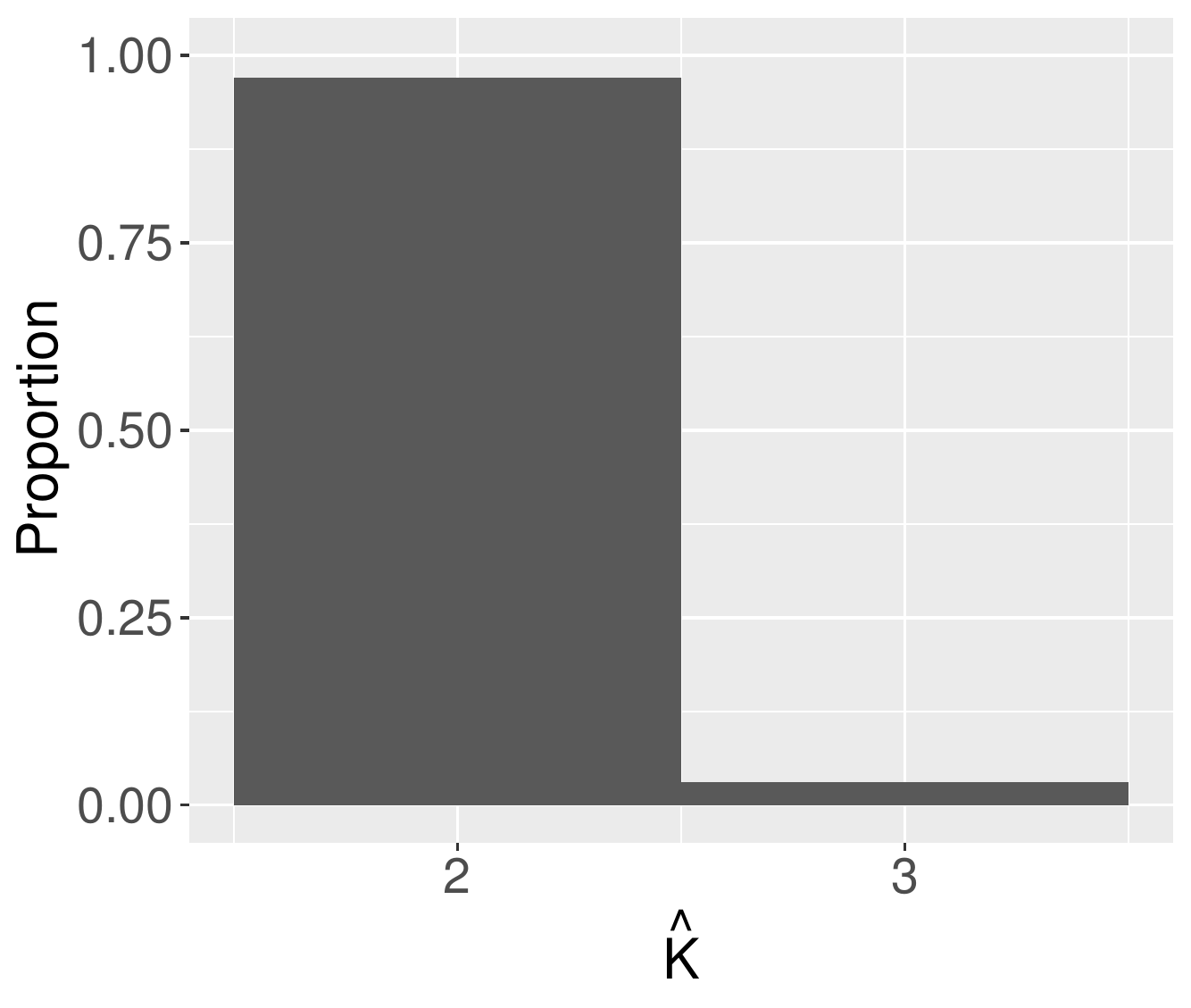} 
	}
	\subfigure[]
	{
		\includegraphics[width=2.7cm]{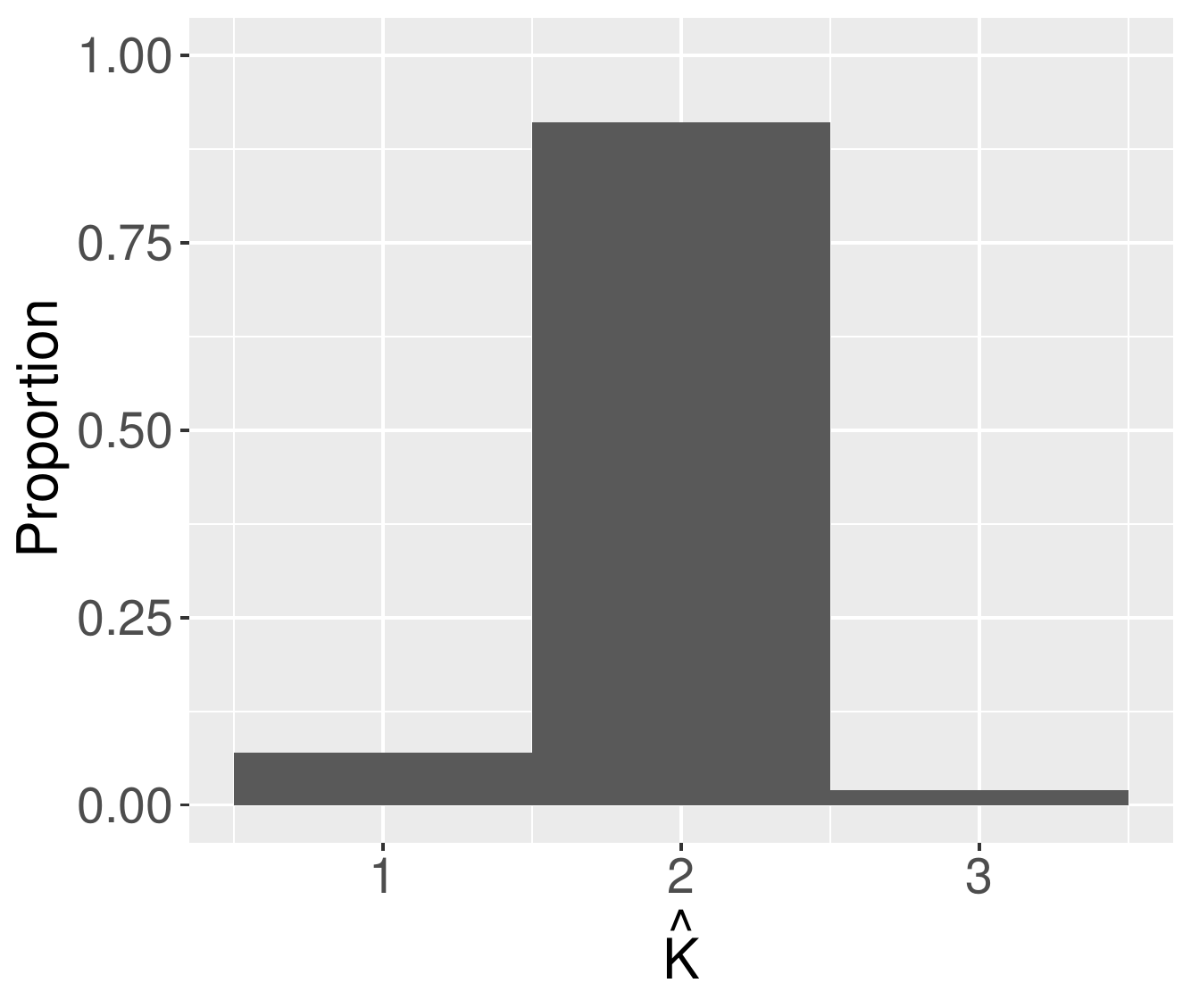}
	}
	\subfigure[]
	{
		\includegraphics[width=2.7cm]{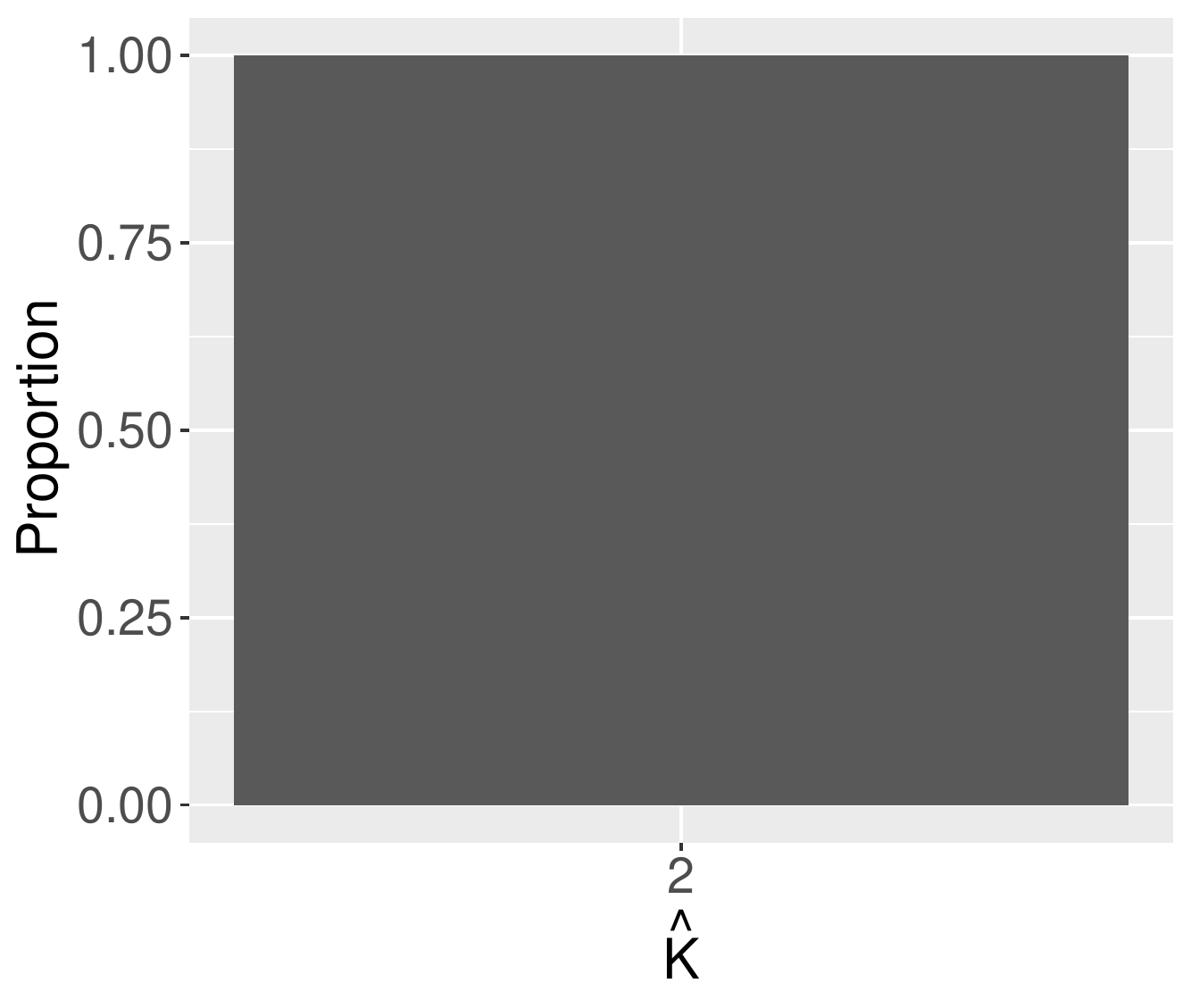} 
	}
	\subfigure[]
	{
		\includegraphics[width=2.7cm]{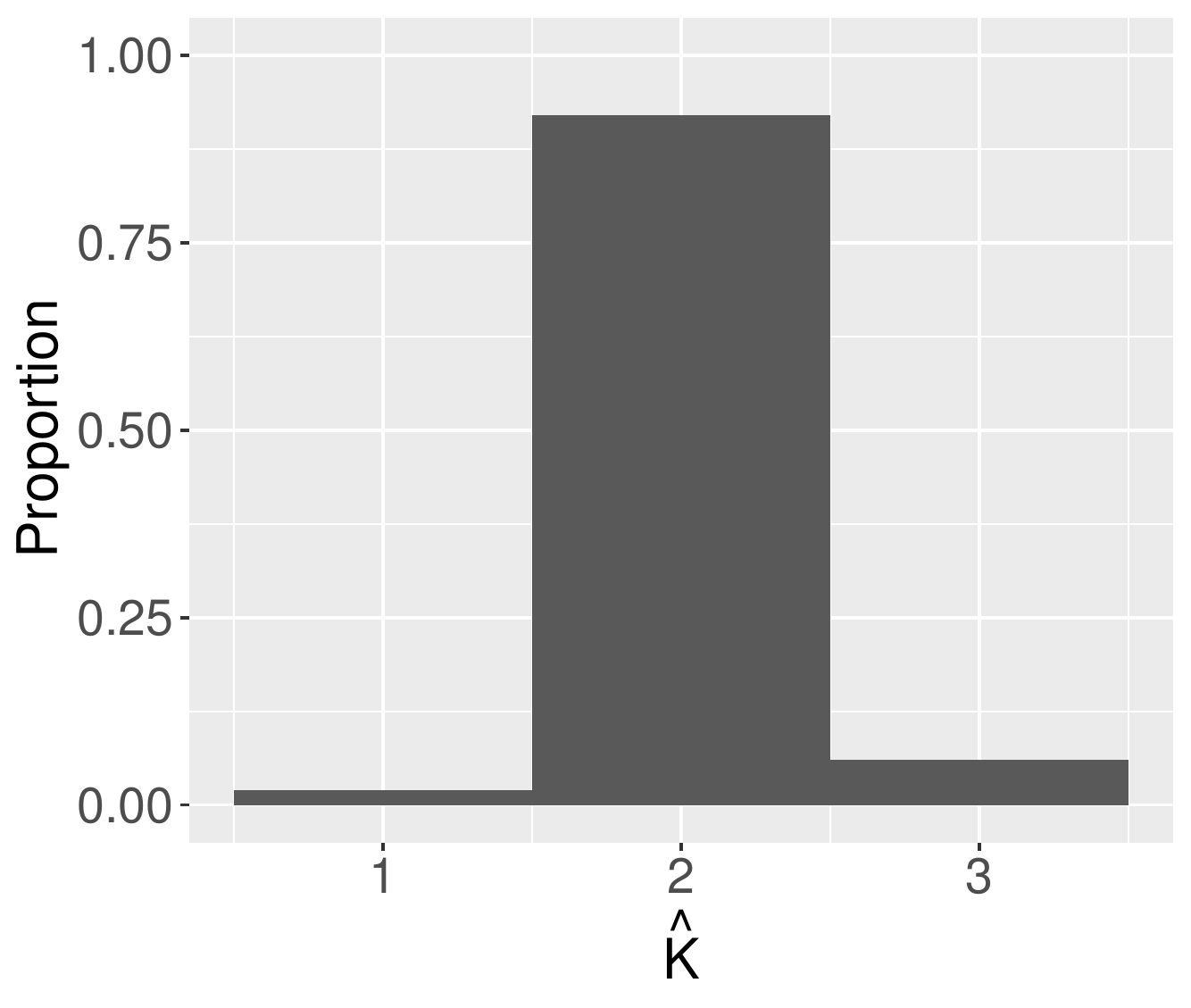}
	}
	\subfigure[]
	{
		\includegraphics[width=2.7cm]{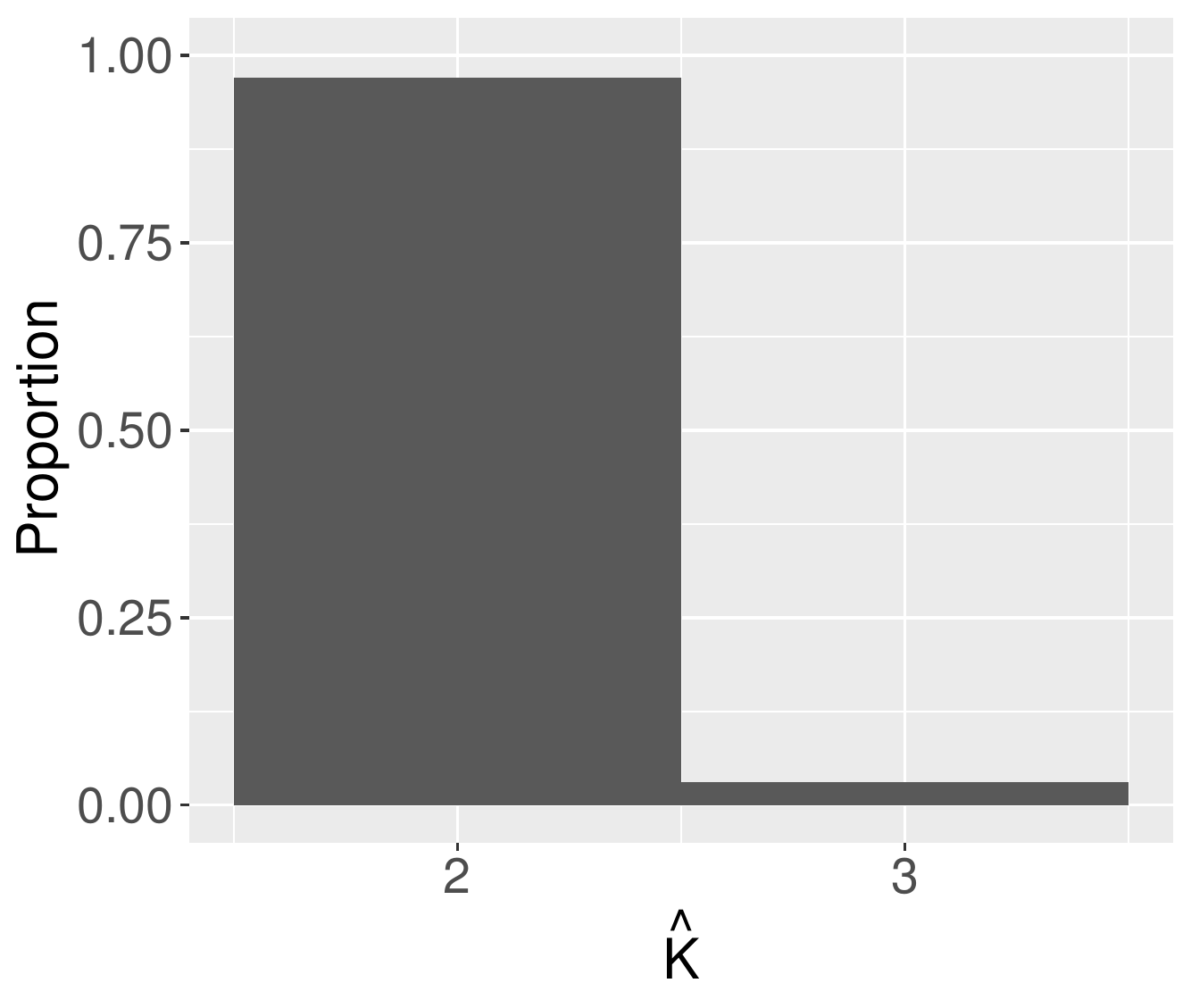} 
	}
	\subfigure[]
	{
		\includegraphics[width=2.7cm]{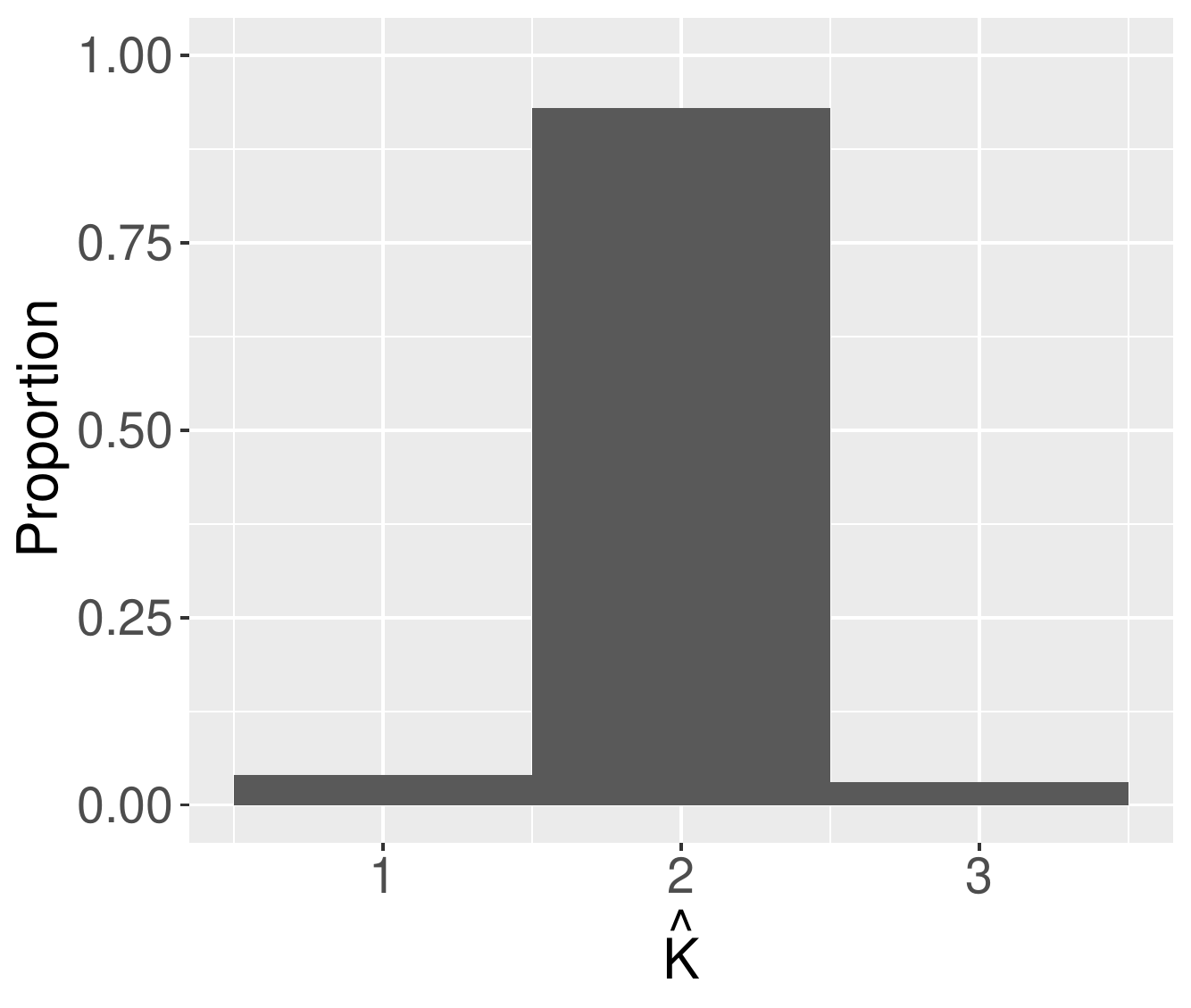}
	}
	\subfigure[]
	{
		\includegraphics[width=2.7cm]{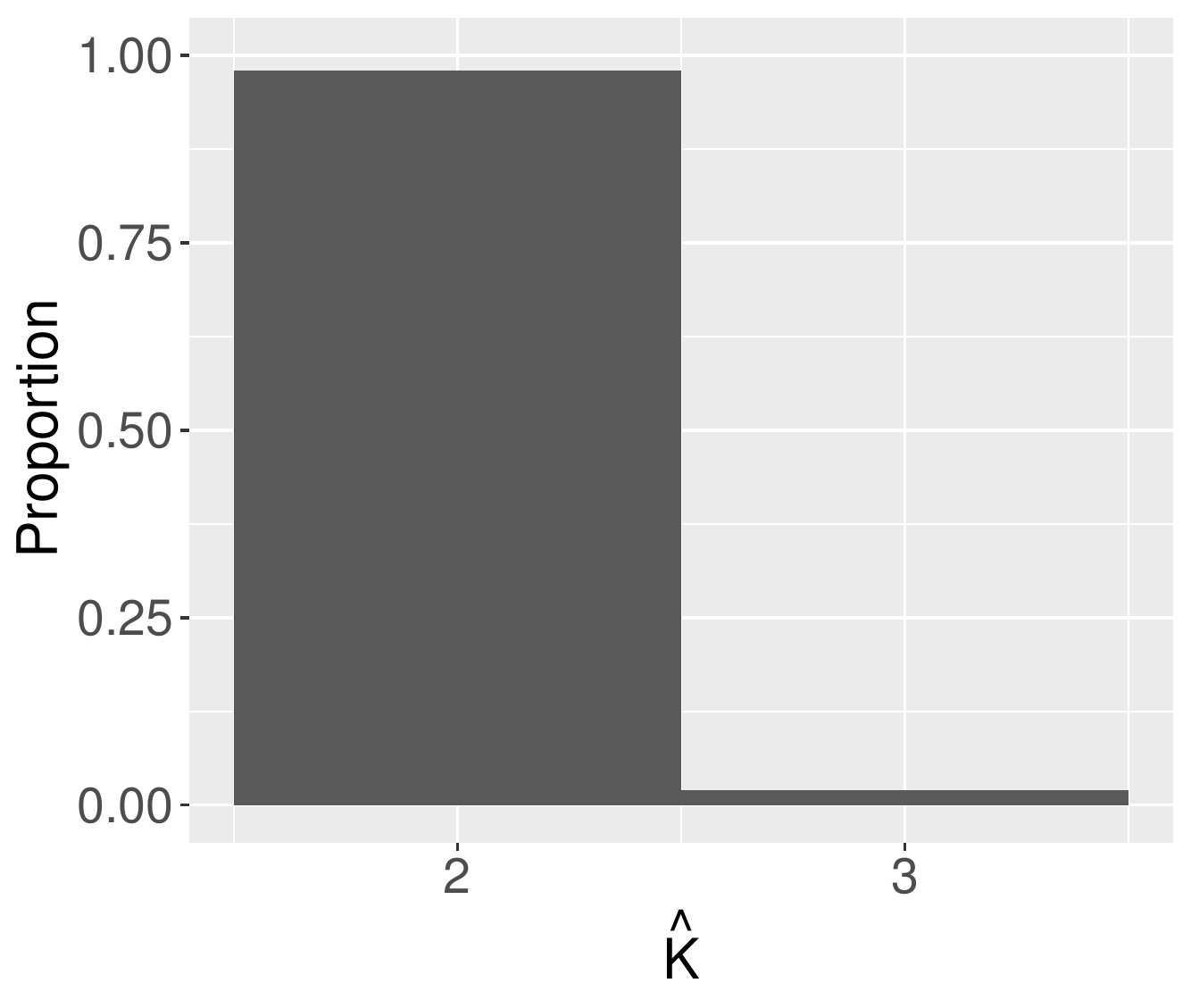} 
	}	
	\caption{Simulation results in Example 1: histograms of the estimated number of subgroups $\hat{K}$ by the proposed approach under Scenario 1. (a) and (e) balanced structure with $n=40$, (b) and (f) balanced structure with $n=200$, (c) and (g) unbalanced structure with $n=40$, (d) and (h) unbalanced structure with $n=200$. (a)-(d) $a_{il}\sim N(2,1)$, and (b)-(h) $a_{il}\sim U(0,4)$.} 
	\label{fig:ex1}
\end{figure}

\begin{table}[]
	\centering
	\caption{Simulation results in Example 1: ARI and MSE under Scenario 1. Each cell shows the mean(s.d.).}
	\resizebox{\linewidth}{!}{\begin{tabular}{ccc|cc|cc}
			\\
			\hline
			\multicolumn{3}{c|}{\multirow{2}{*}{}}            & \multicolumn{2}{c|}{$n=40$} & \multicolumn{2}{c}{$n=200$} \\ \cline{4-7} 
			Structure & $a_{il}$ & Method            & ARI                            & MSE            & ARI                             & MSE         \\ \hline
			B & Norm&    Proposed &   0.960(0.002)        & 1.679(0.105)                    & 0.993(0.001)    & 0.392(0.026)                 \\
			&                       & Oracle  &      & 1.348(0.083)                     &                    &   0.327(0.022)             \\
			&                       & FIMR &  0.889(0.037)                                      &        41.156(11.154)              &    0.933(0.006)     & 8.217(1.838)\\
			&                       & Resp   & 0.935(0.007)       & 1.934(0.187)                   & 0.935(0.005)      & 0.733(0.039)                   \\
			&                       & Resi  & 0.861(0.013)       & 2.136(0.248)                  & 0.987(0.002)      & 0.499(0.047)                  \\
			
			& Unif  & Proposed  & 0.964(0.005)    & 1.369(0.121)                     & 0.992(0.002)      & 0.309(0.014)                    \\
			&                       & FIMR &    0.853(0.045)      &                               42.294(10.365) &   0.921(0.008)         & 6.248(1.530)\\
			&                       & Oracle     &           & 1.008(0.089)                    &             & 0.255(0.016)                                \\ 
			&                       & Resp  & 0.916(0.008)      & 1.531(0.123)               & 0.926(0.004)   & 1.123(0.053)               \\
			&                       & Resi  & 0.845(0.015)  & 1.703(0.149)                 & 0.987(0.002)    & 0.408(0.021)                      \\\hline
			
			UB &Norm & Proposed  & 0.954(0.004)    & 1.337(0.097)                       & 0.981(0.002)          & 0.477(0.025)                 \\
			&                       & Oracle   & & 1.187(0.080)             & & 0.353(0.019)         \\
			&                       & FIMR &  0.908(0.037)      & 12.426(2.170)                               &    0.975(0.001)        & 2.094(0.175) \\
			&                       & Resp  & 0.935(0.005)    & 1.319(0.139)                       & 0.969(0.002)     & 0.618(0.035)                    \\
			&                       & Resi    & 0.924(0.010)   & 2.011(0.363)                  & 0.977(0.004)          & 0.492(0.023)                    \\
			
			& Unif & Proposed & 0.962(0.004)     & 1.278(0.136)                    & 0.995(0.001)             & 0.374(0.029)                  \\
			&                       & Oracle   & &1.018(0.085)&  &0.279(0.014)       \\ 
			&                       & FIMR &   0.913(0.041)          &                               9.371(2.794) & 0.985(0.002)                   & 3.754(0.159) \\
			&                       & Resp  & 0.931(0.007)  & 1.286(0.129)                          & 0.952(0.004)     & 0.594(0.052)                   \\
			&                       & Resi    & 0.927(0.007)  & 1.493(0.214)                       & 0.969(0.002)          & 0.476(0.028)                    \\ \hline
			
			\label{tab: table1}
	\end{tabular}}
\end{table}

\begin{figure}[H]
	\centering    
	\subfigure[]
	{
		\includegraphics[width=3.8cm]{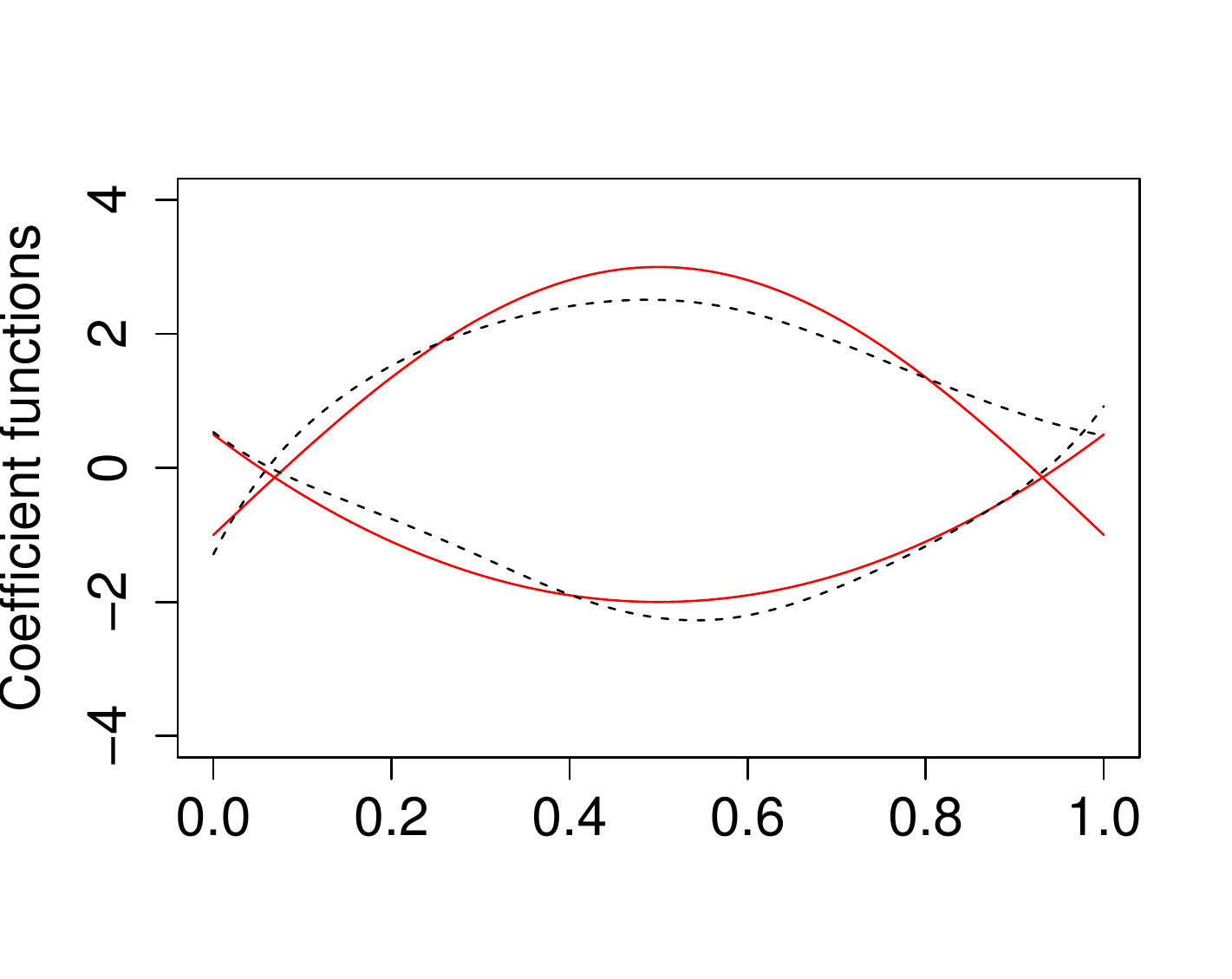}
	}
	\subfigure[]
	{
			\includegraphics[width=3.8cm]{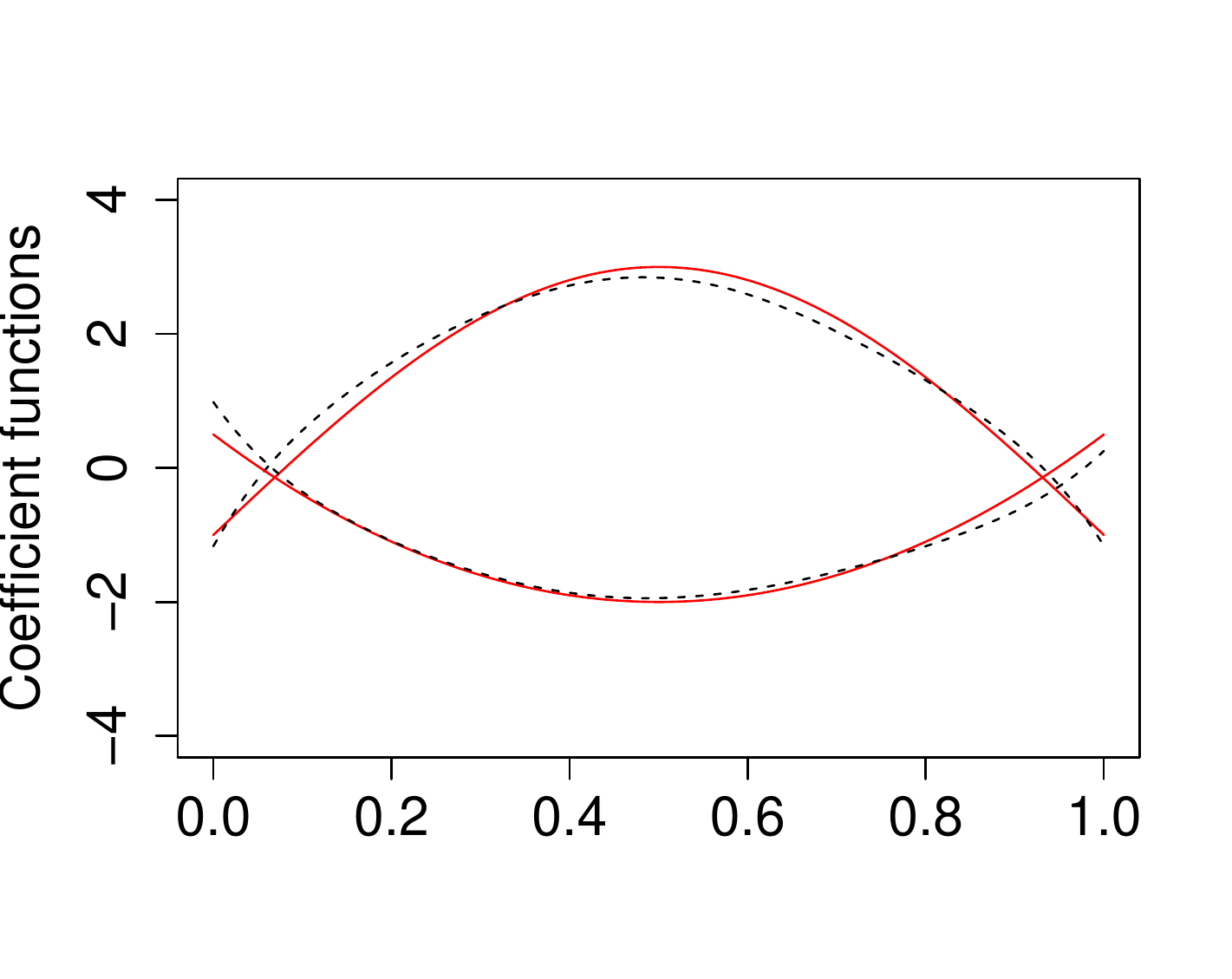} 
		}
		\subfigure[]
		{
			\includegraphics[width=3.8cm]{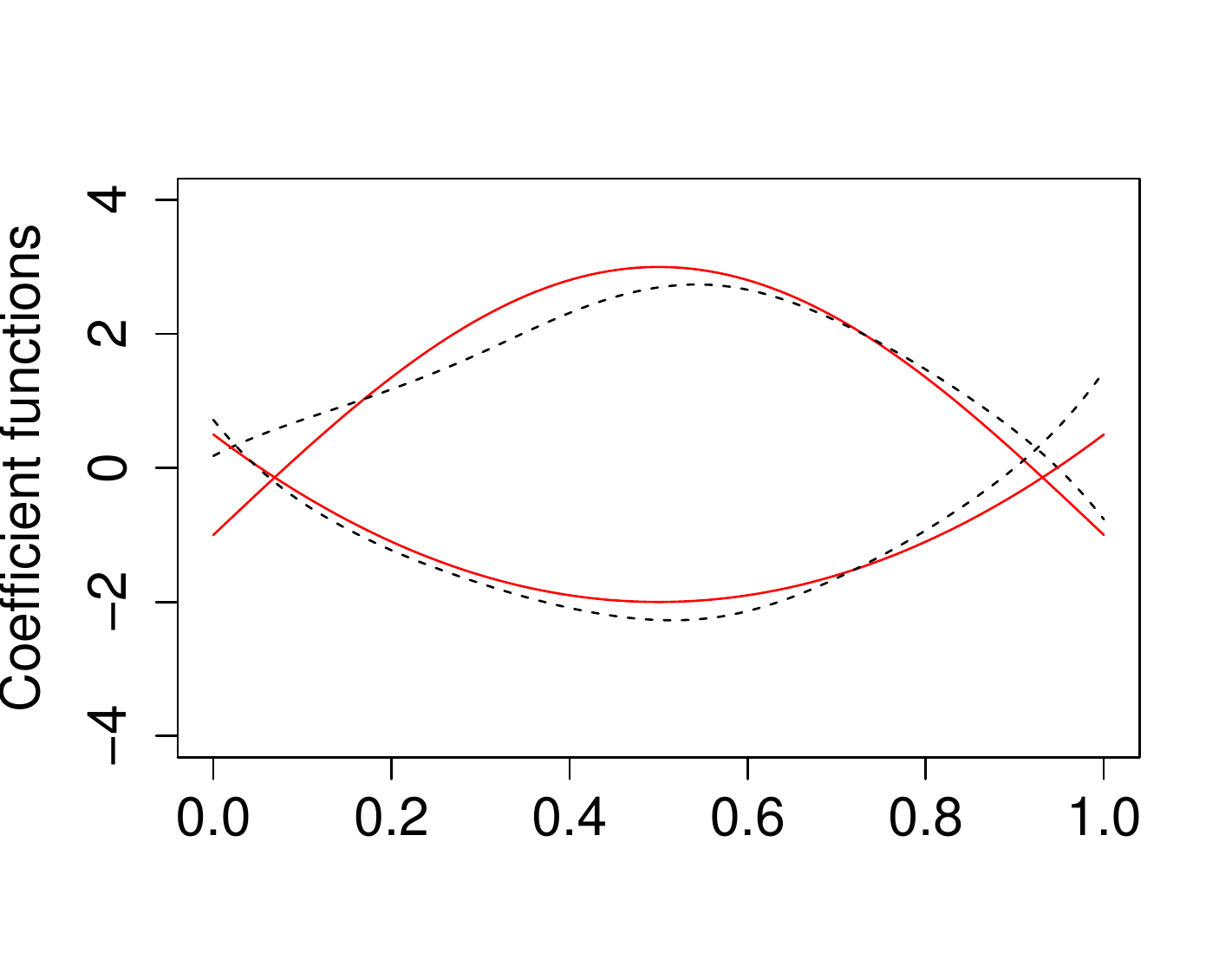}
		}
		\subfigure[]
		{
			\includegraphics[width=3.8cm]{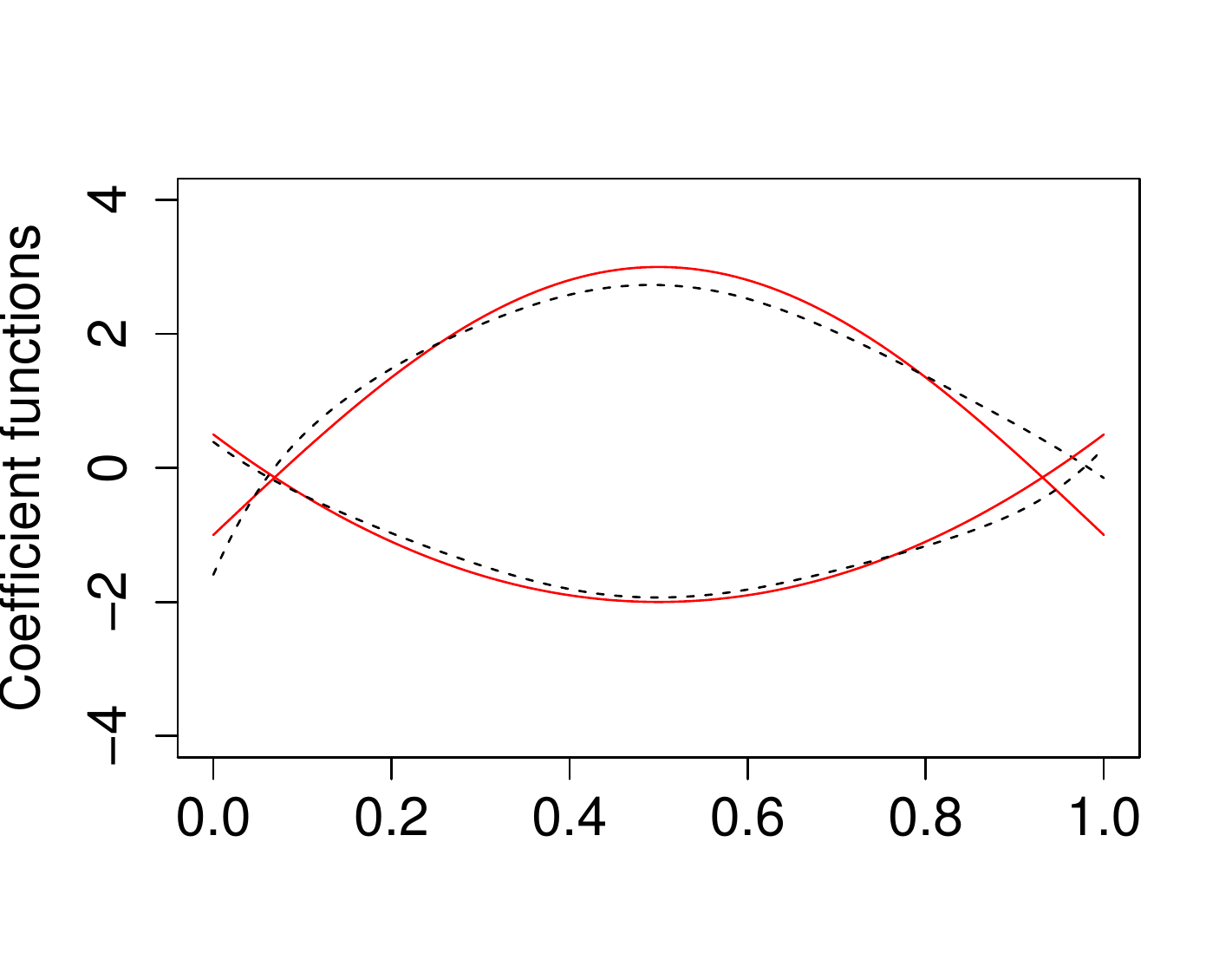} 
		}
		\subfigure[]
		{
			\includegraphics[width=3.8cm]{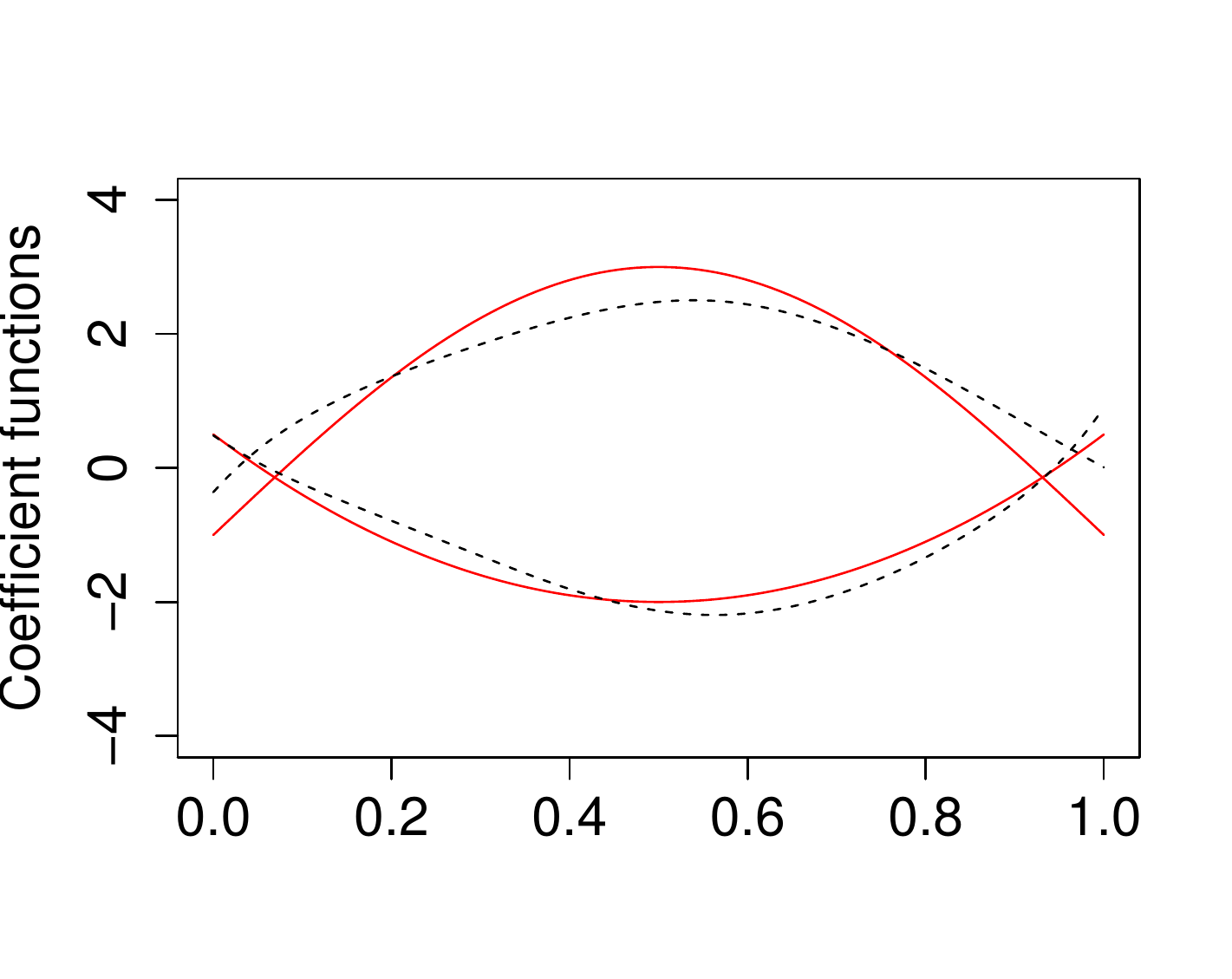}
		}
		\subfigure[]
		{
			\includegraphics[width=3.8cm]{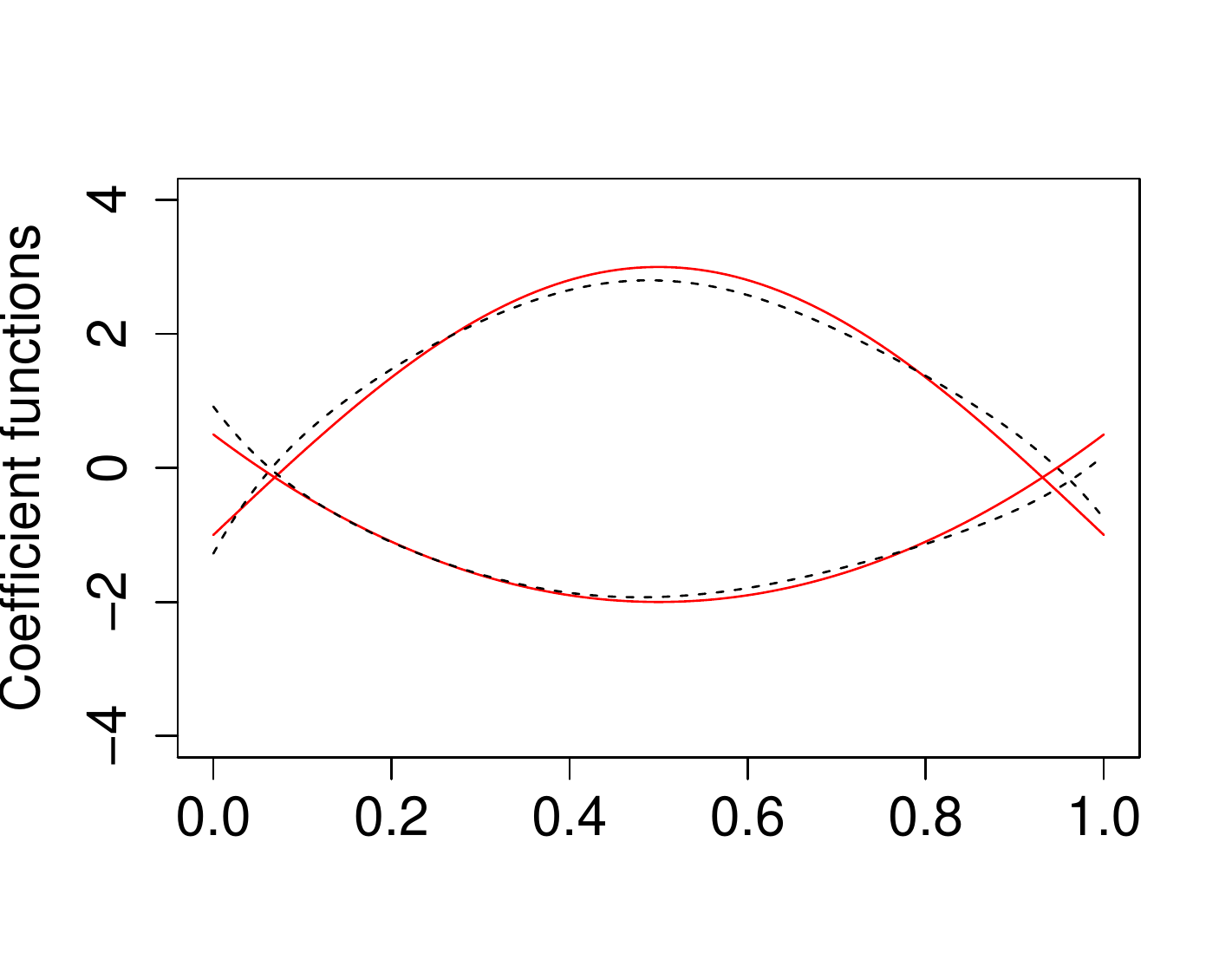} 
		}
		\subfigure[] 
		{
			\includegraphics[width=3.8cm]{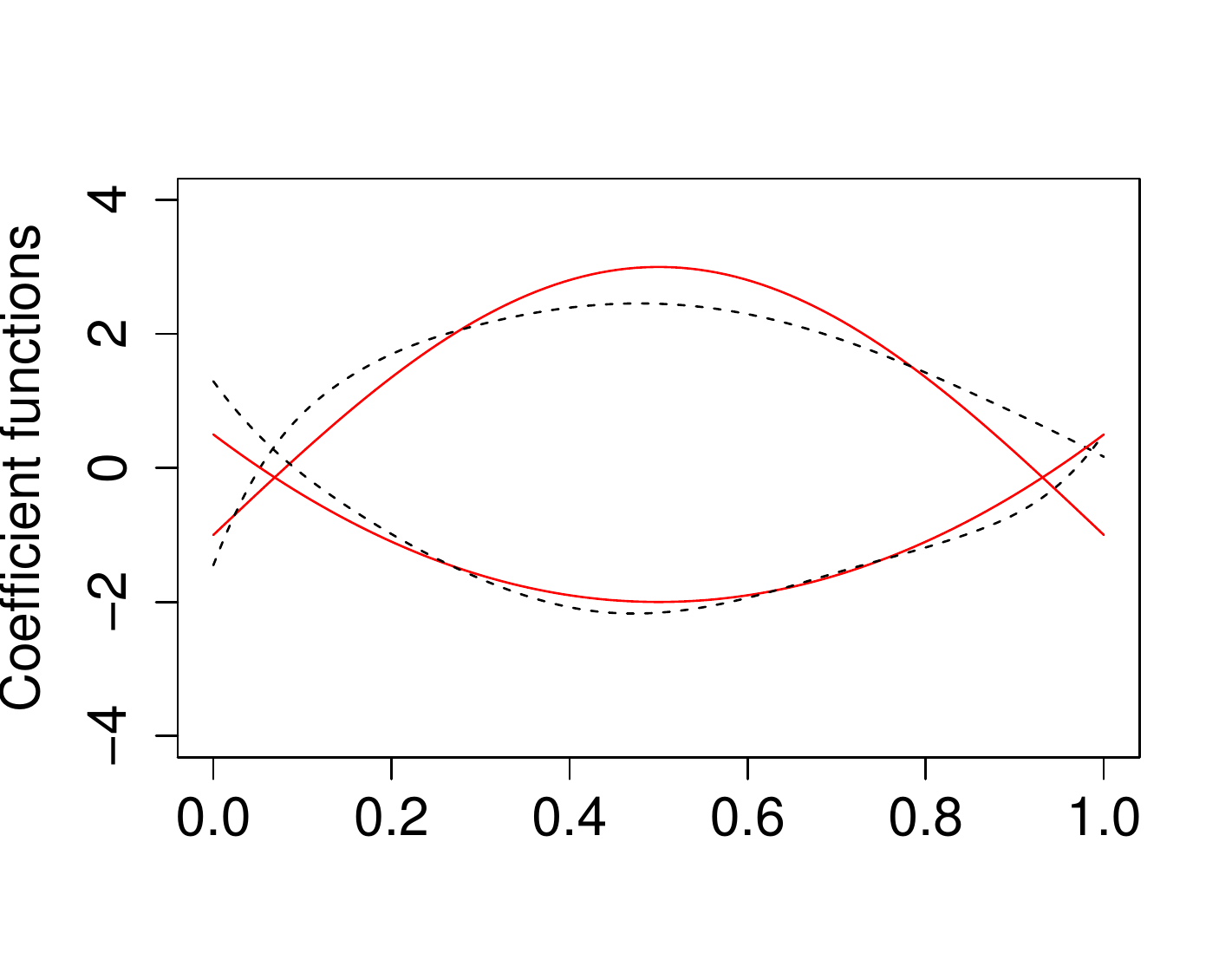}
		}
		\subfigure[]
		{
			\includegraphics[width=3.8cm]{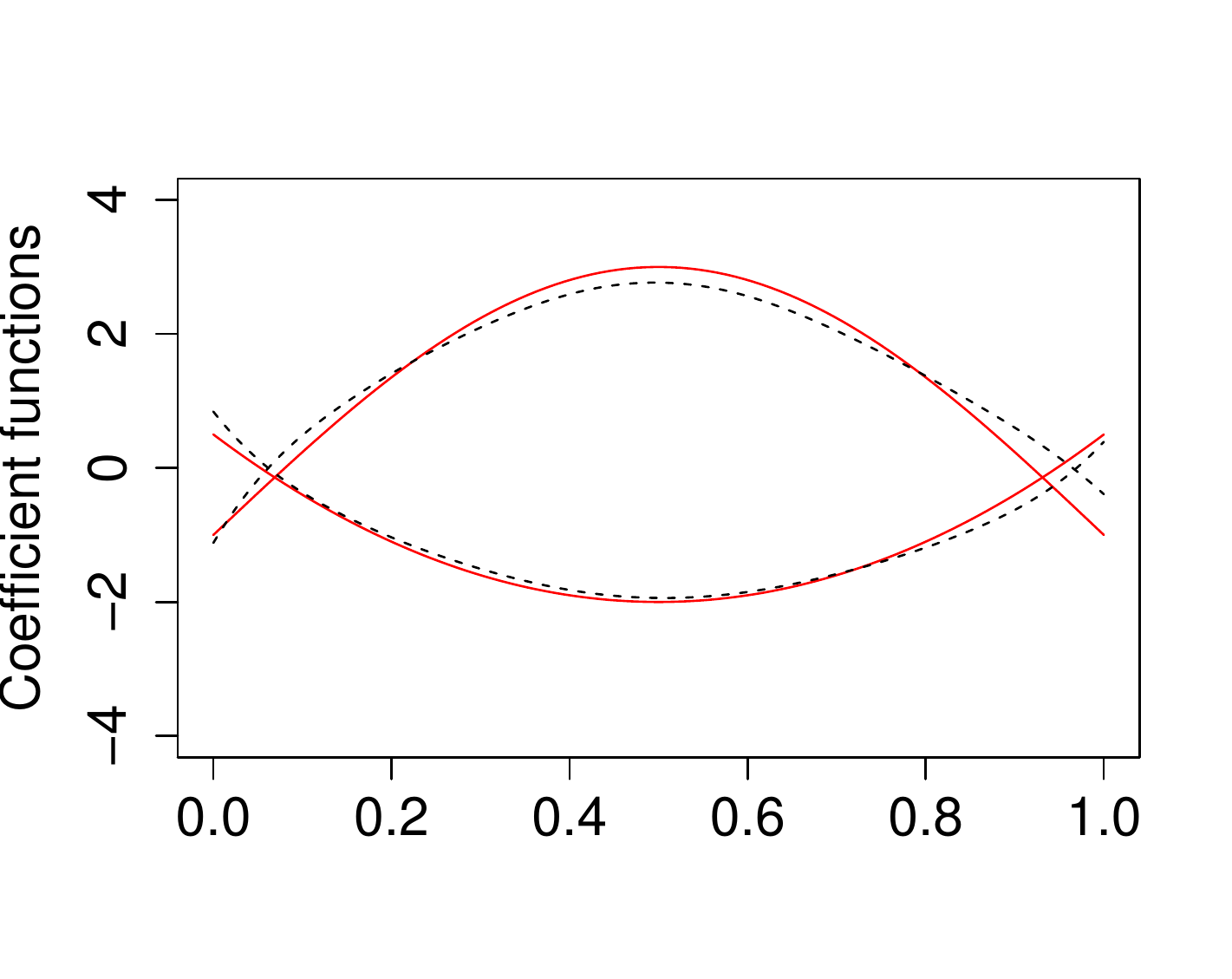} 
		}	
		\caption{Simulation results in Example 1: estimation of the coefficient functions by the proposed approach under Scenario 1. (a) and (e) balanced structure with $n=40$, (b) and (f) balanced structure with $n=100$, (c) and (g) unbalanced structure with $n=40$, (d) and (h) unbalanced structure with $n=200$. (a)-(d) $a_{il}\sim N(2,1)$, and (b)-(h) $a_{il}\sim U(0,4)$. Red solid lines represent the true coefficient functions, and black dashed lines represent the estimated coefficient functions.} 
		\label{fig:coef1}
	\end{figure}
	
	\subsection{Example 2: three subgroups} 
	In this simulation study, we investigate the performance of the proposed approach when there are three subgroups. Let the sample size $n=60$ and $300$. Under the balanced structure, three subgroups have equal sizes; and under the unbalanced structure, four subgroups have sizes in a ratio of 2:3:5. For the $k$th subgroup, we simulate data from model (\ref{eq:sim}) with the predictor function and the error terms generated from the same distribution as given in Example 1. The three subgroups have distinct coefficient functions: $\xi_1(t)=-6t+2$, $\xi_2(t)=6t^3+10t^2-10t-3$, and $\xi_3(t)=-e^{2t}+3\sin(2t)+1.5$, of which the minimum distance (measured in $L^2[0,1]$) between them is 1.84. 
	
	Figure C3 in the supplementary material illustrates the distribution of the estimated number of subgroups by the proposed approach. We obtain similar patterns as those in Figure \ref{fig:ex1} and Figure C1 in the supplementary material for Example 1. Again, the proposed approach can accurately identify the true number of subgroups. Table C2 in the supplementary material summarizes the results of other statistics. Similar to Example 1, aside from the Oracle approach, the ARI of the proposed approach is always the highest; therefore, it outperforms alternatives in the accuracy of identifying subgroup memberships. In terms of estimation accuracy, the proposed approach is again observed to have favorable performance. Compared to Example 1 where there are only two subgroups, here, the advantages of the proposed approach in terms of subgrouping and estimation are more significant. Consider for example $n=60$, balanced structure, and $a_{il}$'s are generated from norm distribution $N(2,1)$. The proposed approach has (ARI, MSE)=(0.948, 1.969), which is superior to the alternatives: (0.533, 395.569) for the FIMR approach, (0.834, 4.230) for the Resp approach, and (0.781, 7.304) for the Resi approach. FIMR has the worst performance. Figure \ref{fig:coef3} displays the estimation of coefficient functions by the proposed approach. It can be observed again that the fitting curves produced by the proposed approach are close to the true ones regardless of the simulation settings, indicating that the proposed approach can recover accurately the underlying true coefficient functions. Furthermore, Figure C4 in the supplementary material displays the consumed time (in seconds) running our algorithm with different sample sizes.

	
	\begin{figure}[H]
		\centering    
		\subfigure[]
		{
			\includegraphics[width=3.8cm]{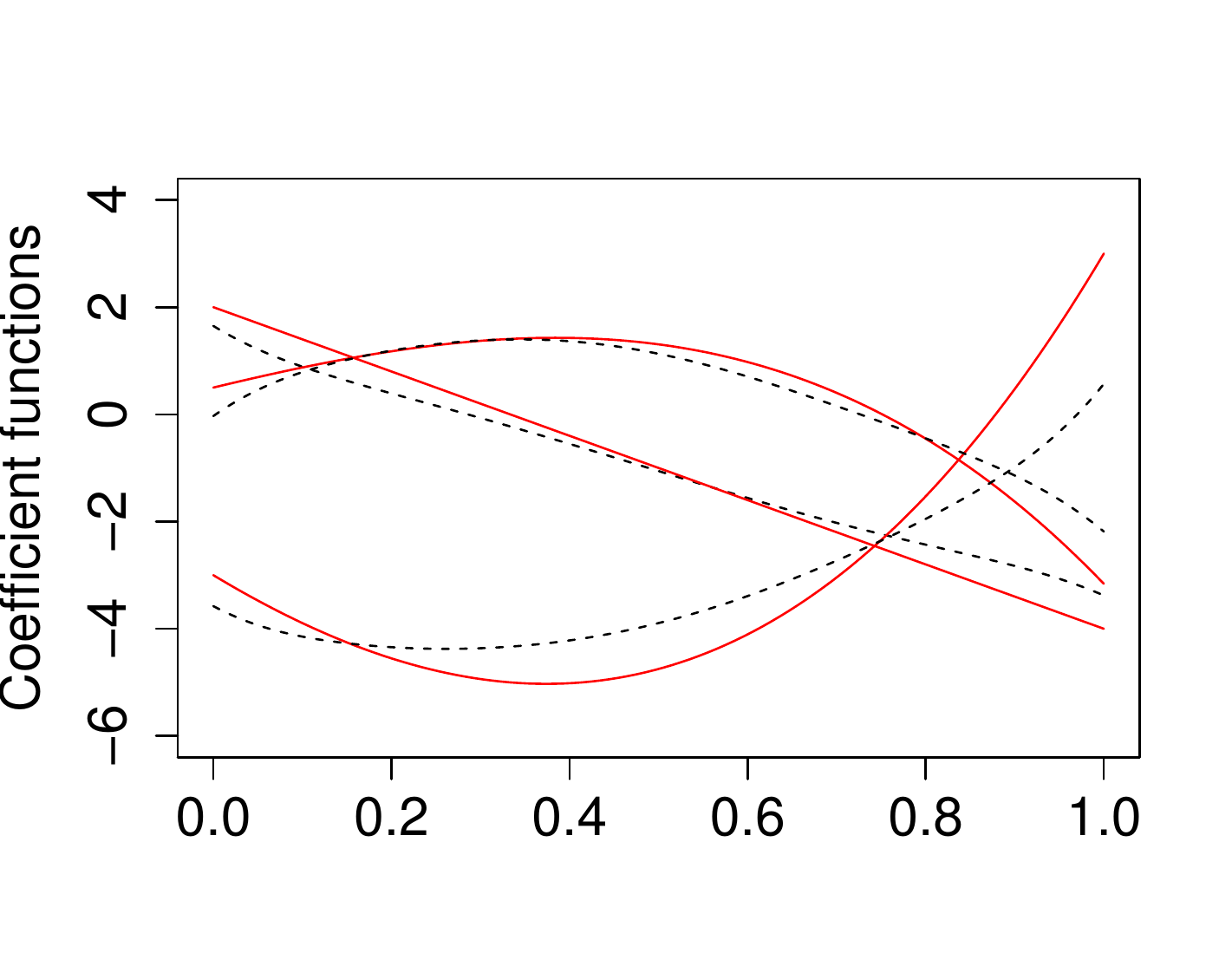}
		}
		\subfigure[]
		{
			\includegraphics[width=3.8cm]{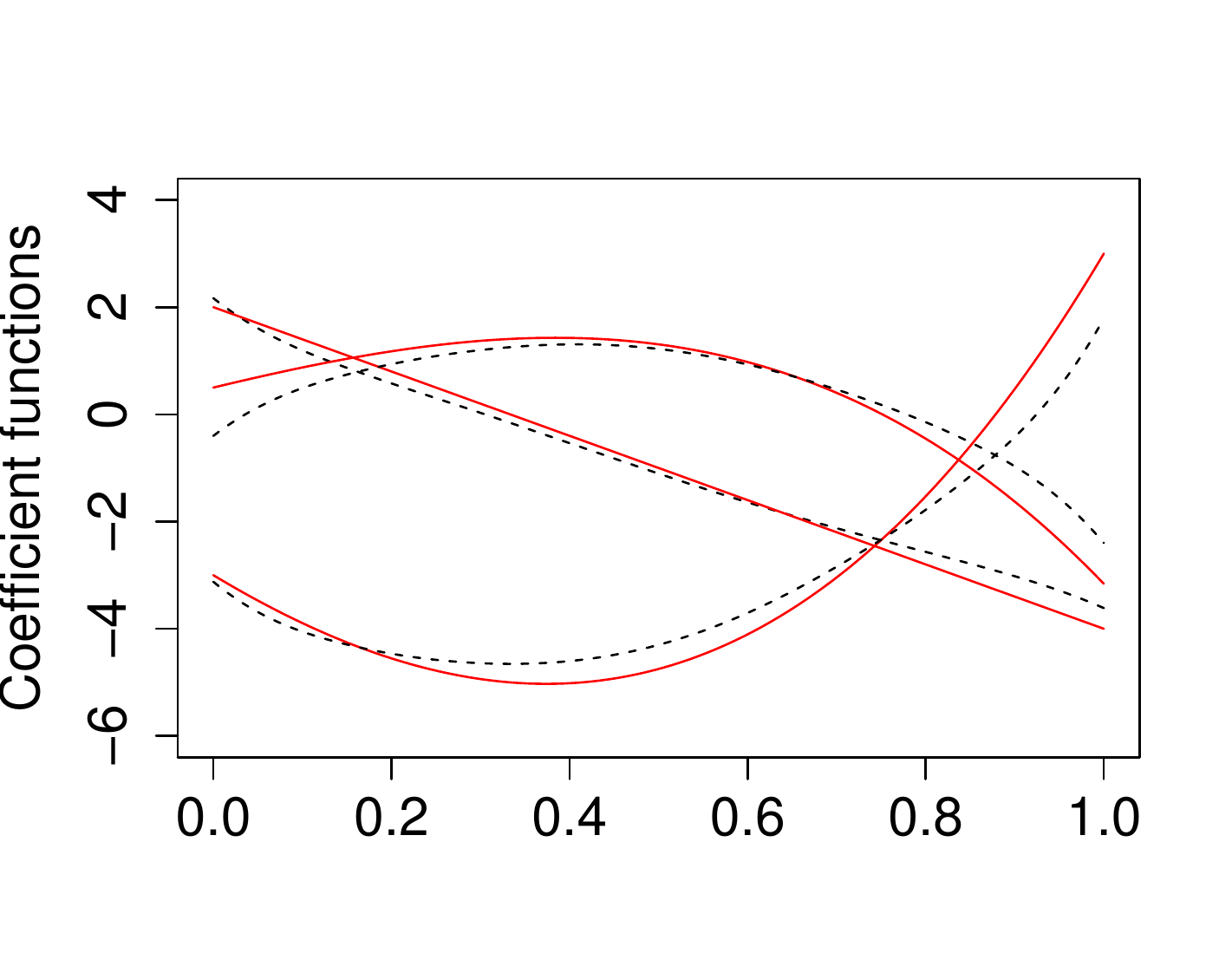} 
		}
		\subfigure[] 
		{
			\includegraphics[width=3.8cm]{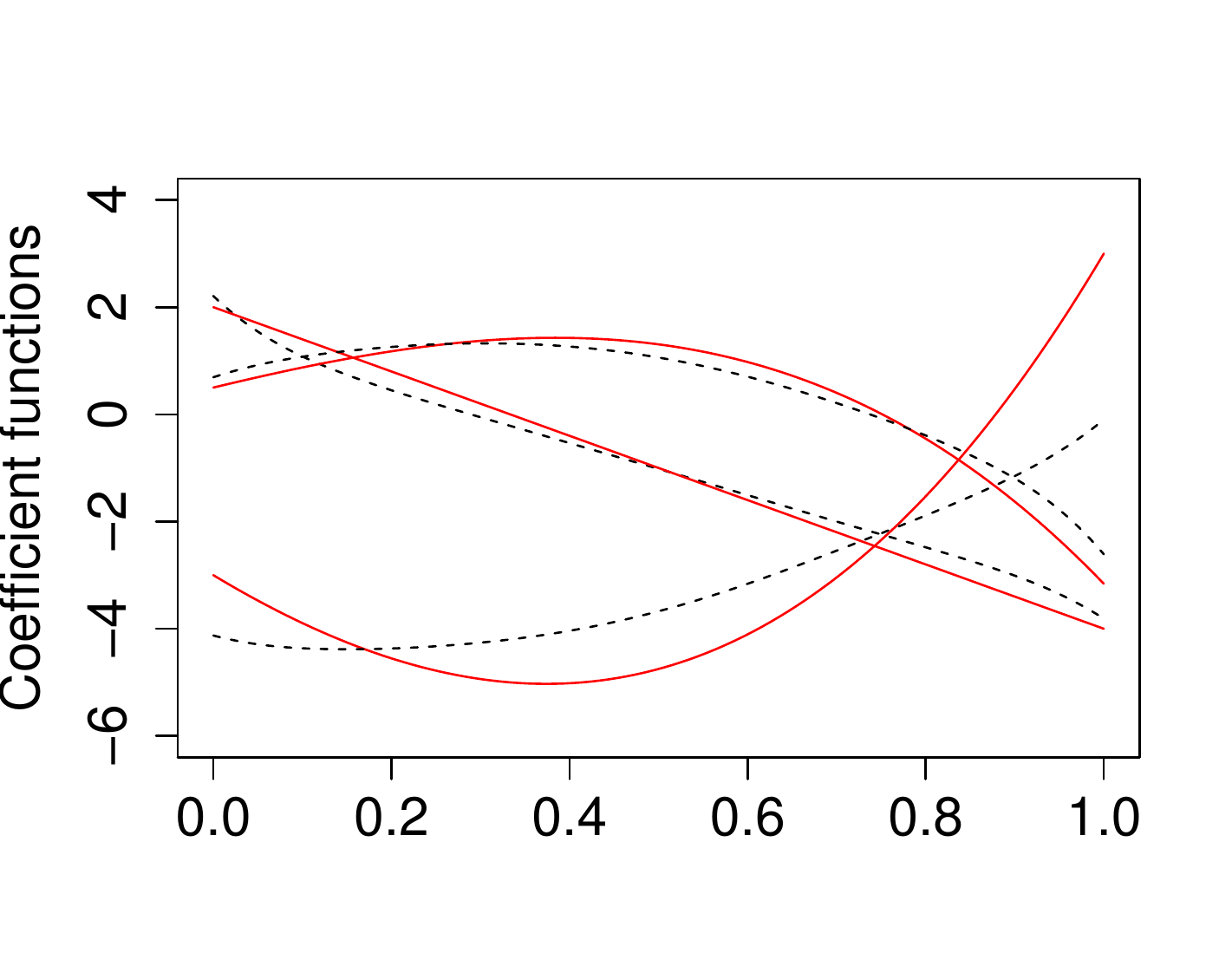}
		}
		\subfigure[]
		{
			\includegraphics[width=3.8cm]{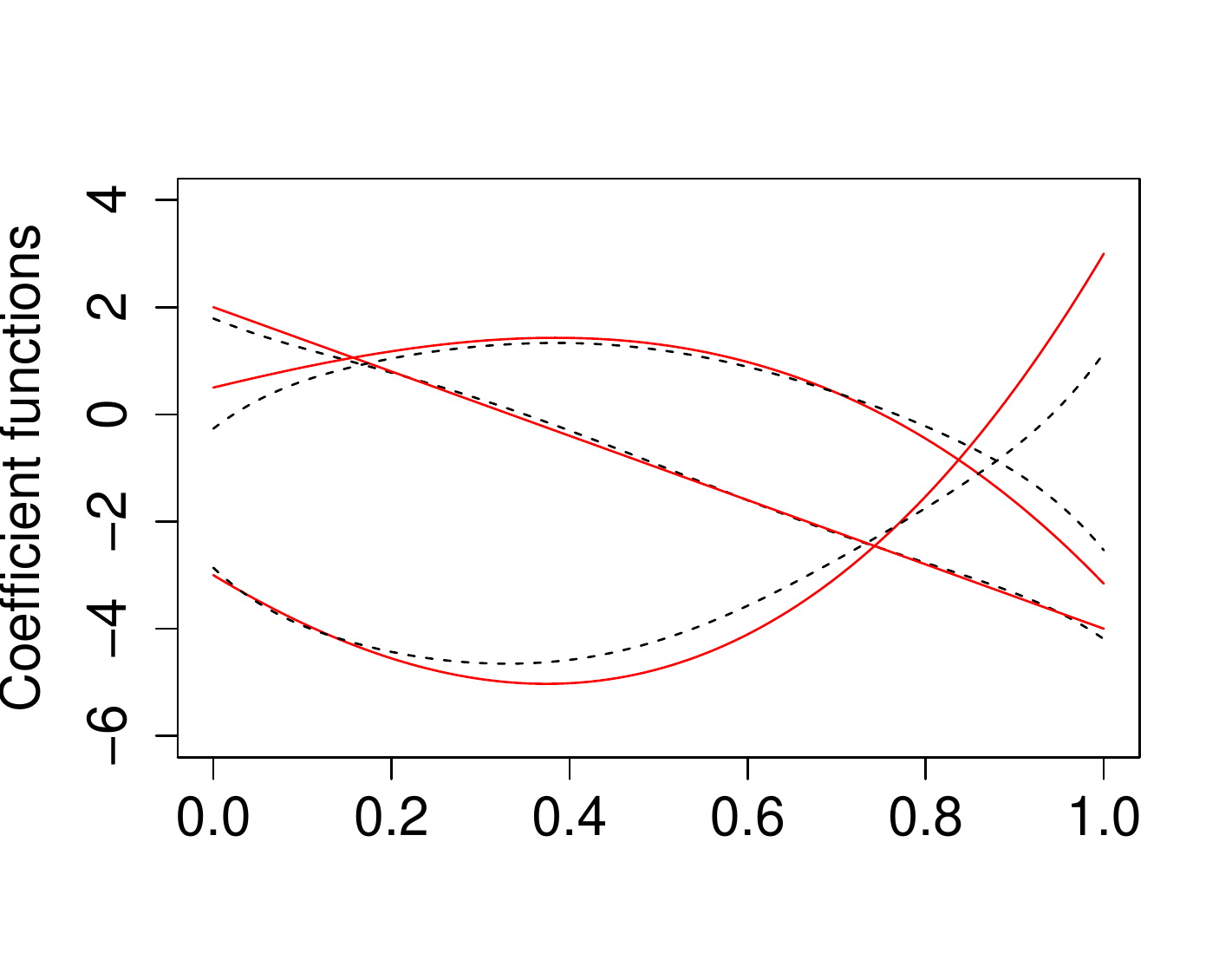} 
		}
		\subfigure[]
		{
			\includegraphics[width=3.8cm]{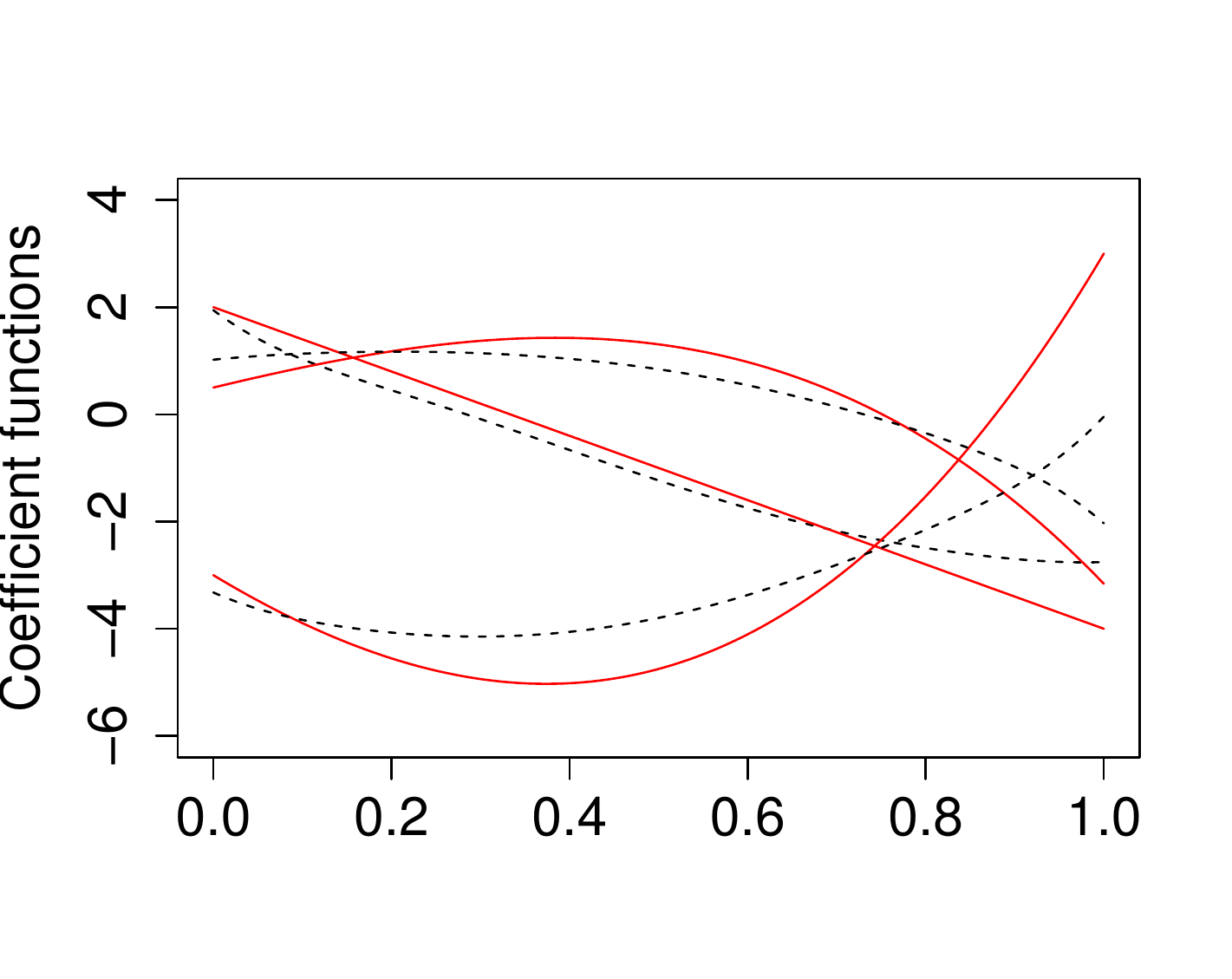}
		}
		\subfigure[]
		{
			\includegraphics[width=3.8cm]{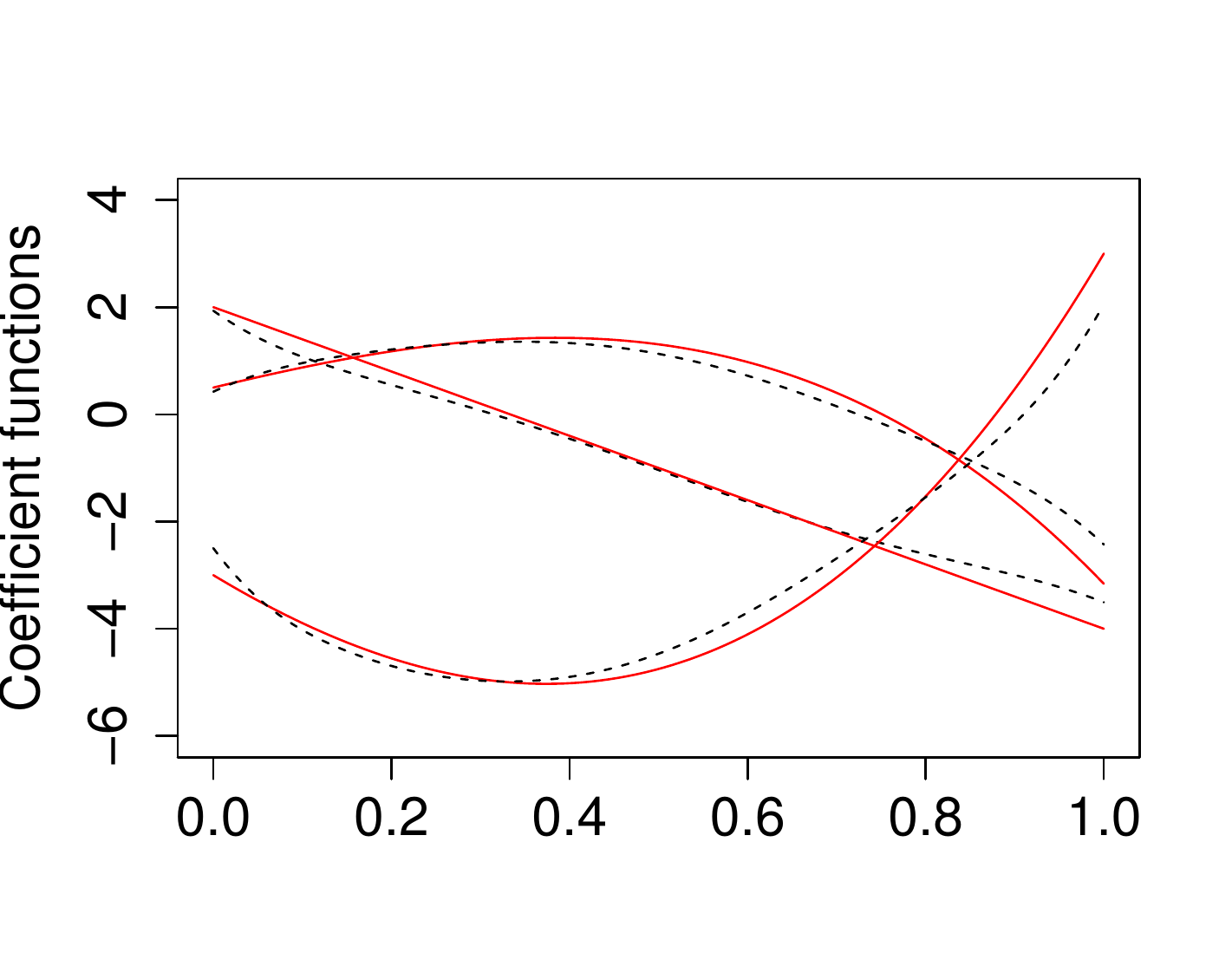} 
		}
		\subfigure[]
		{
			\includegraphics[width=3.8cm]{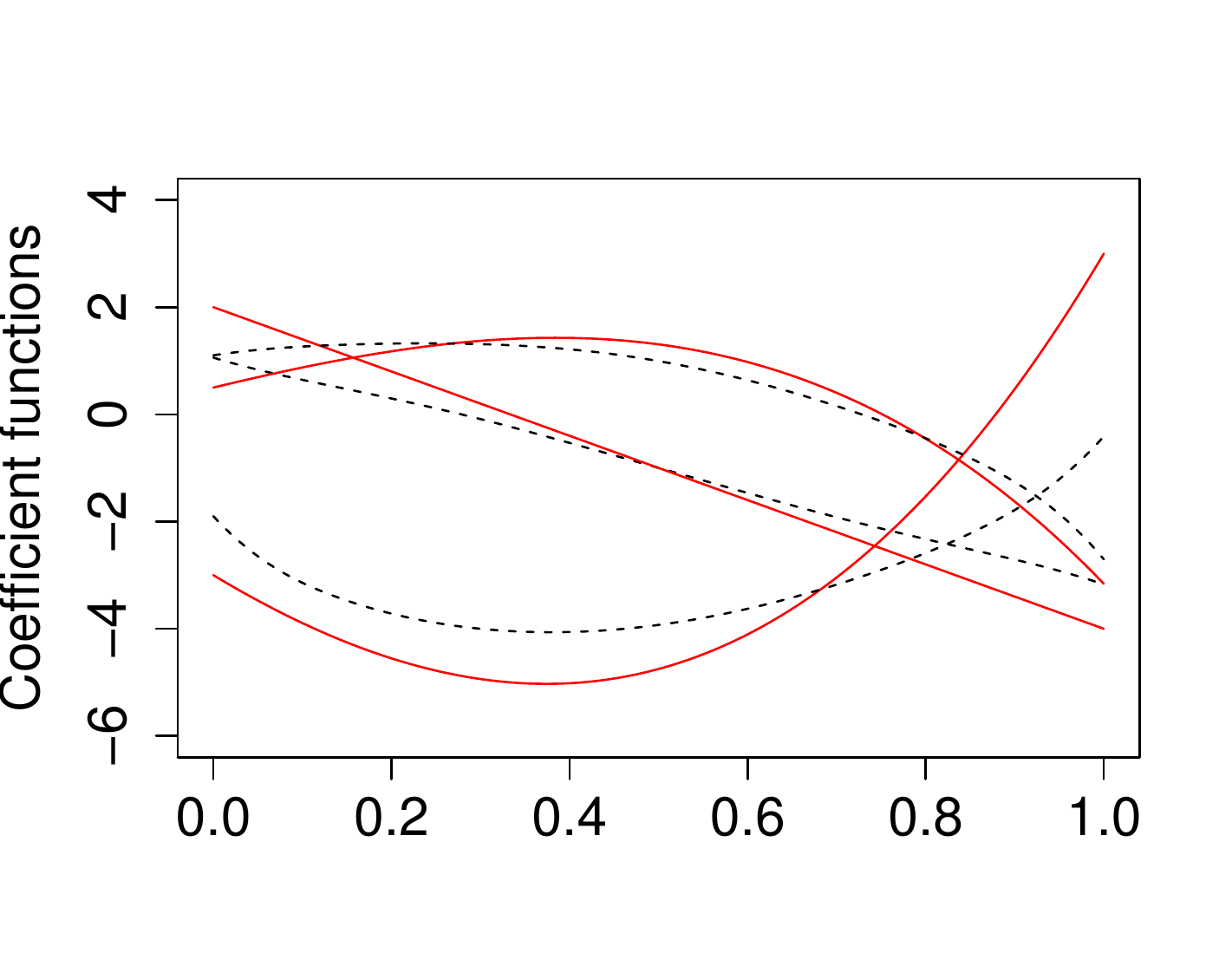}
		}
		\subfigure[]
		{
			\includegraphics[width=3.8cm]{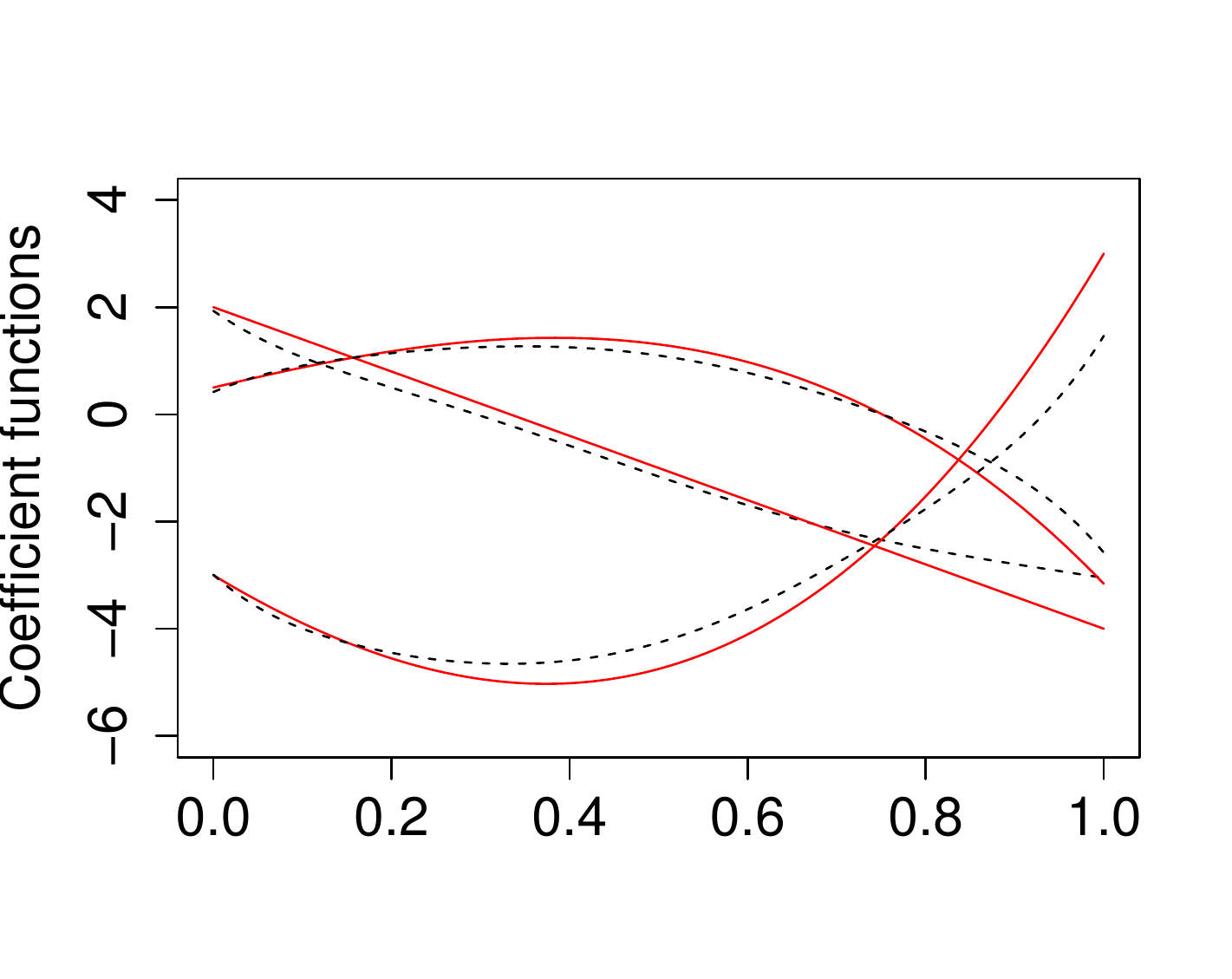} 
		}	
		\caption{Simulation results in Example 2: estimation of the coefficient functions by the proposed approach. (a) and (e) balanced structure with $n=60$, (b) and (f) balanced structure with $n=300$, (c) and (g) unbalanced structure with $n=60$, (d) and (h) unbalanced structure with $n=300$. (a)-(d) $a_{il}\sim N(2,1)$, and (b)-(h) $a_{il}\sim U(0,4)$. Red solid lines represent the true coefficient functions, and black dashed lines represent the estimated coefficient functions.} 
		\label{fig:coef3}
	\end{figure}

\section{Empirical studies}\label{sec4}
In this section, we apply the proposed approach to China's air quality data. This data was collected and organized by the World Air Quality Index \url{http://aqicn.org/}. The concentrations of six major air pollutants, that is, $\text{PM}_{2.5}$, $\text{PM}_{10}$, $\text{SO}_{2}$, $\text{O}_{3}$, $\text{NO}_{2}$ and CO, are recorded once every eight hours. The dataset consists of 1655 air monitoring stations from February 1, 2015, to January 31, 2016. 
The response of interest is the average $\text{PM}_{2.5}$ concentration during the time period in the study. In the literature, it has been suggested that $\text{PM}_{2.5}$ concentration is strongly correlated with $\text{NO}_{2}$ concentration \citep{Blanchard2005, Hua2020}. In fact, about 50\% of total $\text{PM}_{2.5}$ in the air is formed from $\text{NO}_{2}$ and from SO$_2$ as well.  China covers large areas, and these areas have diverse meteorological conditions, population, local industry, traffic, and so on. These differences may lead to distinctive relationship patterns between $\text{PM}_{2.5}$ and $\text{NO}_{2}$ over this huge geographical region. However, most of the existing studies have paid insufficient attention to the potential heterogeneity in the relationship between $\text{PM}_{2.5}$ and $\text{NO}_{2}$ (P-N relationship). As such, we aim to investigate the complex P-N relationship, that is, to explore the underlying subgroups concerning the P-N relationship, and to characterize the relationship pattern for each subgroup. 

We focus on the most populated cities in China, which are located in the range of 108$^\circ$ to 135$^\circ$ longitude and 18$^\circ$ to 50$^\circ$ latitude and consist of 25 provinces in Northeast, North, East, Central, and South China. Most of the cities in the area taken into consideration have multiple monitoring stations, and hence, to reduce the computation cost, we randomly select 400 monitoring stations over these large geographical regions. After removing stations that miss more than $50\%$ measurements in NO$_2$, a total of 373 stations are included. The closer regions are more likely to possess similar climatic characteristics, meteorological conditions, and industry structure, and thus are more likely to have similar/identical P-N relationships. This ``neighborhood effect" can be easily incorporated into our model by setting weights $w_{ij}=1/d_{ij}$ where $d_{ij}$ is the spherical distance (km) between monitoring stations $i$ and $j$. 

The proposed approach utilizes B-spline with an order $\tilde{q}=4$ and the number of knots $\tilde{m}=15$ to approximate the $\text{NO}_2$ concentration (prediction function), and order $q=4$ and the number of knots $m=8$ for unknown coefficient functions. The tuning parameters are selected by following the same rule as given in Section \ref{sec32}. With the optimal tuning parameters $(\lambda_1^*,\lambda_2^*)=(0.1,0.6)$, five subgroups are identified, with sizes 102, 121, 51, 52, and 47. Figure \ref{fig:em} (a) displays the grouping results, where monitoring stations within the same subgroup are marked in the same color. We observe a clear spatial clustering that is highly consistent with geographical regionalization of China. The first subgroup (marked in green) consists of sites in Northeast China, the second subgroup (marked in red) is composed mainly of sites in North China, the third subgroup (marked in orange), the fourth subgroup (marked in purple), and the fifth subgroup (marked in blue) consists of sites in East, Central, and South China, respectively. In addition, our approach successfully combines sites that are geographically nonadjacent yet showing similar P-N relationship into one subgroup. For example, the Northeast (first) subgroup includes sites across three provinces in Northeast China (Heilongjiang, Jilin, and Liaoning) and those in north edge of Jiaodong Peninsular. The Jiaodong Peninsular is offshore across from the Liaodong peninsula which is located south of Liaoning province, and both have temperate oceanity monsoon climate. The estimated coefficient functions of all subgroups are displayed in Figure \ref{fig:em} (b). Overall, the five subgroups have significantly different coefficient functions. 

\begin{figure}[H]
	\centering   
	\subfigure[]{\includegraphics[height=6cm]{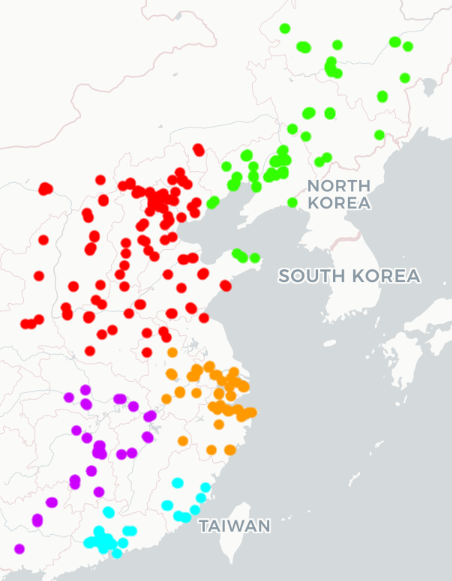}}
	\quad 
	\subfigure[]{\includegraphics[height=7.5cm]{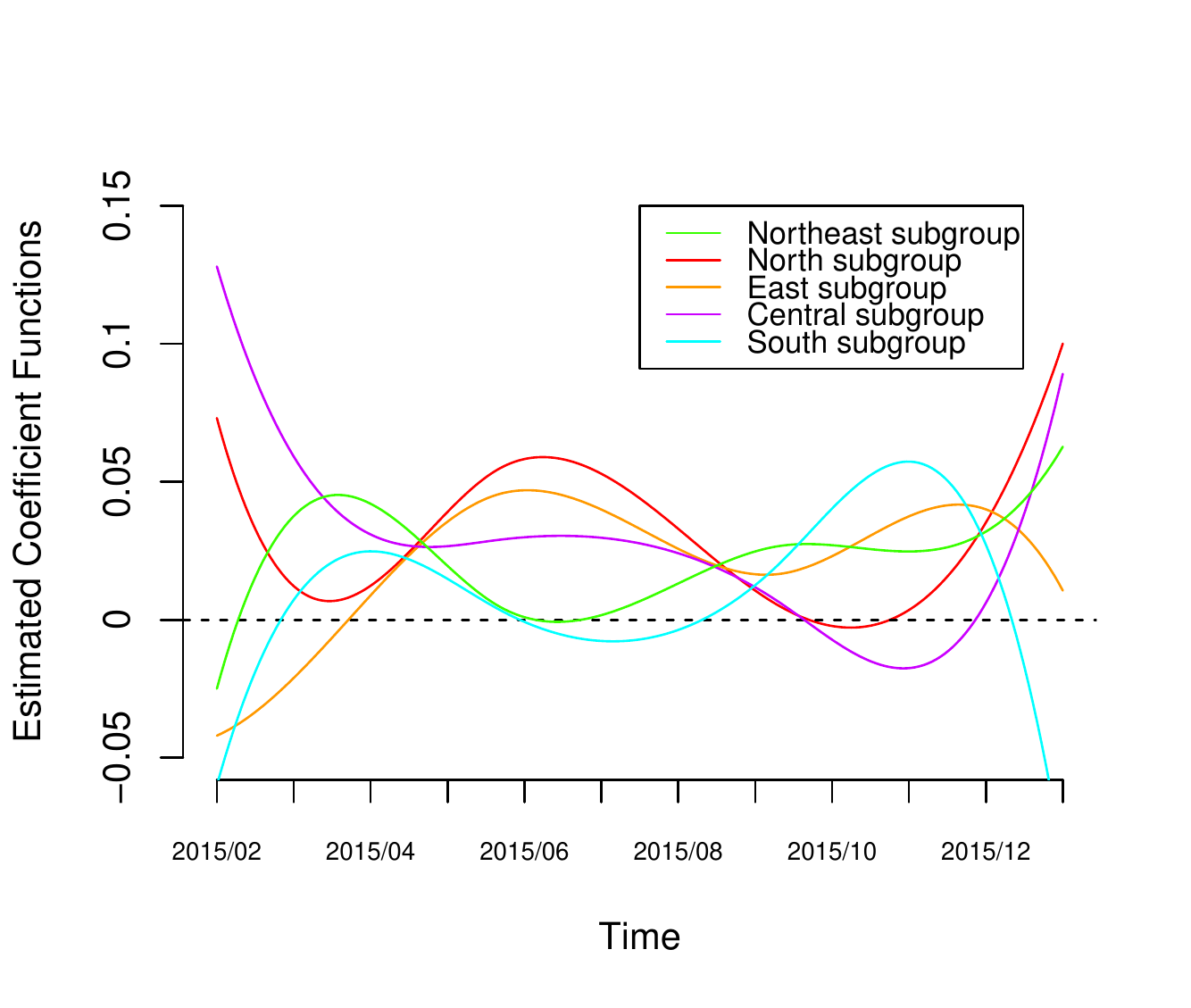}}
	\caption{Data analysis: (a) geographical distribution and (b) estimated coefficient functions of five subgroups. Each color represents a different subgroup.} 
	\label{fig:em}
\end{figure}

Data are also analyzed using alternative methods: FIMR, Resi, and Resp. To make different approaches comparable, we fix the number of subgroups in the alternatives to the same as that of the proposed approach. 
We compute the Normalized Mutual Information (NMI) to measure the similarity between the subgrouping results, with range $[0,1]$ and a larger value indicating a higher degree of similarity. The results are provided in Table C3 in the supplementary material. We find that the proposed approach has low similarity with the alternatives. As has been noted in many published studies, it is difficult to objectively assess which set of subgroup analysis results are more sensible. As such, we evaluate prediction performance, which may provide support to a certain extent. Specifically,  95\%, 90\%, 85\%, and 80\% of the sites are randomly selected as the training data, and the remaining sites are served as the testing data. Estimation is conducted using the training set, and prediction is made with the testing set. The process is repeated independently 50 times. Besides the aforementioned alternatives, we also apply the homogeneous FLM (Homo) to this data. Figure \ref{fig:MSE} presents the prediction mean square error (MSE). It can be observed that 
the proposed approach outperforms the alternatives in prediction. 
\begin{figure}[H]
	\centering    %
	\subfigure[] %
	{
		\begin{minipage}{0.46 \linewidth}
			\centering
			\includegraphics[width=5.5cm]{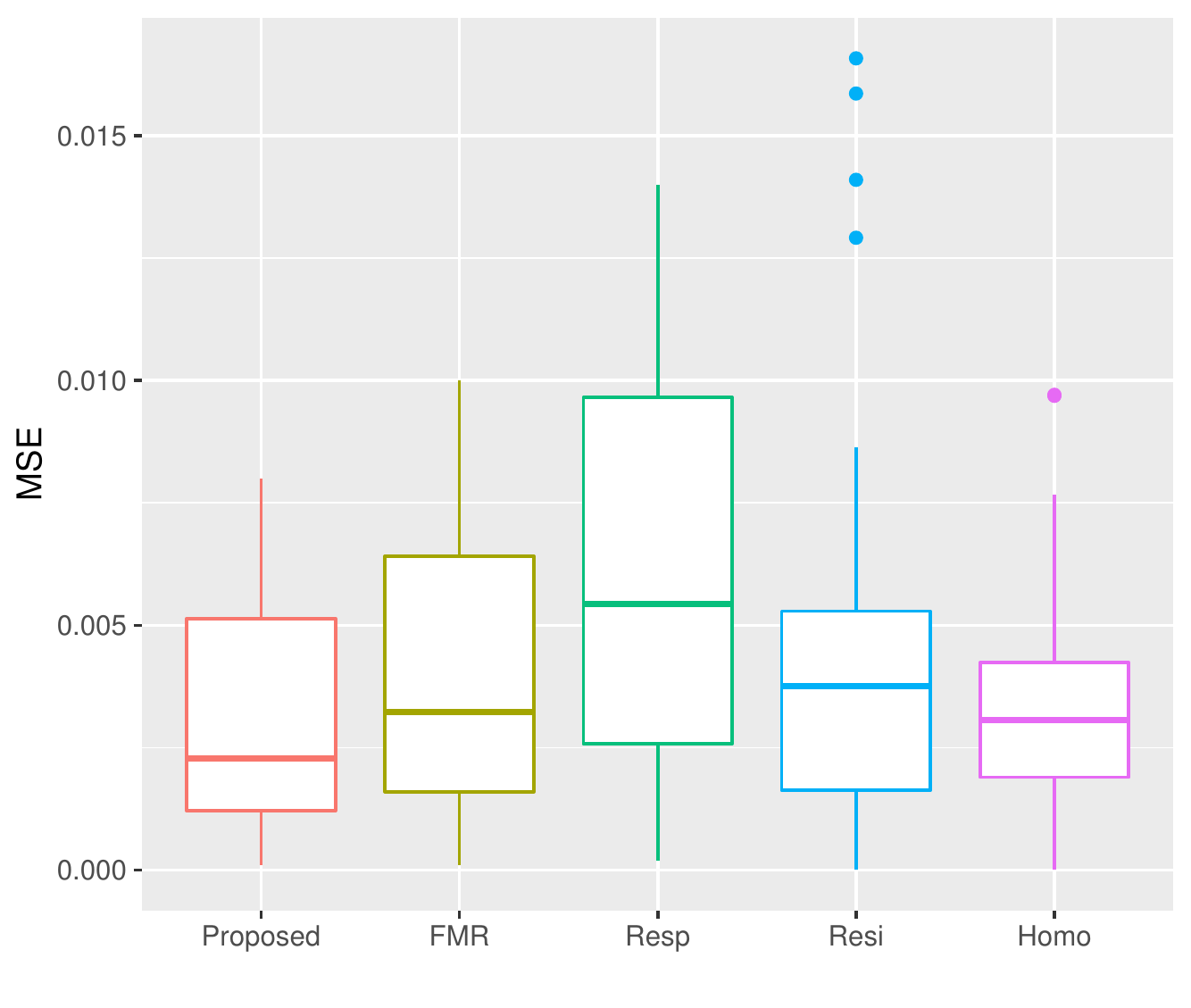}
		\end{minipage}
	}
	\subfigure[] %
	{
		\begin{minipage}{0.46 \linewidth}
			\centering     
			\includegraphics[width=5.5cm]{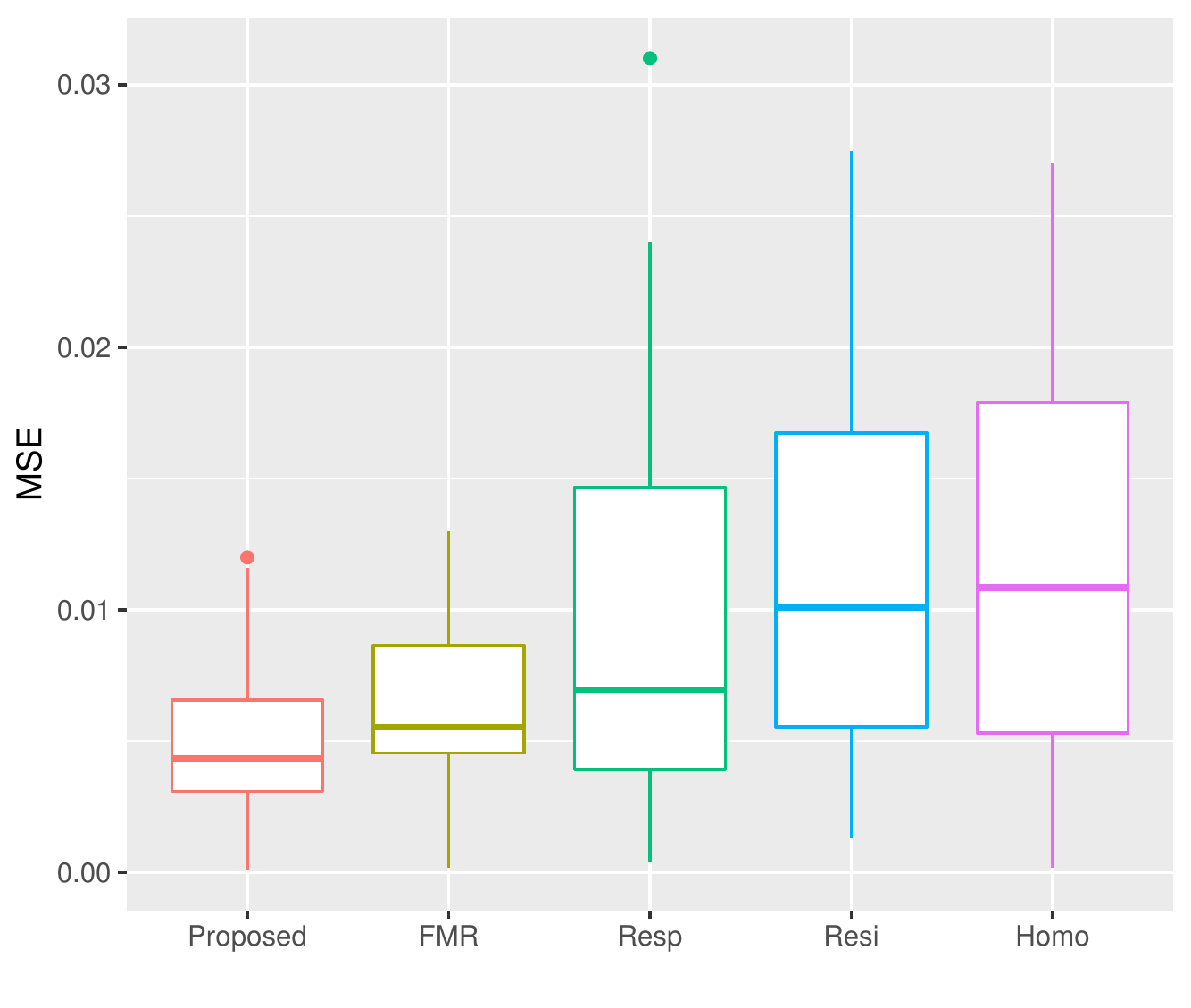}
		\end{minipage}
	}
	\subfigure[] 
	{
		\begin{minipage}{0.46 \linewidth}
			\centering     
			\includegraphics[width=5.5cm]{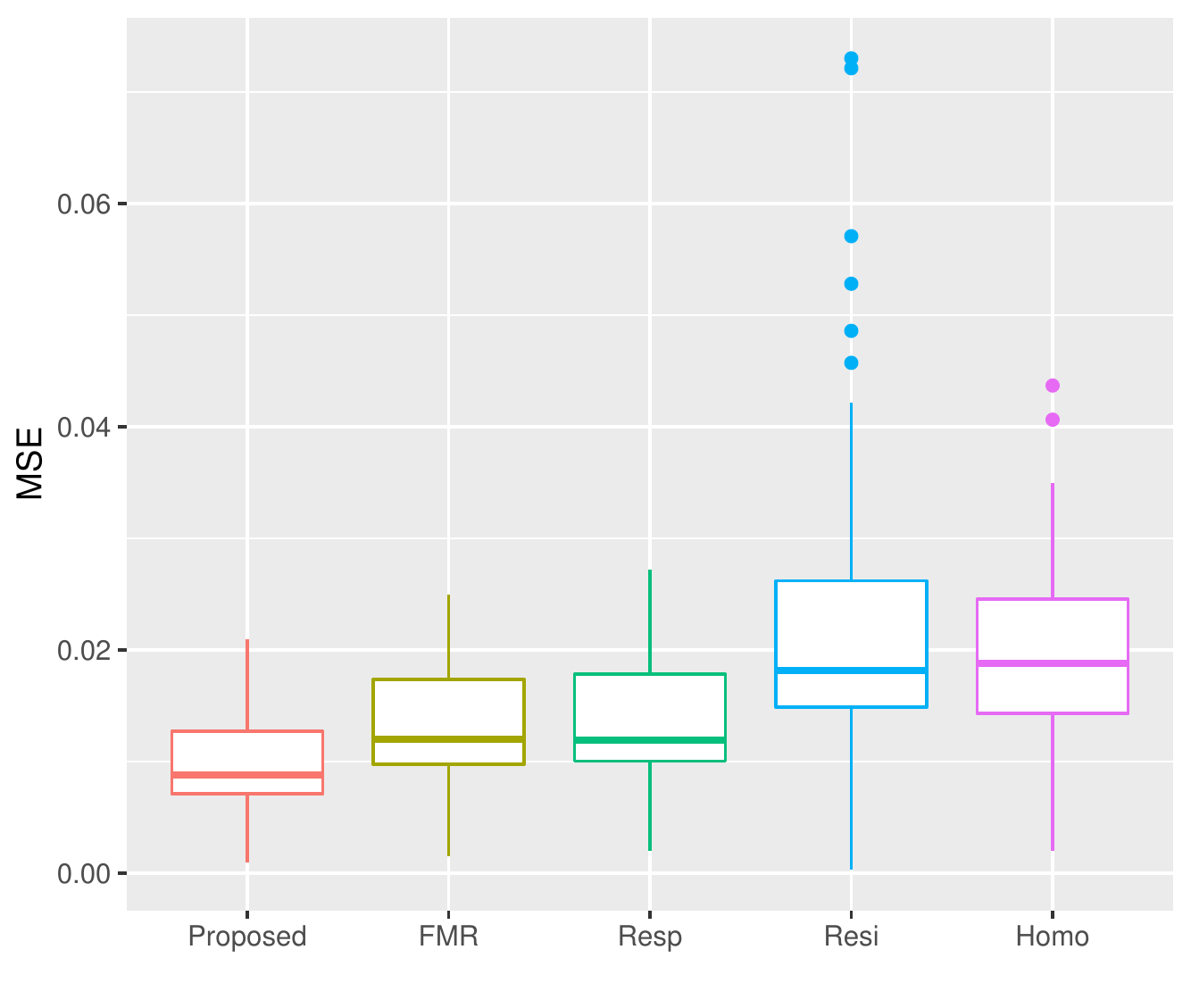}
		\end{minipage}
	}
	\subfigure[] 
	{
		\begin{minipage}{0.46 \linewidth}
			\centering     
			\includegraphics[width=5.5cm]{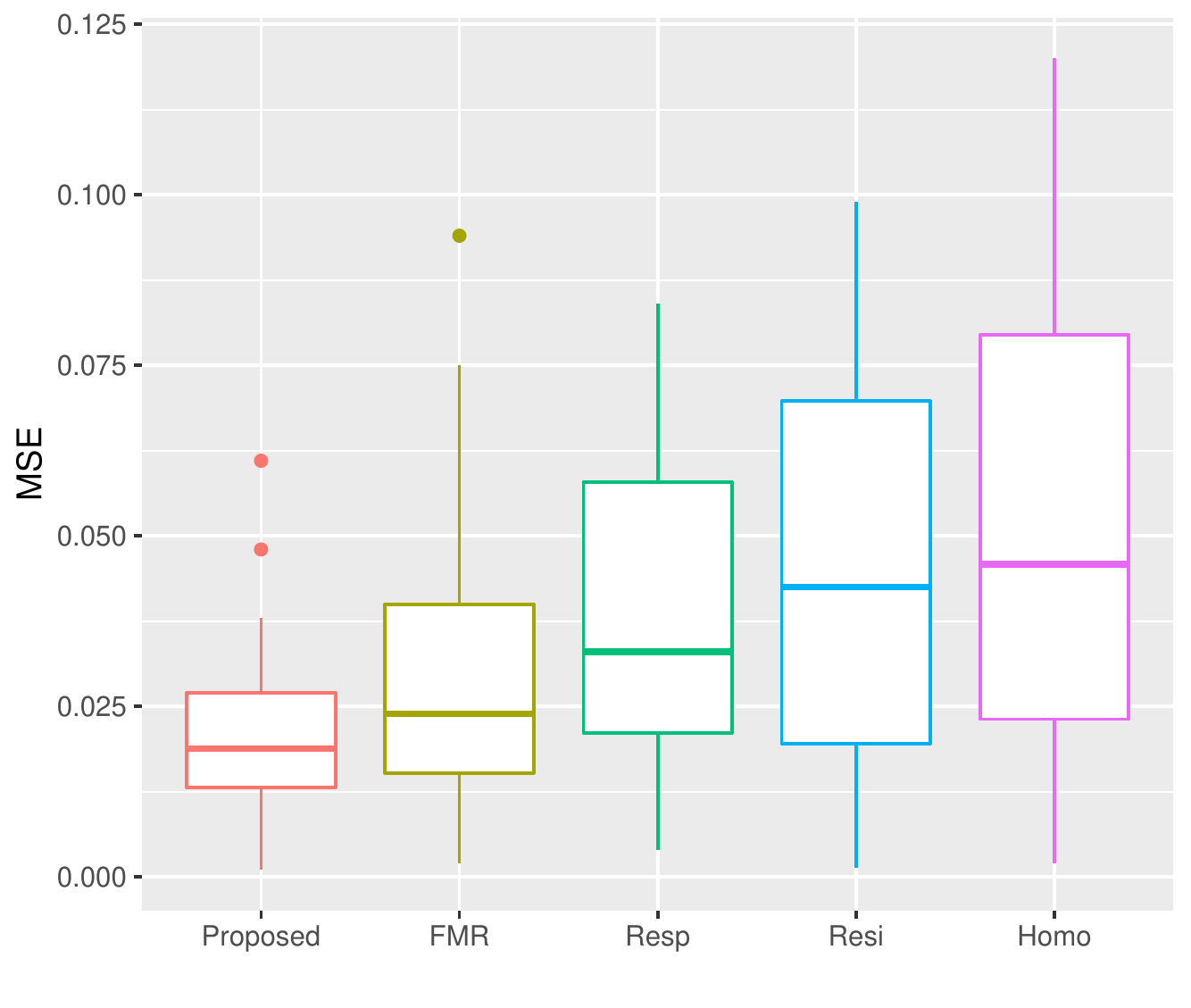}
		\end{minipage}
	}
	\caption{Prediction mean square errors on test sets which accounts for (a) 5\%, (b) 10\%, (c) 15\%, and (d) 20\% of data.} 
	\label{fig:MSE} 
\end{figure}

\section{Conclusion}\label{sec5}
In this paper, we developed a fusion penalized approach to conduct a subgroup analysis for the heterogeneous FLM where the responses are scalars and the covariates are functions. The proposed approach can identify the subgrouping structure and estimate the heterogeneous coefficient functions simultaneously and automatically. We adopted the B-spline method to approximate the unknown coefficient functions and applied additional penalization to smooth the estimated coefficient functions. In addition, we grouped subjects by penalizing pairwise differences of B-spline coefficient vectors. We presented the theoretical properties of the proposed approach to ensure its estimation consistency. We implemented an ADMM-based iterative algorithm for numerical computation. The simulation demonstrated favorable performance. In the analysis of air quality data, findings different from those of the alternatives were generated, and an improved prediction was observed. 

The proposed approach involves two tuning parameters. How to select viable tuning parameters, and in particular, the impact of heterogeneity on standard tuning parameter procedure, warrant further investigation. The proposed approach can be naturally extended to the FLM with multiple predictor functions. Beyond the FLM considered in this article, there are many other types of responses/models for functional data. It will be of interest to extend the proposed analysis to other responses/models.

\section*{Acknowledgements}
This research was supported by the National Natural Science Foundation of China (grants 12171479), and the Fund for building world-class universities (disciplines) of the Renmin University of China (Project No. KYGJC2022007). The computer resources were provided by the Public Computing Cloud Platform of the Renmin University of China. 
\par

\section*{Appendix}

\noindent
This supplementary material contains the details of ADMM algorithm and the proof of Theorem 1 and 2 in the main manuscript, as well as additional figures and tables in the experiments.
\par



\subsection*{A. Details of ADMM algorithm}

For a given $(\bm\eta,\bm\zeta)$, the component of the Lagrange function $L_n$ that depends on $\bm\theta$ is 
\begin{equation*}
	\mathcal{L}_n(\bm{\theta};\bm\lambda)=\frac{1}{2}(\bm{y}-\bm{H}\bm{\theta})^\top(\bm{y}-\bm{H}\bm{\theta})+\frac{1}{2}\lambda_1\bm{\theta}^\top\bm{G}\bm{\theta}+\frac{\delta}{2}\sum_{i<j}\Vert\bm\theta_i-\bm\theta_j-\bm\eta_{ij}-\frac{1}{\delta}\bm\zeta_{ij}\Vert_2^2.
\end{equation*}
Introduce $\bm \Delta=\{(\bm e_i-\bm e_j),i<j\}^\top$ with $\bm e_i$ being a $n$-dimensional column vector whose $i$th element equals 1 and the rest are equal to 0, and $\bm A=\bm \Delta \otimes \bm I_p$ where $\bm I_p$ is a $p\times p$ identity matrix and $\otimes$ is the Kronecker product. Then the expression of $\mathcal{L}_n(\bm{\theta};\bm\lambda)$ can be rewritten in a more compact form: 
\begin{equation*}
	\mathcal{L}_n(\bm{\theta};\bm\lambda)=\frac{1}{2}(\bm{y}-\bm{H}\bm{\theta})^\top(\bm{y}-\bm{H}\bm{\theta})+\frac{1}{2}\lambda_1\bm{\theta}^\top\bm{G}\bm{\theta}+\frac{\delta}{2}\Vert \bm A\bm \theta-\bm \eta-\frac{1}{\delta}\bm\zeta\Vert_2^2. 
\end{equation*}
Setting $\frac{\partial \mathcal{L}_n}{\partial \bm\theta}=0$ leads to the update equation for $\bm \theta$. For a given estimate $(\bm\eta^{(s)}, \bm\zeta^{(s)})$ at the $s$th iteration, the $(s+1)$th estimate of $\bm\theta$ is 
\begin{equation*}
	\bm{\theta}^{(s+1)}=\left(\bm{H}^\top\bm{H}+\lambda_1\bm{G}+\delta\bm{A}^\top\bm{A}\right)^{-1}\left[\bm{H}^\top\bm{y}+\delta\bm{A}^\top(\bm{\eta}^{(s)}+\frac{1}{\delta}\bm{\zeta}^{(s)})\right]. 
\end{equation*}
For $\bm\eta_{ij}$, the relevant component in $L_n$ is 
\[\frac{\delta}{2}\Vert \bm\eta_{ij}-\bm\theta_i+\bm\theta_j+\frac{1}{\delta}\bm\zeta_{ij}\Vert_2^2+\rho_{\tau}(\Vert \bm\eta_{ij}\Vert_2,\omega_{ij}\lambda_2).\]
Hence, the close-form solution for the MCP penalty with $\tau>\frac{1}{\delta}$ is 
\begin{equation*}
	\bm\eta_{ij}^{(s+1)}=
	\begin{cases}
		\bm u_{ij}^{(s+1)}\quad& \text{if}\ \Vert \bm u_{ij}^{(s+1)}\Vert_2\geq\tau\omega_{ij}\lambda_2,\\
		\frac{\tau\delta}{\tau\delta-1}\left(1-\frac{\lambda_2}{\delta\Vert \bm u_{ij}^{(s+1)}\Vert_2}\right)_+\bm u_{ij}^{(s+1)} \quad&\text{if}\ \Vert \bm u_{ij}^{(s+1)}\Vert_2<\tau\omega_{ij}\lambda_2,
	\end{cases}
\end{equation*}
where $\bm u_{ij}^{(s+1)}=\bm \theta_i^{(s+1)}-\bm \theta_j^{(s+1)}-\frac{1}{\delta}\bm\zeta_{ij}^{(s)}$, and $(x)_+=x$ if $x>0$ and $(x)_+=0$ otherwise. Finally, the estimate of $\bm \eta$ is updated as 
\[\bm\zeta^{(s+1)}=\bm\zeta^{(s)}+\delta(\bm A\bm\theta^{(s+1)}-\bm\eta^{(s+1)}).\] 

\subsection*{B. Proofs of Theorem 1 and 2}\label{appA}
\subsubsection*{Proof of Theorem 1}\label{proofthm1}
By Theorem 3.1 of \cite{cardot03} and Conditions (C1)-(C3), we have 
$$\text{E}(\Vert \hat\beta_i^{\text{or}}-\beta_i^\text{0}\Vert_2^2)=O_p(\phi_{|\mathcal{G}_{k}|}).$$ 
As $\text{E}^2(\Vert\hat\beta_i^{\hat{\text{or}}}-\beta_i^\text{0}\Vert_2)\leq \text{E}(\Vert \hat\beta_i^{\text{or}}-\beta_i^\text{0}\Vert_2^2)$, then 
$$\text{Var}(\Vert \hat\beta_i^{\text{or}}-\beta_i^\text{0}\Vert_2)=\text{E}(\Vert \hat\beta_i^{\text{or}}-\beta_i^\text{o}\Vert_2^2)-\text{E}^2(\Vert\hat\beta_i^{\text{or}}-\beta_i^\text{0}\Vert_2)=O_p(\phi_{|\mathcal{G}_{k}|}).$$ 
Let $V_k=\text{Var}(\Vert \hat\beta_i^{\text{or}}-\beta_i^0\Vert_2)$ and $E_k=\text{E}(\Vert \hat\beta_i^{\text{or}}-\beta_i^0\Vert_2)$. For any $\epsilon >0$, there exists a constant $C_{\epsilon}^1$ such that $\Pr(V_k>C_{\epsilon}^1\phi_{|\mathcal{G}_{k}|})\leq\epsilon$. By Chebyshev Inequality, $\Pr(|\Vert \hat\beta_i^{\text{or}}-\beta_i^\text{0}\Vert_2-E_k|\geq \sigma)\leq \frac{V_k}{\sigma^2}$ for any $\sigma>0$. Then  
\begin{equation*}
	\begin{split}
		\Pr(\Vert \hat\beta_i^{\text{or}}-\beta_i^0\Vert_2\geq\sigma)&\leq \Pr(\Vert \hat\beta_i^{\text{or}}-\beta_i^0\Vert_2\geq\sigma, V_k\leq C_{\epsilon}^1\phi_{|\mathcal{G}_{k}|})+\Pr(V_k>C_{\epsilon}^1\phi_{|\mathcal{G}_{k}|})\\
		&\leq \Pr(\Vert \hat\beta_i^{\text{or}}-\beta_i^0\Vert_2-E_k+E_k\geq\sigma, V_k\leq C_{\epsilon}^1\phi_{|\mathcal{G}_{k}|})\\
		&+\Pr(V_k>C_{\epsilon}^1\phi_{|\mathcal{G}_{k}|})\\
		&\leq \Pr(\Vert \hat\beta_i^{\text{or}}-\beta_i^0\Vert_2-E_k\geq\frac{\sigma}{2},V_k\leq C_{\epsilon}^1\phi_{|\mathcal{G}_{k}|})+\Pr(E_k\geq\frac{\sigma}{2})\\
		&+\Pr(V_k>C_{\epsilon}^1\phi_{|\mathcal{G}_{k}|})\\
		&\leq\frac{4C_{\epsilon}^1\phi_{|\mathcal{G}_{k}|}}{\sigma^2}+\Pr(E_k\geq\frac{\sigma}{2})+\epsilon.
	\end{split}
\end{equation*}
Under Condition (C4) and all assumptions in Theorem 1, we have $\phi_{|\mathcal{G}_{k}|}=O(|\mathcal{G}_{\min}|^{-a})$, where $a\in (\frac{1+\sigma_0}{4},\frac{1-\sigma_0}{2}]$. Let $C_\epsilon^3=2(\epsilon^{-1}C_{\epsilon}^1)^{1/2}$ and $\sigma=C_\epsilon^3\phi_{|\mathcal{G}_{k}|}^{1/2}$. Then for sufficiently large $|\mathcal{G}_{k}|$, we have 
\[\Pr(\Vert \hat\beta_i^{\text{or}}-\beta_i^0\Vert_2\geq C_\epsilon^3\phi_{|\mathcal{G}_{k}|}^{1/2})\leq 3\epsilon.\]
As a result, $\Vert \hat\beta_i^{\text{or}}-\beta_i^0\Vert_2=O_p(\phi_{|\mathcal{G}_{k}|}^{1/2})$. 

\subsubsection*{Proof of Theorem 2}\label{proofthm2}
Let $\rho(t)=\lambda_2^{-1}\rho_{\tau}(t,\lambda_2)$. Similar to \cite{Ma2020}, we define 
\[L_n(\bm{\theta})=\frac{1}{2}\Vert\bm{y}-\bm{H}\bm{\theta}\Vert_2^2, D_n(\bm{\theta};\lambda_1)=\frac{1}{2}\lambda_1\bm{\theta}^\top\bm{G}\bm{\theta}, P_n(\bm{\theta};\lambda_2)=\lambda_2\sum_{i<j}\rho(\Vert\bm\theta_i-\bm\theta_j\Vert_2).\]
For simplicity, here we assume weight $w_{ij}=1$ for all $1\leq i<j\leq n$, and the following proof can be easily generalized to general $w_{ij}$'s.
Then $Q_n(\bm{\theta};\bm{\lambda})=L_n(\bm{\theta})+D_n(\bm{\theta};\lambda_1)+P_n(\bm{\theta};\lambda_2)$. 

Define $T:\mathcal{M}_{\mathcal{G}}\rightarrow \mathbb{R}^{Kp}$ as a mapping, under which $T(\bm\theta)$ is the $Kp$-vector consisting of $K$ vectors with dimension $p$ and its $k$th vector component equals the common value of $\bm\theta_i$ for $i\in \mathcal{G}_k$, and $T^*:\mathbb{R}^{np}\rightarrow\mathbb{R}^{Kp}$ as another mapping, under which $T^*(\bm{\theta})=\{|\mathcal{G}_k|^{-1}\sum_{i\in\mathcal{G}_k}\bm{\theta_i}^\top, k\in\{1,\ldots,K\} \}^\top$. Obviously, when $\bm\theta\in\mathcal{M}_{\mathcal{G}}$, $T(\bm\theta)=T^*(\bm\theta)$. 

Let $\bm\theta^0=({{\bm\theta}_1^0}^\top,\ldots,{{\bm\theta}_n^0}^\top)^\top$ be the true B-spline coefficient functions corresponding to the true subgroup set $\mathcal{G}$, and $\bm\alpha^0=({{\bm\alpha}_1^0}^\top,\ldots,{{\bm\alpha}_K^0}^\top)^\top$ be the distinct value of $\bm\theta^0$. 
Consider the neighborhood of $\bm\theta^0$: 
$$\Theta=\{\bm{\theta}\in\mathbb{R}^{np}: \underset{i}{\sup}\Vert\bm\theta_i-\bm\theta_i^0\Vert_2\leq\Lambda_n\triangleq\phi_{|\mathcal{G}_{k}|}^{1/2}m^{1/2}=o(1).\}$$
By Theorem 26 of \cite{Lyche2018}, there exist constants $c_{m,q}>0$, which are only dependent on $m$ and $q$, such that for $1\leq i\leq n$,
\begin{equation}\label{eq:norm}
	c_{m,q}\Vert \bm{\hat\theta}_i^\text{or}-\bm{\theta}_i^{0}\Vert_2\leq \Vert\bm B(\bm{\hat\theta}_i^{\text{or}}-\bm{\theta}_i^{0})\Vert_2=
	\Vert\hat\beta_i^{\text{or}}-\bm B \bm{\theta}_i^{0}\Vert_2\leq \Vert\hat\beta_i^{\text{or}}-\beta_i^0\Vert_2+\Vert \beta_i^0-\bm B \bm{\theta}_i^{0}\Vert_2,
\end{equation}
where $c_{m,q}\sim q^{-1}m^{-1/2}$. Note that the B-spline order $q$ is a constant. By (B.7) (Lemma \ref{lemma1}), $\Vert \beta_i^0-\bm B \bm{\theta}_i^{0}\Vert_2=o(\Vert\hat\beta_i^{\text{or}}-\beta_i^0\Vert_2)$. 
Therefore, by Theorem 1, we have $\Pr(\bm{\hat\theta}^{\text{or}}\in\Theta)\to 1$. For any $\bm\theta\in\mathcal{R}^{np}$, let $T^*(\bm\theta)=\bm\alpha=(\bm\alpha_1^\top,\ldots,\bm\alpha_K^\top)^\top$ and $\bm\theta^*=T^{-1}(T^*(\bm\theta))=T^{-1}(\bm\alpha)$. 

We will show that $\bm{\hat\theta}^{\text{or}}$ is a strictly local minimizer of the objective function $Q_n(\bm{\theta};\bm{\lambda})$ with probability approaching 1 through the following two steps: \\
\noindent(i) $Q_n(\bm{\theta}^*;\bm{\lambda})>Q_n(\bm{\hat\theta}^{\text{or}};\bm{\lambda})$ for any $\bm\theta\in\Theta$ and $\bm{\theta}^*\neq\bm{\theta}^{\text{or}}$. \\
\noindent(ii) For any $\bm\theta\in\Theta$, $Q_n(\bm{\theta};\bm{\lambda})\geq Q_n(\bm{\theta}^*;\bm{\lambda})$ with probability approaching 1.

We first prove the results in (i). Since $\bm{\hat\theta}^{\text{or}}$ is the global minimizer of $L_n(\bm{\theta})+D_n(\bm{\theta};\lambda_1)$ for $\bm\theta\in\mathcal{M}_{\mathcal{G}}$, then $L_n(\bm{\theta}^*)+D_n(\bm{\theta}^*;\lambda_1)>L_n(\bm{\hat\theta}^{\text{or}})+D_n(\bm{\hat\theta}^{\text{or}};\lambda_1)$. Hence we only need to consider $P_n(\bm{\theta}^*;\lambda_2)$. 
Again by Theorem 26 of \cite{Lyche2018}, there exist constants $c_m\sim m^{-1/2}$, such that for $1\leq k\leq K$,
\[\Vert\bm B\bm\alpha_k^0-\bm B\bm\alpha_{k'}^0\Vert_2\leq c_m \Vert\bm\alpha_k^0-\bm\alpha_{k'}^0\Vert_2.\]
Similarly, we have
\[\Vert\xi_k^0-\xi_{k'}^0\Vert_2\leq\Vert\bm B\bm\alpha_k^0-\bm B\bm\alpha_{k'}^0\Vert_2+\Vert\bm B\bm\alpha_k^0-\xi_k^0\Vert_2+\Vert\bm B\bm\alpha_{k'}^0-\xi_{k'}^0\Vert_2,\]
where $\Vert\bm B\bm\alpha_k^0-\xi_k^0\Vert_2=o(\Vert\hat\xi_k^{\text{or}}-\xi_k^0\Vert_2)$ and $\Vert\bm B\bm\alpha_{k'}^0-\xi_{k'}^0\Vert_2=o(\Vert\hat\xi_{k'}^{\text{or}}-\xi_{k'}^0\Vert_2)$. Thus $\Vert\bm\alpha_k^0-\bm\alpha_{k'}^0\Vert_2\geq c\Vert\xi_k^0-\xi_{k'}^0\Vert_2\geq cb$ for sufficiently large $n$.
Then,
\begin{eqnarray}\label{sourceneq}
	\Vert\bm\alpha_k-\bm\alpha_{k'}\Vert_2&\geq&\Vert\bm\alpha_k^0-\bm\alpha_{k'}^0\Vert_2-\Vert\bm\alpha_k-\bm\alpha_k^0\Vert_2-\Vert\bm\alpha_{k'}^0-\bm\alpha_{k'}\Vert_2\nonumber\\
	&\geq& \Vert\bm\alpha_k^0-\bm\alpha_{k'}^0\Vert_2-2\underset{j}{\sup}\Vert\bm\alpha_j-\bm\alpha_{j}^0\Vert_2\nonumber\\
	&\geq& c b-2\underset{j}{\sup}\Vert |\mathcal{G}_j|^{-1}\sum_{i\in\mathcal{G}_j}(\bm\theta_i-\bm\theta_i^0)\Vert_2\nonumber\\
	& \geq & c b-2\underset{j}{\sup}|\mathcal{G}_j|^{-1}\sum_{i\in\mathcal{G}_j}\Vert \bm\theta_i-\bm\theta_i^0\Vert_2\nonumber \\
	&\geq &c b-2\underset{i}{\sup}\Vert \bm\theta_i-\bm\theta_i^0\Vert_2\nonumber\\
	& \geq & c b-2\Lambda_n\nonumber\\
	& >& a\lambda_2,
\end{eqnarray}
where the last inequality follows from the assumption that $c b>a\lambda_2\gg\Lambda_n$. Consequently, $\rho(\Vert\bm\alpha_k-\bm\alpha_{k'}\Vert_2)$ is a constant by (C5). Since  
\[P_n(\bm\theta^*;\lambda_2)=\lambda_2\sum_{i<j}\rho(\Vert\bm\theta_i^*-\bm\theta_j^*\Vert_2)=\lambda_2\sum_{k\neq k'}\frac{|\mathcal{G}_k||\mathcal{G}_{k'}|}{2}\rho(\Vert\bm\alpha_k-\bm\alpha_{k'}\Vert_2),\]
then $P_n(\bm\theta^*;\lambda_2)=C_n$ where $C_n$ is a constant. Therefore, $L_n(\bm{\theta}^*)+D_n(\bm{\theta}^*;\lambda_1)+C_n>L_n(\bm{\hat\theta}^{\text{or}})+D_n(\bm{\hat\theta}^{\text{or}};\lambda_1)+C_n$, that is, $Q_n(\bm{\theta}^*;\bm{\lambda})>Q_n(\bm{\hat\theta}^{\text{or}};\bm{\lambda})$, for all $\bm{\theta^*}\neq \bm{\hat\theta}^{\text{or}}$. The result in (i) is proved. 

Now we prove the results in (ii). For $\bm\theta\in\Theta$, by Taylor's expansion, we have 
\[Q_n(\bm{\theta};\bm{\lambda})-Q_n(\bm{\theta}^*,\bm{\lambda})=\Gamma_1+\Gamma_2+\Gamma_3,\]
where $\Gamma_1=-(\bm y-\bm{H\theta}^m)^\top\bm H(\bm\theta-\bm\theta^*)$, $\Gamma_2=\lambda_1\bm{\theta}^{m^\top}\bm{G}(\bm\theta-\bm\theta^*)$, $\Gamma_3=\frac{\partial P_n(\bm{\theta};\lambda_2)}{\partial\bm{\theta}}\big|_{\bm{\theta}=\bm{\theta}^m}(\bm\theta-\bm{\theta}^*)$, and $\bm{\theta}^m=c\bm{\theta}+(1-c)\bm{\theta}^*$ for some constant $c\in(0,1)$.

Let $-\bm{H}^\top(\bm{y}-\bm{H\theta}^m)=\bm{w}=(\bm w_1^\top,\bm w_2^\top,...,\bm w_n^\top)^\top$. Then 
\begin{equation}\label{eq:gamma1}
	\begin{split}
		\Gamma_1=\bm{w}^\top(\bm{\theta}-\bm{\theta}^*)=\sum_{i=1}^{n}\bm w_i^\top(\bm\theta_i-\bm\theta_i^*)&=\sum_{k=1}^{K}\sum_{\{i,j\in\mathcal{G}_k, i<j\}}\frac{(\bm w_i-\bm w_j)^\top(\bm\theta_i-\bm\theta_j)}{|\mathcal{G}_k|}\\
		&\geq -\sum_{k=1}^{K}\sum_{\{i,j\in\mathcal{G}_k, i<j\}}\frac{2\underset{i}{\sup}\Vert \bm w_i\Vert_2\cdot\Vert\bm\theta_i-\bm\theta_j\Vert_2}{|\mathcal{G}_{\min}|},
	\end{split}
\end{equation}
where the third equality is obtained from $\bm{\theta}^*=T^{-1}(T^*(\bm{\theta}))$ and $\sum_{i=1}^{n}\bm{\omega}_i^\top \bm{\theta}_i^* = \sum_{k=1}^K\frac{1}{|\mathcal{G}_k|}\sum_{i,j\in\mathcal{G}_k}\bm{\omega}_i^\top\bm{\theta}_j$.

Let $\lambda_1\bm{G}\bm{\theta}^m=\bm{v}=(\bm v_1^\top,\cdots,\bm v_n^\top)^\top$. Then 
\begin{equation}\label{eq:gamma2}
	\Gamma_2=\bm{v}^T(\bm{\theta}-\bm{\theta}^*)\geq -\sum_{k=1}^{K}\sum_{\{i,j\in\mathcal{G}_k, i<j\}}\frac{2\underset{i}{\sup}\Vert \bm v_i\Vert_2\cdot\Vert\bm \theta_i-\bm \theta_j\Vert_2}{|\mathcal{G}_{\min}|}.
\end{equation}

\begin{eqnarray*}
	\Gamma_3&=&\frac{\partial P_n(\bm{\theta};\lambda_2)}{\partial\bm{\theta}}\big|_{\bm{\theta}=\bm{\theta}^m}(\bm{\theta-\theta^*})\\\nonumber
	&=&\lambda_2\sum_{i<j}\rho'(\Vert\bm\theta_i^m-\bm\theta_j^m\Vert_2)\frac{(\bm\theta_i^m-\bm\theta_j^m)^\top}{\Vert\bm\theta_i^m-\bm\theta_j^m\Vert_2}\left[(\bm\theta_i-\bm\theta_i^*)-(\bm\theta_j-\bm\theta_j^*)\right].
\end{eqnarray*}
Recall that $\underset{i}{\sup}\Vert\bm\theta_i-\bm\theta_i^0\Vert_2\leq\Lambda_n$ for $\bm{\theta}\in\Theta$. By the same reasoning as (\ref{sourceneq}), we have $\underset{k}{\sup}\Vert\bm\alpha_k-\bm\alpha_{k}^0\Vert_2\leq\Lambda_n$. Then $\underset{i}{\sup}\Vert\bm\theta_i^*-\bm\theta_i^0\Vert_2\leq\Lambda_n$. Since $\bm{\theta}^m$ is between $\bm\theta$ and $\bm\theta^*$, we have 
\begin{equation}\label{eq:Lambda_n}
	\underset{i}{\sup}\Vert\bm\theta_i^m-\bm\theta_i^0\Vert_2\leq\Lambda_n. 
\end{equation}
Hence for $i\in\mathcal{G}_k$, $j\in\mathcal{G}_{k'}$, $k\neq k'$, 
\begin{eqnarray*}
	\Vert\bm\theta_i^m-\bm\theta_j^m\Vert_2&\geq&\Vert\bm\theta_i^0-\bm\theta_j^0\Vert_2-2\underset{l}{\sup}\Vert\bm\theta_l^m-\bm\theta_l^0\Vert_2\\\nonumber
	&=&\Vert\bm\alpha_k^0-\alpha_{k'}^0\Vert_2-2\underset{l}{\sup}\Vert\bm\theta_l^m-\bm\theta_l^0\Vert_2\\\nonumber
	&\geq& b-2\Lambda_n>a\lambda_2,
\end{eqnarray*}
and thus $\rho'(\Vert\bm\theta_i^m-\bm\theta_j^m\Vert_2)=0$. Consequently, 
\begin{equation*} 
	\Gamma_3=\lambda_2\sum_{k=1}^K\sum_{i,j\in\mathcal{G}_k,i<j}\rho'(\Vert\bm\theta_i^m-\bm\theta_j^m\Vert_2)\frac{(\bm\theta_i^m-\bm\theta_j^m)^\top}{\Vert\bm\theta_i^m-\bm\theta_j^m\Vert_2}\left[(\bm\theta_i-\bm\theta_i^*)-(\bm\theta_j-\bm\theta_j^*)\right].
\end{equation*}
When $i,j\in\mathcal{G}_k$, $\bm\theta_i^*=\bm\theta_j^*$ and $\bm\theta_i^m-\bm\theta_j^m=c(\bm\theta_i-\bm\theta_j)$. Thus 
\begin{equation*} 
	\Gamma_3=\lambda_2\sum_{k=1}^K\sum_{i,j\in\mathcal{G}_k,i<j}\rho'(\Vert\bm\theta_i^m-\bm\theta_j^m\Vert_2)\Vert\bm\theta_i-\bm\theta_j\Vert_2.
\end{equation*}
Further, 
\begin{equation*}
	\begin{split}
		\underset{k}{\sup}\underset{i,j\in\mathcal{G}_k}{\sup}\Vert\bm\theta_i^m-\bm\theta_j^m\Vert_2&=\underset{k}{\sup}\underset{i,j\in\mathcal{G}_k}{\sup}\Vert\bm\theta_i^m-\bm\theta_i^*+\bm\theta_i^*-\bm\theta_j^*+\bm\theta_j^*-\bm\theta_j^m\Vert_2\\
		&\leq 2\underset{i}{\sup}\Vert\bm\theta_i^m-\bm\theta_i^*\Vert_2\leq 4\Lambda_n.
	\end{split}
\end{equation*}
By the concavity of $\rho(\cdot)$, 
\begin{equation}\label{eq:gamma3}
	\Gamma_3\geq\sum_{k=1}^{K}\sum_{\{i,j\in\mathcal{G}_k,i<j\}}\lambda_2\rho'(4\Lambda_n)\Vert\bm\theta_i-\bm\theta_j\Vert_2.
\end{equation}

Therefore, with (\ref{eq:gamma1}), (\ref{eq:gamma2}), and (\ref{eq:gamma3}), we have 
\begin{eqnarray*}
	Q_n(\bm{\theta};\bm{\lambda})-Q_n(\bm{\theta}^*,\bm{\lambda})&=&\Gamma_1+\Gamma_2+\Gamma_3\nonumber\\
	&\geq&\sum_{k=1}^{K}\sum_{\{i,j\in\mathcal{G}_k,i<j\}}\Vert\bm\theta_i-\bm\theta_j\Vert_2\left[\lambda_2\rho'(4\Lambda_n)-2\frac{\underset{i}{\sup}\Vert \bm w_i\Vert_2+\underset{i}{\sup}\Vert\bm v_i\Vert_2}{|\mathcal{G}_{\min}|}\right]
\end{eqnarray*}
By Lemma \ref{lemma1}, we have $\lambda_2\gg\frac{\underset{i}{\sup}\Vert \bm w_i\Vert_2+\underset{i}{\sup}\Vert \bm v_i\Vert_2}{|\mathcal{G}_{\min}|}$ with 
probability approaching 1. Since $\rho'(4\Lambda_n)\to 1$, we have  
\begin{equation*}
	Q_n(\bm{\theta};\bm{\lambda})-Q_n(\bm{\theta}^*,\bm{\lambda})\geq 0
\end{equation*}
with probability approaching 1, so that the result in (ii) is proved. Therefore, there exists a local minimizer $\bm{\hat\theta}$ of the objective function $\bm{Q_n}(\bm{\theta};\bm{\lambda})$, such that $\Pr(\bm{\hat\theta}=\bm{\hat\theta}^{\text{or}})\to1$. By (\ref{eq:norm}), we can conclude that there exists a local minimizer $\bm{\hat\theta}$ of the objective function ${Q_n}(\bm{\theta};\bm{\lambda})$ such that $\Pr(\bm{\hat\beta}=\bm{\hat\beta}^{\text{or}})\to 1$. 

\begin{lemma}\label{lemma1}
	Under Conditions (C1)-(C3), $\bm{w}$ and $\bm{v}$ are defined in the proof of Theorem 2, we have $\lambda_2\gg\frac{\underset{i}{\sup}\Vert w_i\Vert_2+\underset{i}{\sup}\Vert v_i\Vert_2}{|\mathcal{G}_{\min}|}$ with probability approaching 1. 
\end{lemma}

\begin{proof}
	\indent Let $s_i=\bm B\bm\theta_i^0$. Define $e_i=\beta_i^0-s_i$ and $e_{i2}=(\beta_i^0){''}-s_i{''}$. By Theorem 48 of \cite{Lyche2018}, we have 
	\begin{eqnarray}\label{errorfunc}
		&&\Vert e_i\Vert_2\leq C_{q}(\frac{T}{m})^{q-2}\Vert {\beta_i^0}^{(q-2)}\Vert_2,\nonumber\\
		&&\Vert e_{i2}\Vert_2\leq C_{q}(\frac{T}{m})^{q-4}\Vert {\beta_i^0}^{(q-2)}\Vert_2,
	\end{eqnarray}
	where $C_q$ is a constant depending only on $q$, and $\beta_i^{(q-2)}$ is the $(q-2)$-th derivative of $\beta_i$.
	Let $E_i=\int_\chi X_i(t)e_i(t)\text{d}t$. Then $y_i=\bm H_i\bm\theta_i^0+E_i+\epsilon_i$, and hence 
	\begin{eqnarray*}
		\bm w_i&=&-\bm H_i^{\top}(y_i-\bm H_i\bm\theta_i^m)\\
		&=&\bm H_i^{\top}\bm H_i(\bm\theta_i^m-\bm\theta_i^0)-\bm H_i^{\top}E_i-\bm H_i^{\top}\epsilon_i
	\end{eqnarray*}
	According to Lemma 8 from \cite{Garoni2014}, 
	\begin{equation}\label{eq:bprop}
		\sum_{l=1}^p \Vert B_l\Vert_2^2=\int_\chi\sum_{l=1}^p (B_l(t))^2\text{d}t\leq\int_\chi\sum_{l=1}^p B_l(t)\text{d}t\leq T.
	\end{equation}
	Since $\underset{i}{\sup}\Vert\bm\theta_i^m-\bm\theta_i^0\Vert_2\leq\Lambda_n$, (\ref{eq:bprop}) and $H_i=(\int_\chi X_i(t)B_1(t)\text{d}t,\ldots, \int_\chi X_i(t)B_p(t)\text{d}t)$, then 
	\begin{eqnarray}\label{eq:w1}
		\Vert\bm H_i^{\top}\bm H_i(\bm\theta_i^m-\bm\theta_i^0)\Vert_2&\leq&\Vert\bm H_i^{\top}\bm H_i\Vert_F\Vert \bm\theta_i^m-\bm\theta_i^0\Vert_2\nonumber\\
		&\leq&\Lambda_n\sum_{l=1}^p (\int_\chi X_i(t)B_{l}(t)\text{d}t)^2\nonumber\\
		&\leq&\Lambda_n\Vert X_i\Vert_2^2\sum_{l=1}^p\Vert B_l\Vert_2^2 \nonumber\\
		&\leq&T\Lambda_n\Vert X_i\Vert_2^2.
	\end{eqnarray}
	Since 
	\[\bm{H}_{i}^\top E_i=\left(\int_\chi X_i(t)B_1(t)\text{d}t\int_\chi X_i(t)e_i(t)\text{d}t,\ldots,\int_{\chi}X_i(t)B_p(t)dt\int_\chi X_i(t)e_i(t)\text{d}t\right)^\top,\]
	thus 
	\begin{eqnarray}\label{eq:w2}
		&&\Vert\bm H_i^\top E_i\Vert_2\nonumber\\
		&&=\sqrt{(\int_\chi X_i(t)e_i(t)\text{d}t)^2\left[(\int_\chi X_i(t)B_1(t)\text{d}t)^2+\cdots+(\int_{\chi}X_i(t)B_p(t)\text{d}t)^2\right]}\nonumber\\
		&&\leq\Vert X_i\Vert_2\sqrt{\sum_{l=1}^p\Vert B_l\Vert_2^2}\int_\chi X_i(t)e_i(t)\text{d}t\nonumber\\
		&&\leq\sqrt{T}\Vert X_i\Vert_2^2\Vert e_i\Vert_2,
	\end{eqnarray}
	Since 
	\[\bm{H}_{i}^\top\epsilon_i= (\epsilon_i\int_\chi X_i(t)B_1(t)\text{d}t,\ldots, \epsilon_i\int_\chi X_i(t)B_p(t)\text{d}t)^\top,\]
	then 
	\begin{eqnarray}\label{eq:w3}
		\Vert\bm{H}_{i}^\top\epsilon_i\Vert_2&=&\sqrt{\epsilon_i^2(\int_\chi X_i(t)B_1(t)\text{d}t)^2+\ldots+\epsilon_i^2(\int_\chi X_i(t)B_p(t)\text{d}t)^2}\nonumber\\
		&\leq&|\epsilon_i|\Vert X_i\Vert_2\sqrt{\sum_{l=1}^p\Vert B_l\Vert_2^2}\nonumber\\
		&\leq&\sqrt{T}|\epsilon_i|\Vert X_i\Vert_2.
	\end{eqnarray}
	Therefore, with (\ref{eq:w1}), (\ref{eq:w2}), and (\ref{eq:w3}), we have, 
	\begin{equation*}
		\Vert \bm w_i\Vert_2\leq\sqrt{T}\Vert X_i\Vert_2(\sqrt{T}\Lambda_n\Vert X_i\Vert_2+\Vert X_i\Vert_2\Vert e_i\Vert_2+|\epsilon_i|).
	\end{equation*}
	Since $\Pr(|\epsilon_i|>C\sqrt{p})\leq \frac{\sigma^2}{Cp}$ for any positive constant $C$, then, by Condition (C1) and assumption in Theorem 1, $|\epsilon_i|=o_p(\sqrt{p})$. Recall $\Lambda_n=o(1)$. Therefore, by Condition (C2)-(C3) and (\ref{errorfunc}), 
	\begin{equation}\label{eq:w}
		\underset{i}\sup\Vert\bm w_i\Vert_2=o_p(1).
	\end{equation}
	Recall $\bm{v}=(\bm v_1^\top,\ldots, \bm v_n^\top)^\top=\lambda_1\bm{G}\bm{\theta}^m$ and $\bm G=\bm I_n\otimes \bm G_0$, then 
	\[\bm v_i=\lambda_1\bm G_0\bm\theta_i^m=\lambda_1\bm G_0(\bm\theta_i^m-\bm\theta_i^0)+\lambda_1\bm G_0\bm\theta_i^0,\]
	and hence 
	\[\Vert\bm v_i\Vert_2\leq\lambda_1\Vert \bm G_0(\bm\theta_i^m-\bm\theta_i^0)\Vert_2+\lambda_1\Vert\bm G_0\bm\theta_i^0\Vert_2.\]
	
	By the definition of $\bm G_0$, we have
	\begin{eqnarray}\label{eq:v1}
		\Vert\bm G_0(\bm\theta_i^m-\bm\theta_i^0)\Vert_2
		&\leq&\Vert\bm G_0\Vert_F\Vert\bm\theta_i^m-\bm\theta_i^0\Vert_2\nonumber\\
		&\leq&\Lambda_n\sqrt{\sum_{l_1=1}^p\sum_{l_2=1}^p(\int_\chi B_{l_1}''(t)B_{l_2}''(t)\text{d}t})^2\nonumber\\
		&\leq&p\Lambda_n\underset{1\leq l\leq p}{\sup}\{\Vert B_l''\Vert^2_2\},
	\end{eqnarray}
	where the second inequality follows from (\ref{eq:Lambda_n}).
	Since 
	\[\bm G_0\bm\theta_i^0= (\int_\chi\sum_{l=1}^p B_1''(t)B_l''(t)\theta_{il}^0\text{d}t,\ldots,\int_\chi\sum_{l=1}^p B_p''(t)B_l''(t)\theta_{il}^0\text{d}t)^\top,\]
	then 
	\begin{eqnarray}\label{eq:v2}
		\Vert\bm G_0\bm\theta_i^0\Vert_2&= & \sqrt{(\int_\chi B_1''(t)s_{pi}''(t)\text{d}t)^2+\cdots+(\int_\chi B_p''(t)s_{pi}''(t)\text{d}t)^2}\nonumber\\
		&\leq& \sqrt{\int_\chi B_1''(t)^2dt\int_\chi s_{i}''(t)^2\text{d}t+\cdots+\int_\chi B_p''(t)^2\text{d}t\int_\chi s_{i}''(t)^2\text{d}t}\nonumber\\
		&\leq& \sqrt{p\underset{l}{\sup}\{\int_\chi B_l''(t)^2\text{d}t\}\int_\chi s_{i}''(t)^2\text{d}t}\nonumber\\
		&=& \sqrt{p\underset{l}{\sup}\{\Vert B_l''\Vert_2^2\}\Vert \beta^0_i{''}(t)-e_{i2}(t)\Vert_2^2}\nonumber\\
		&& 
	\end{eqnarray}
	Therefore, with (\ref{eq:v1}) and (\ref{eq:v2}), we have
	\begin{equation*}
		\Vert \bm v_i\Vert_2\leq\lambda_1\underset{l}{\sup}\{\Vert B_l''(t)\Vert_2\}(\sqrt{p}\Vert\beta^0_i{''}(t)-e_{i2}(t)\Vert_2+p\Lambda_n\underset{l}{\sup}\{\Vert B_l''(t)\Vert_2\}).
	\end{equation*}
	According to the derivative of B-spline function given in \cite{Boor1978}, we have $\underset{l}{\sup}\{\Vert B_l''\Vert_2^2\}\sim p^3$. Moreover, by (\ref{errorfunc}), $\Vert\beta^0_i{''}-e_{i2}\Vert_2\sim \Vert\beta^0_i{''}\Vert_2$. Recall $\Lambda_n= o(1)$. Consequently, under Conditions (C2)-(C3), 
	\begin{equation}\label{eq:v}
		\underset{i}{\sup}\Vert\bm v_i\Vert_2= o(\lambda_1p^{4}).
	\end{equation}
	With (\ref{eq:w}) and (\ref{eq:v}), we have 
	
	$$\frac{\underset{i}{\sup}\Vert\bm w_i\Vert_2+\underset{i}{\sup}\Vert\bm v_i\Vert_2}{|\mathcal{G}_{\min}|}=o(\frac{\lambda_1p^{4}}{|\mathcal{G}_{\min}|})$$
	with probability approaching 1. Recall (C4) and $\lambda_2\gg\frac{n^{-(1-\sigma_0)/2} p^{4}}{|\mathcal{G}_{\min}|}$. Therefore, we have $\lambda_2\gg\frac{\underset{i}{\sup}\Vert\bm w_i\Vert_2+\underset{i}{\sup}\Vert\bm v_i\Vert_2}{|\mathcal{G}_{\min}|}$ with probability approaching 1.
\end{proof}

\subsection*{C. Tables and Figures}
\setcounter{table}{0} 
\renewcommand{\thetable}{C\arabic{table}}
\setcounter{figure}{0} 
\renewcommand{\thefigure}{C\arabic{figure}}
\begin{figure}[H]
	\centering    
	\subfigure[]
	{
		\includegraphics[width=3.8cm]{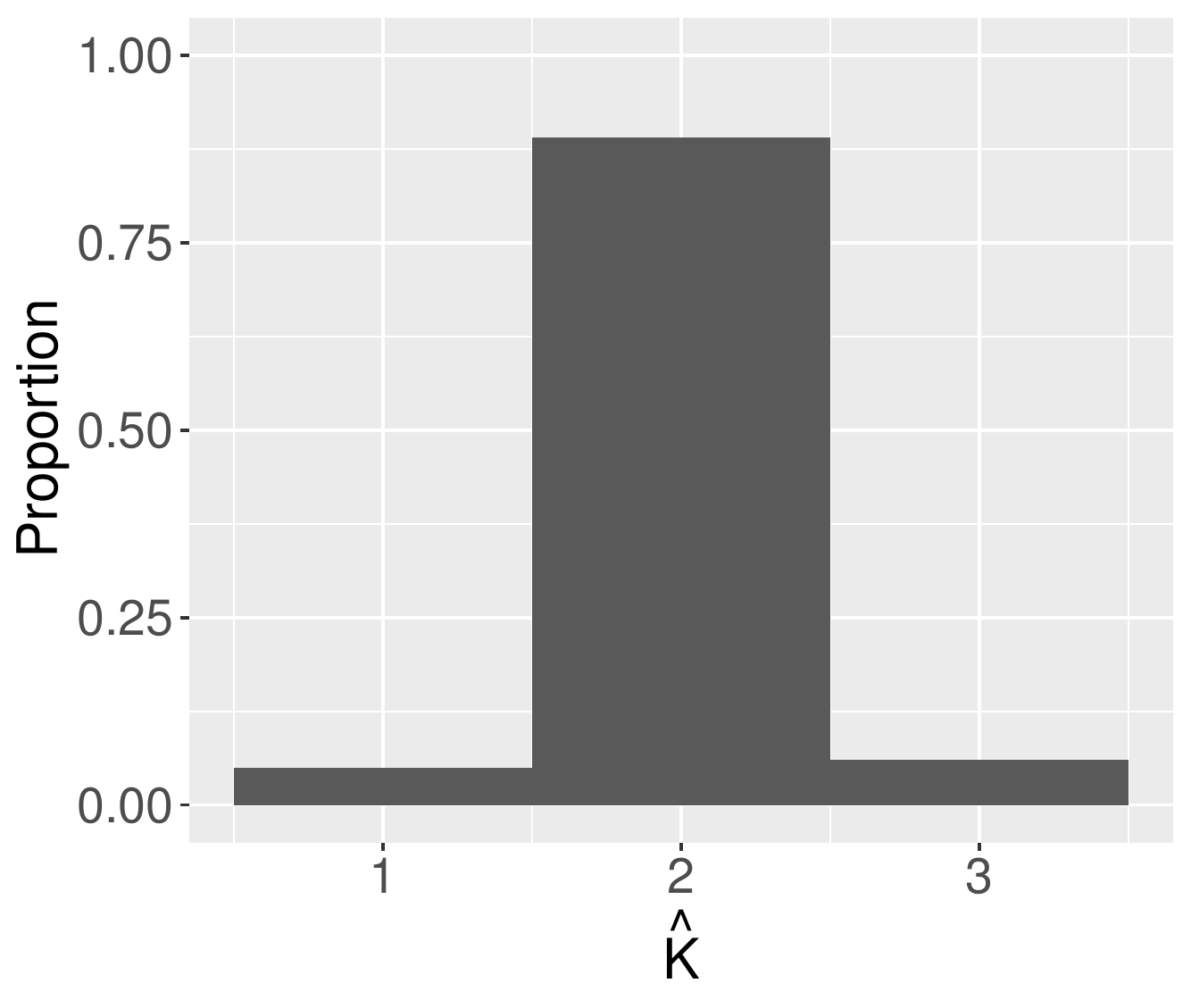}
	}
	\subfigure[]
	{
		\includegraphics[width=3.8cm]{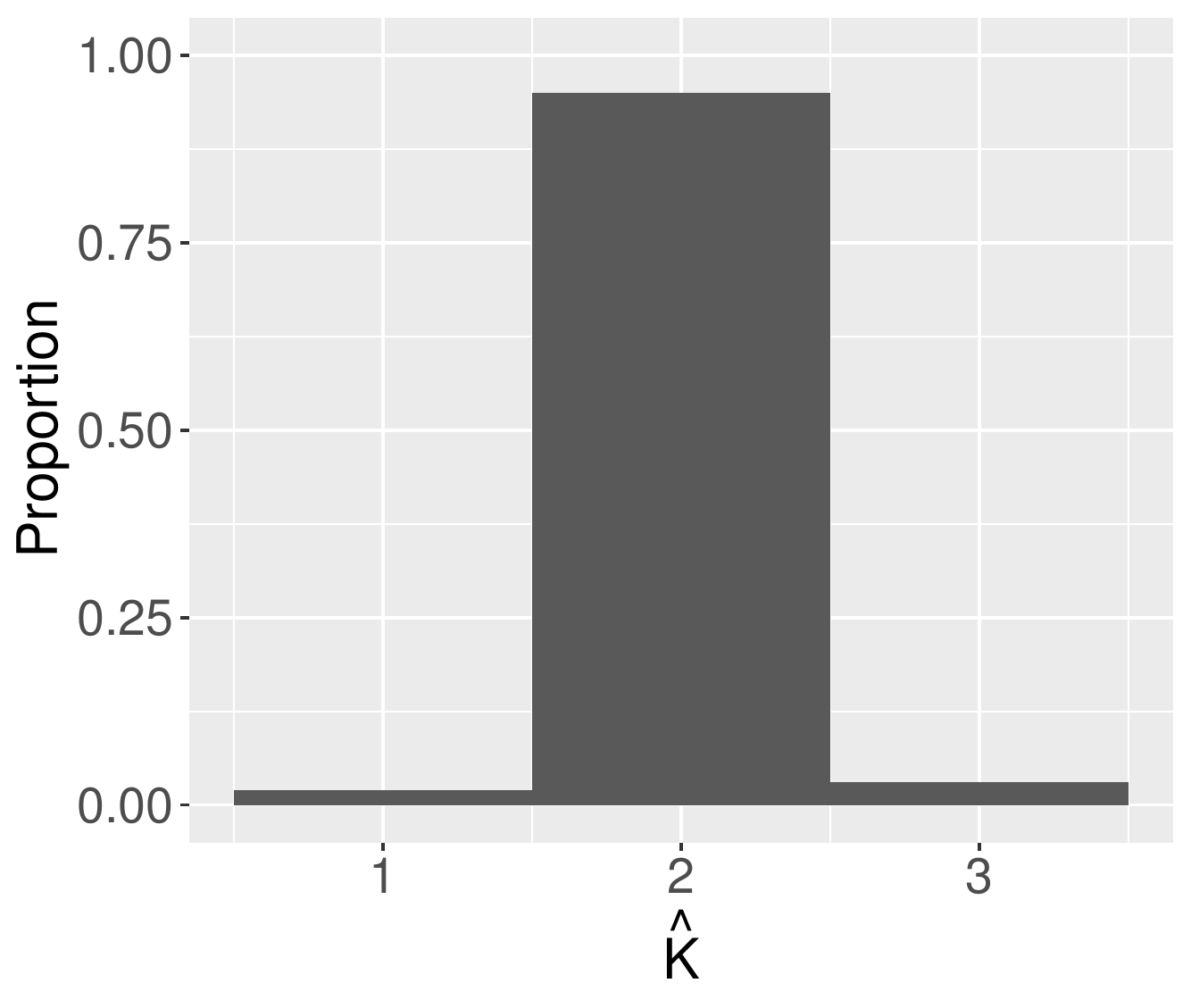} 
	}
	\subfigure[]
	{
		\includegraphics[width=3.8cm]{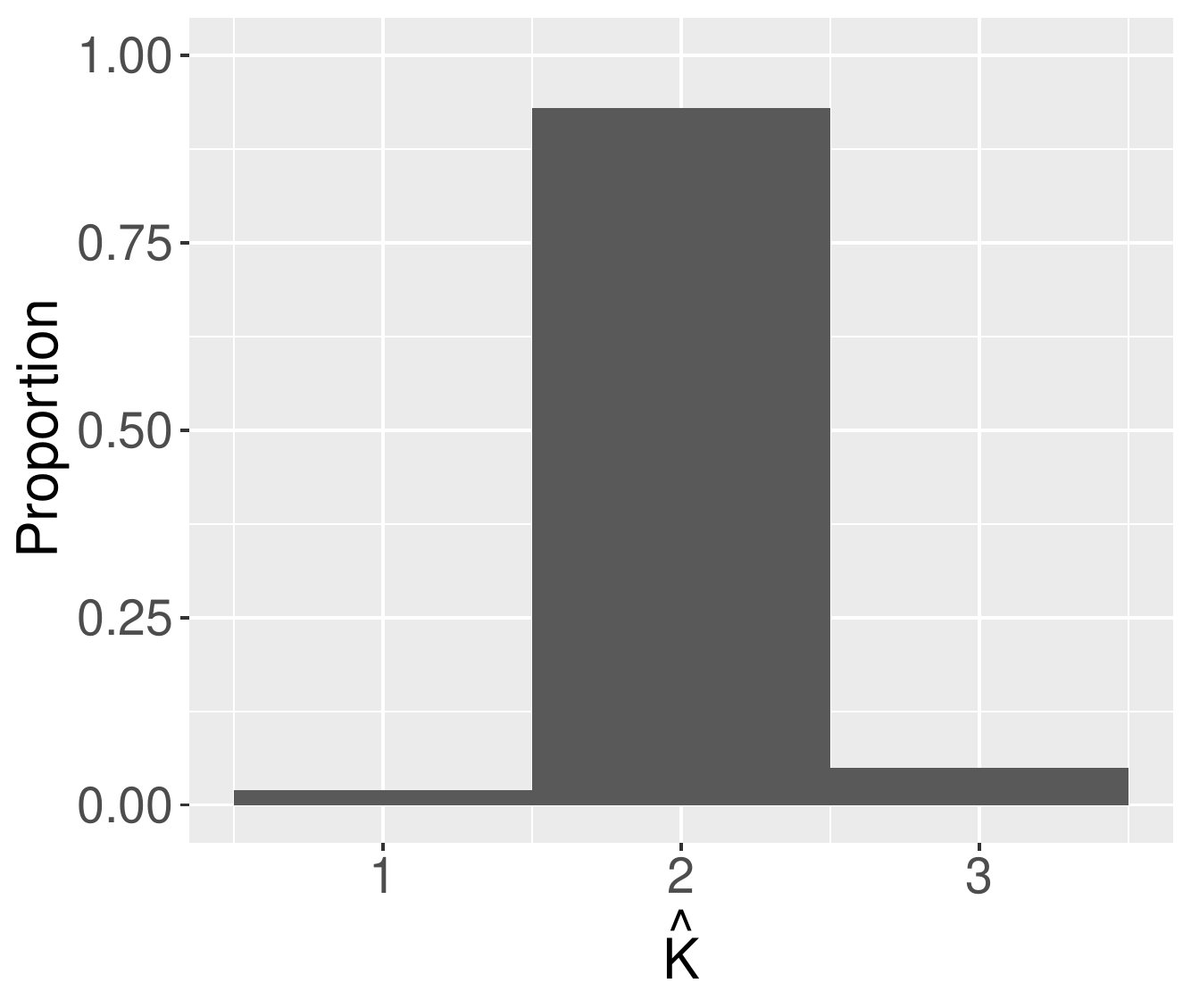}
	}
	\subfigure[]
	{
		\includegraphics[width=3.8cm]{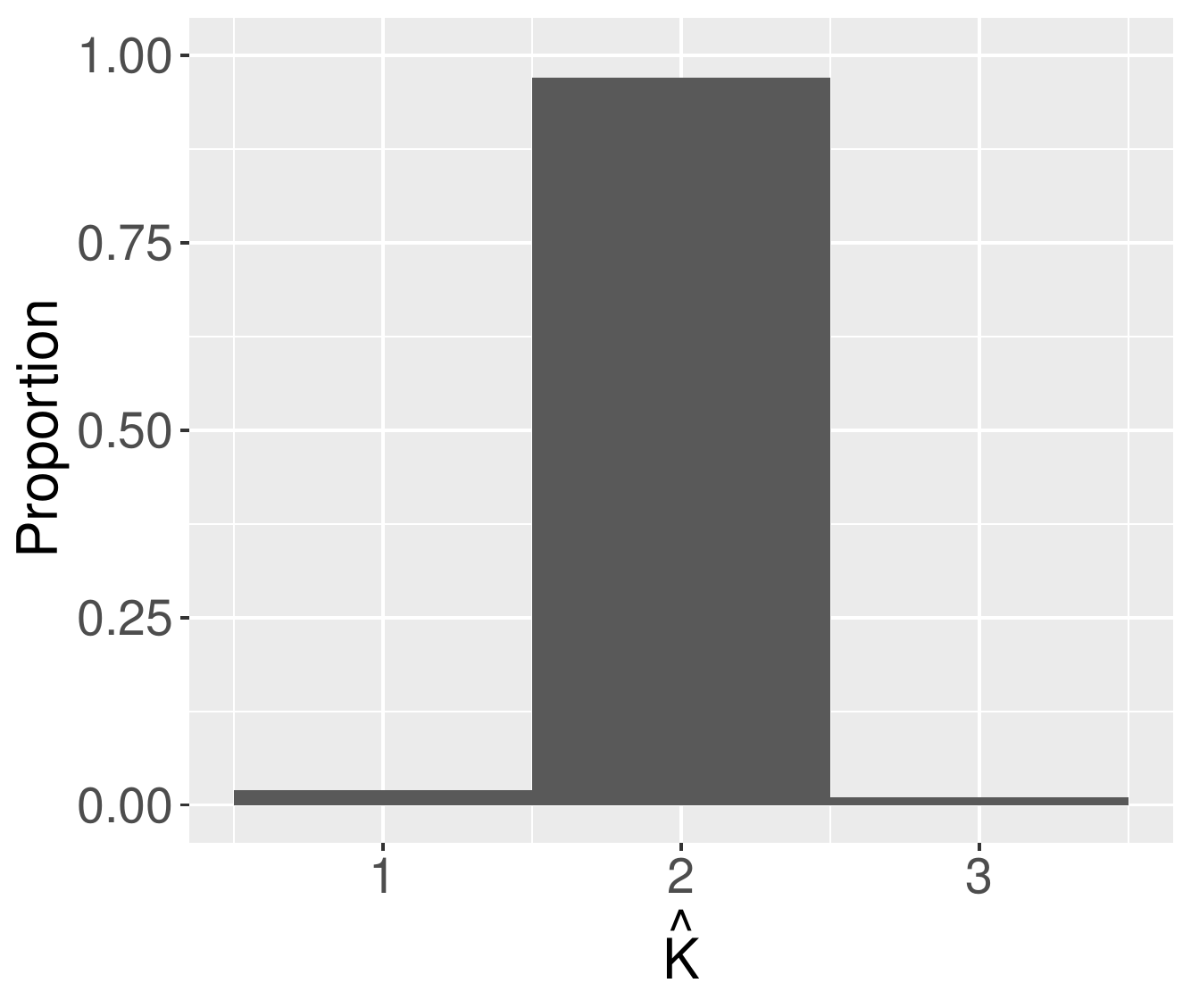} 
	}
	\subfigure[]
	{
		\includegraphics[width=3.8cm]{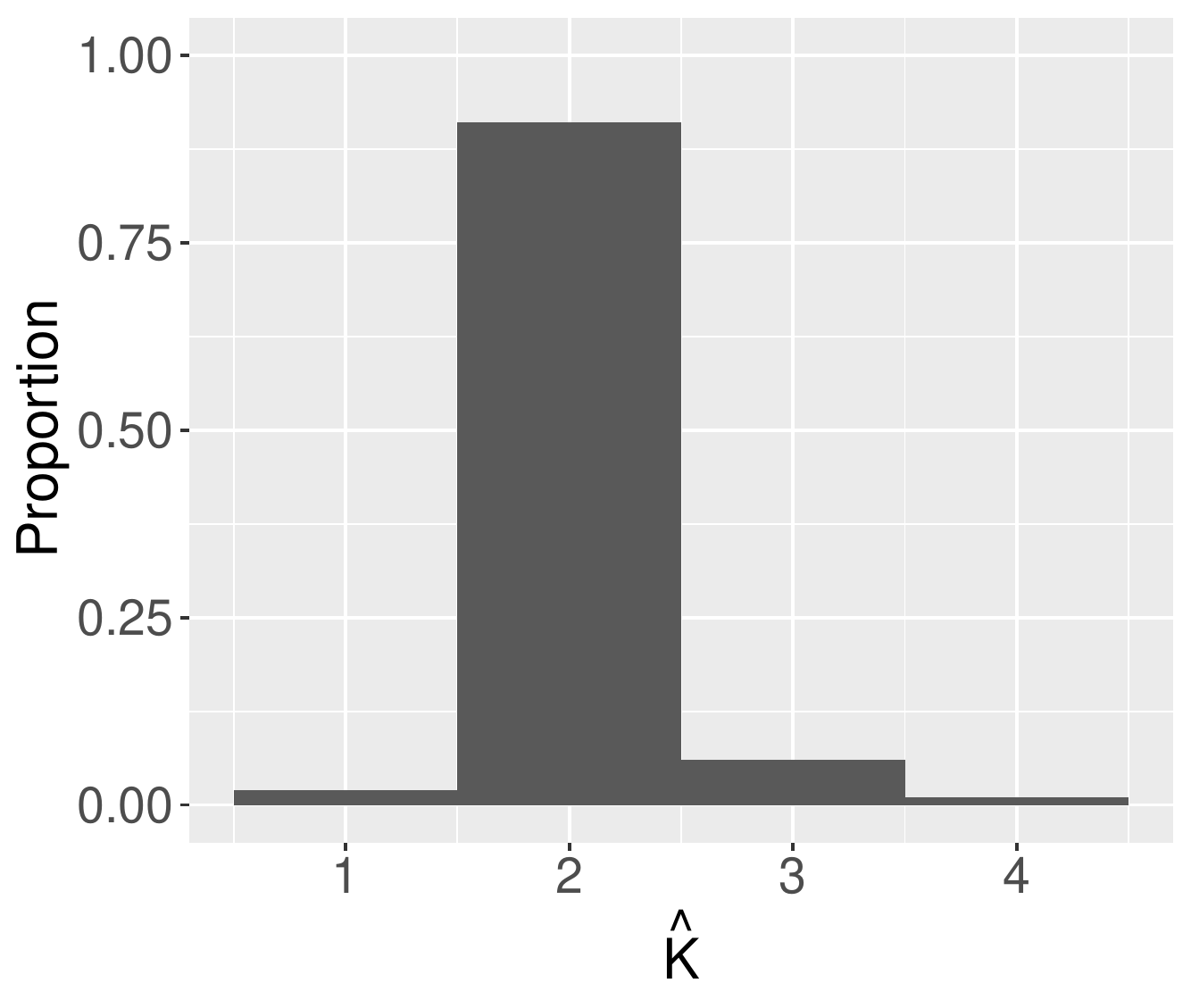}
	}
	\subfigure[]
	{
		\includegraphics[width=3.8cm]{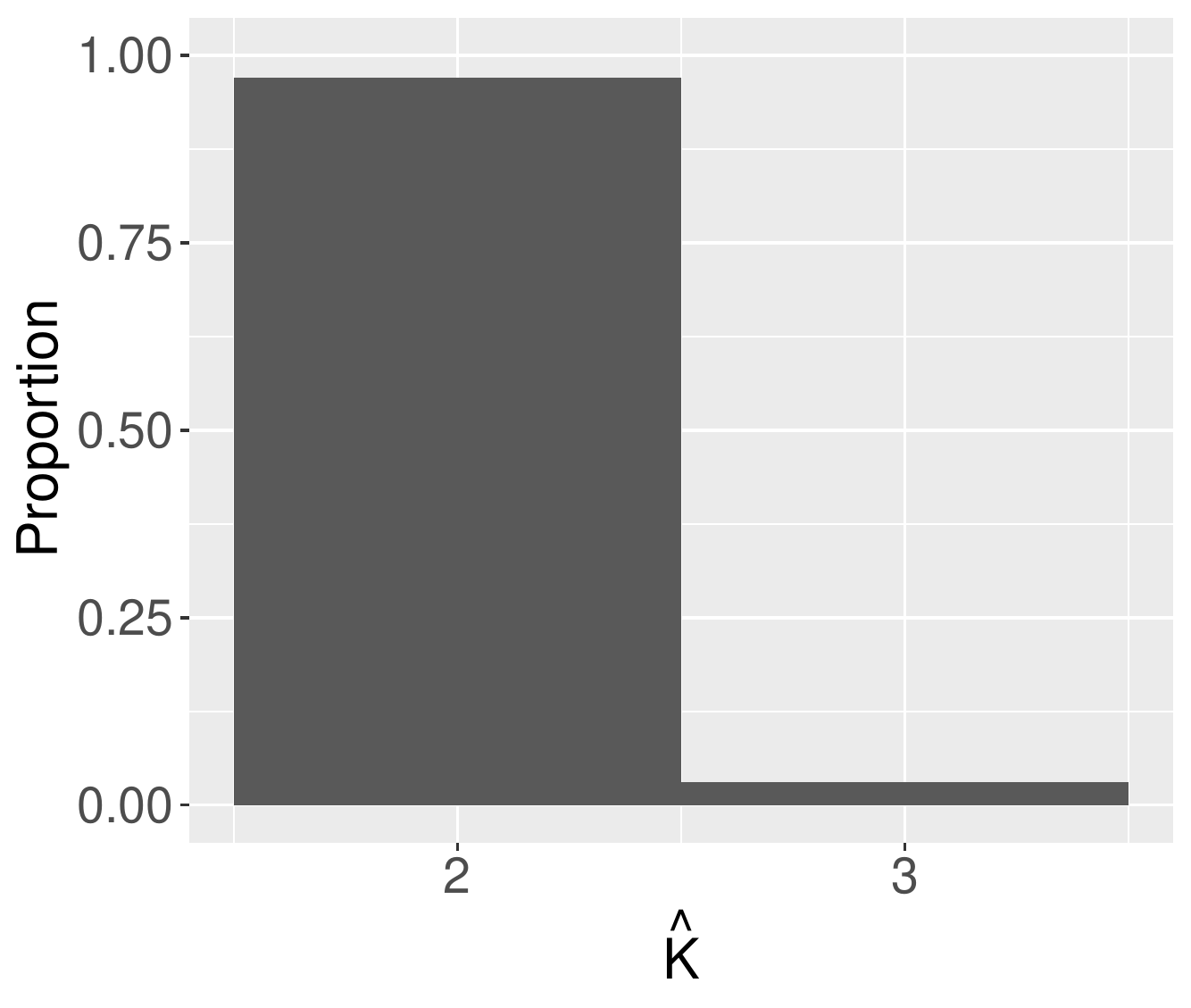} 
	}
	\subfigure[]
	{
		\includegraphics[width=3.8cm]{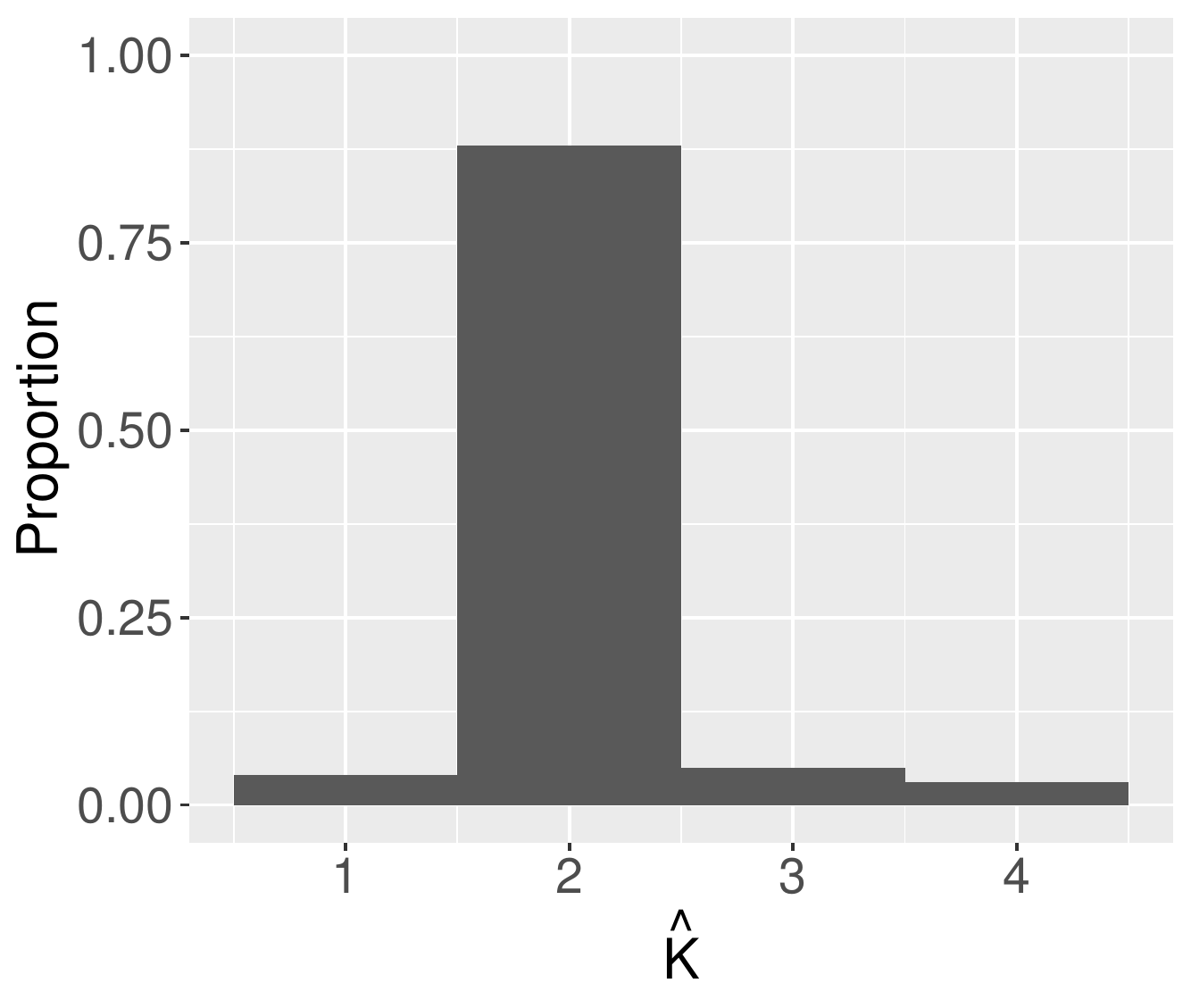}
	}
	\subfigure[]
	{
		\includegraphics[width=3.8cm]{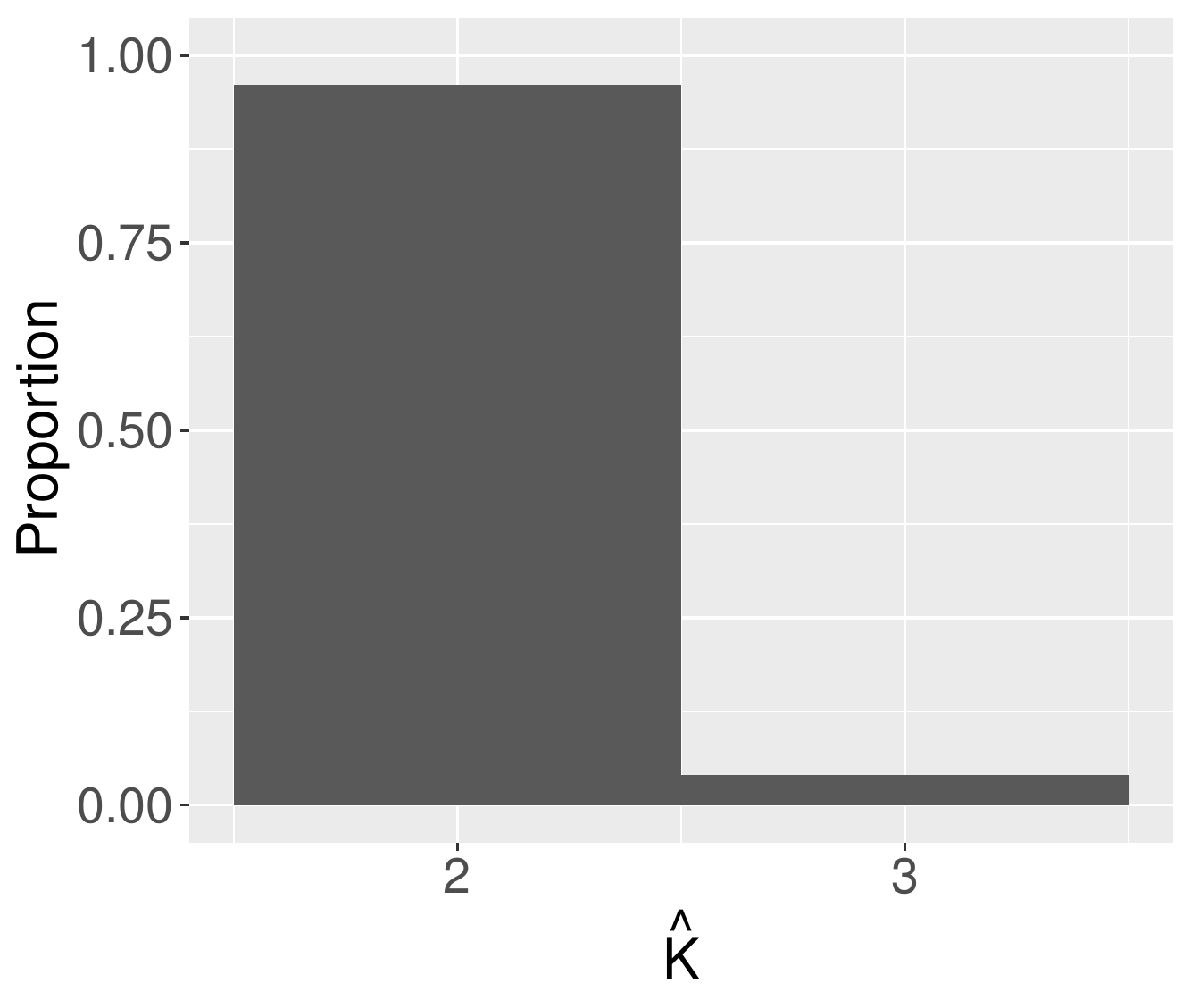} 
	}	
	\caption{Simulation results in Example 1: histograms of the estimated number of subgroups $\hat{K}$ under Scenario 2. (a) and (e) balanced structure with $n=40$, (b) and (f) balanced structure with $n=200$, (c) and (g) unbalanced structure with $n=40$, (d) and (h) unbalanced structure with $n=200$. (a)-(d) $a_{il}\sim N(2,1)$, and (b)-(h) $a_{il}\sim U(0,4)$.} 
	\label{fig:ex2}
\end{figure}

Figure \ref{fig:ex2} shows the distribution of the number of identified subgroups $\hat{K}$ for Scenario 2 in Example 1. The proposed approach can satisfactorily identify the number of true subgroups for all settings. As $n$ increases, the percentage of correctly determining the number of subgroups becomes larger.

\begin{figure}[H]
	\centering    
	\subfigure[]
	{
		\includegraphics[width=3.8cm]{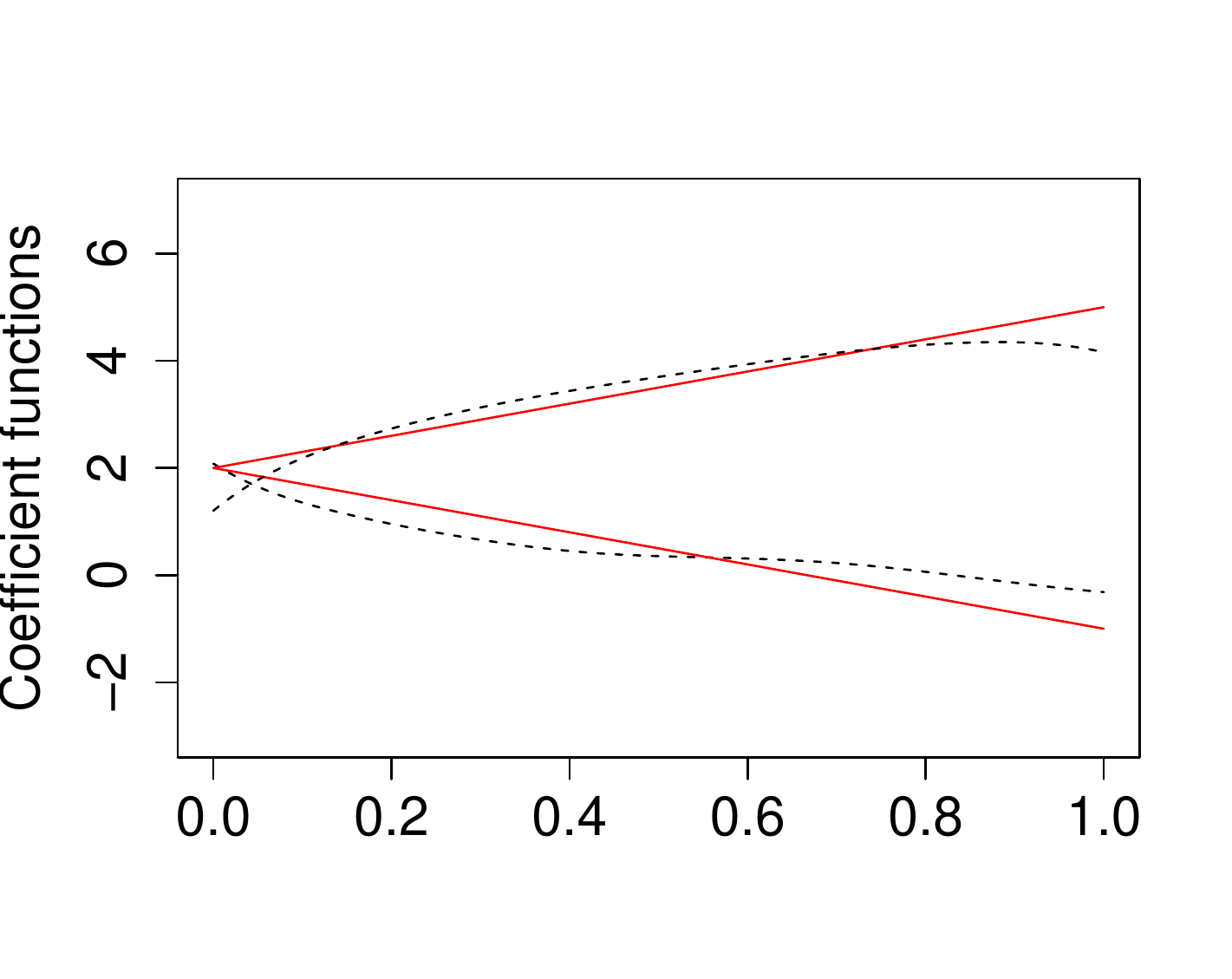}
	}
	\subfigure[]
	{
			\includegraphics[width=3.8cm]{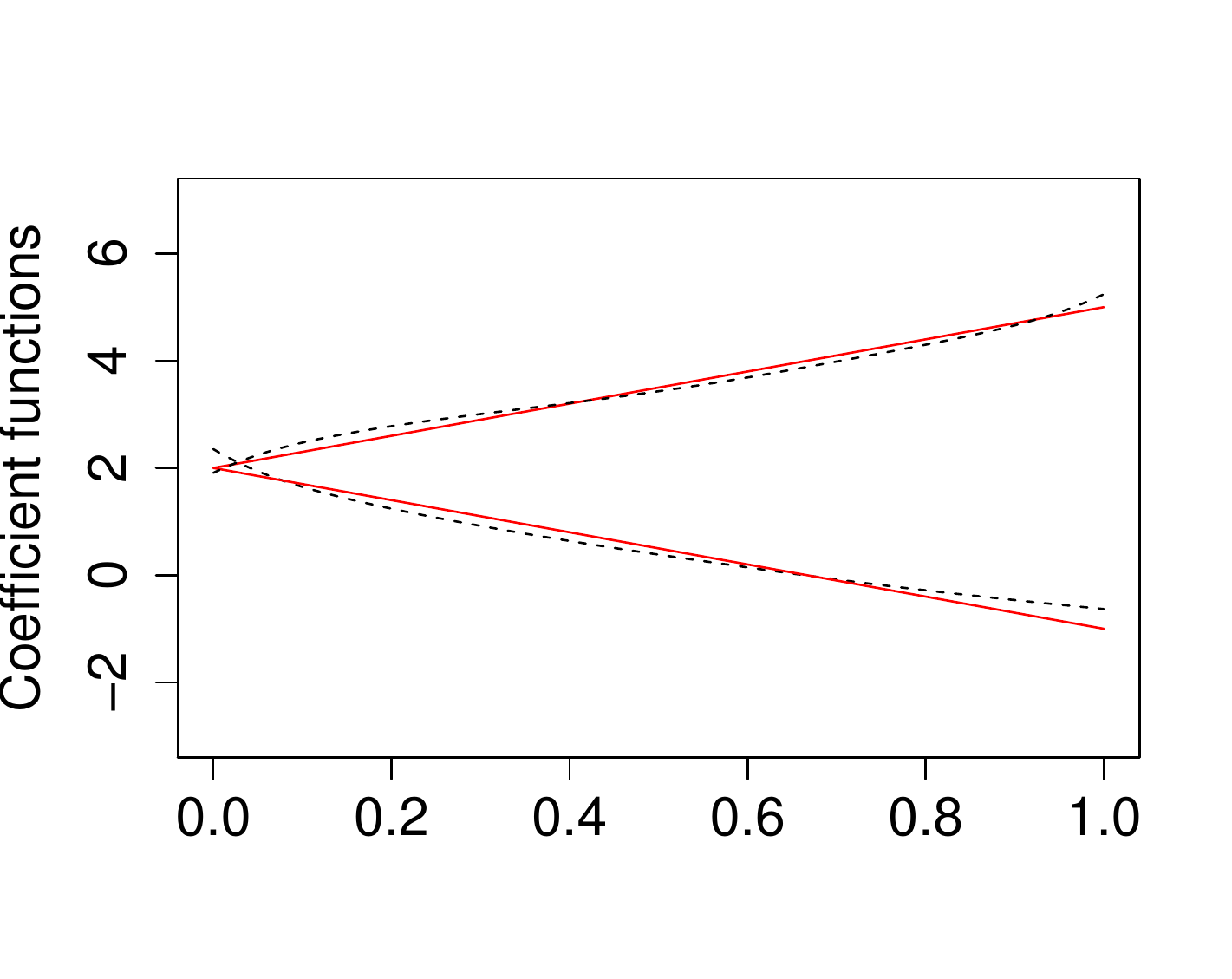} 
		}
		\subfigure[]
		{
			\includegraphics[width=3.8cm]{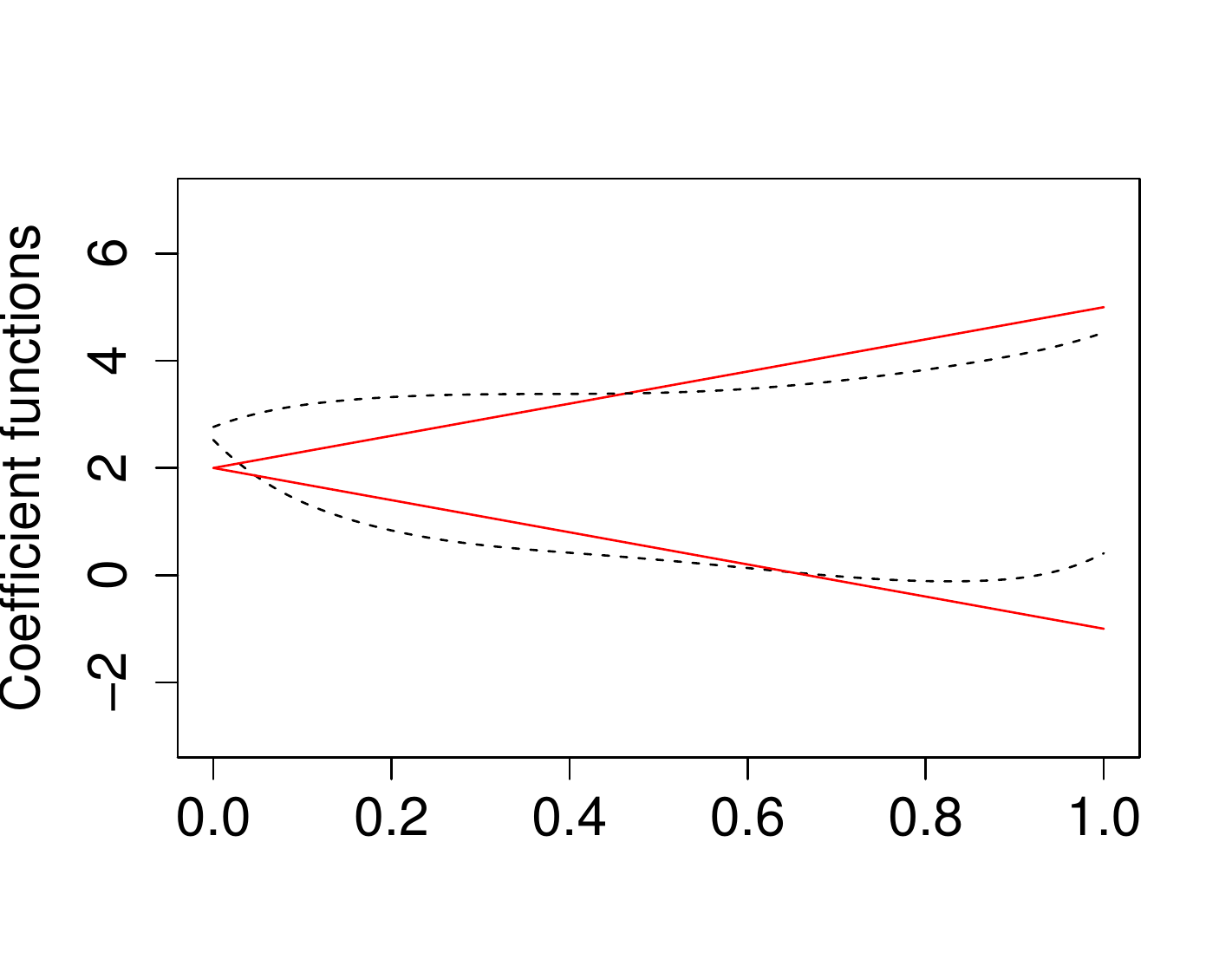}
		}
		\subfigure[]
		{
			\includegraphics[width=3.8cm]{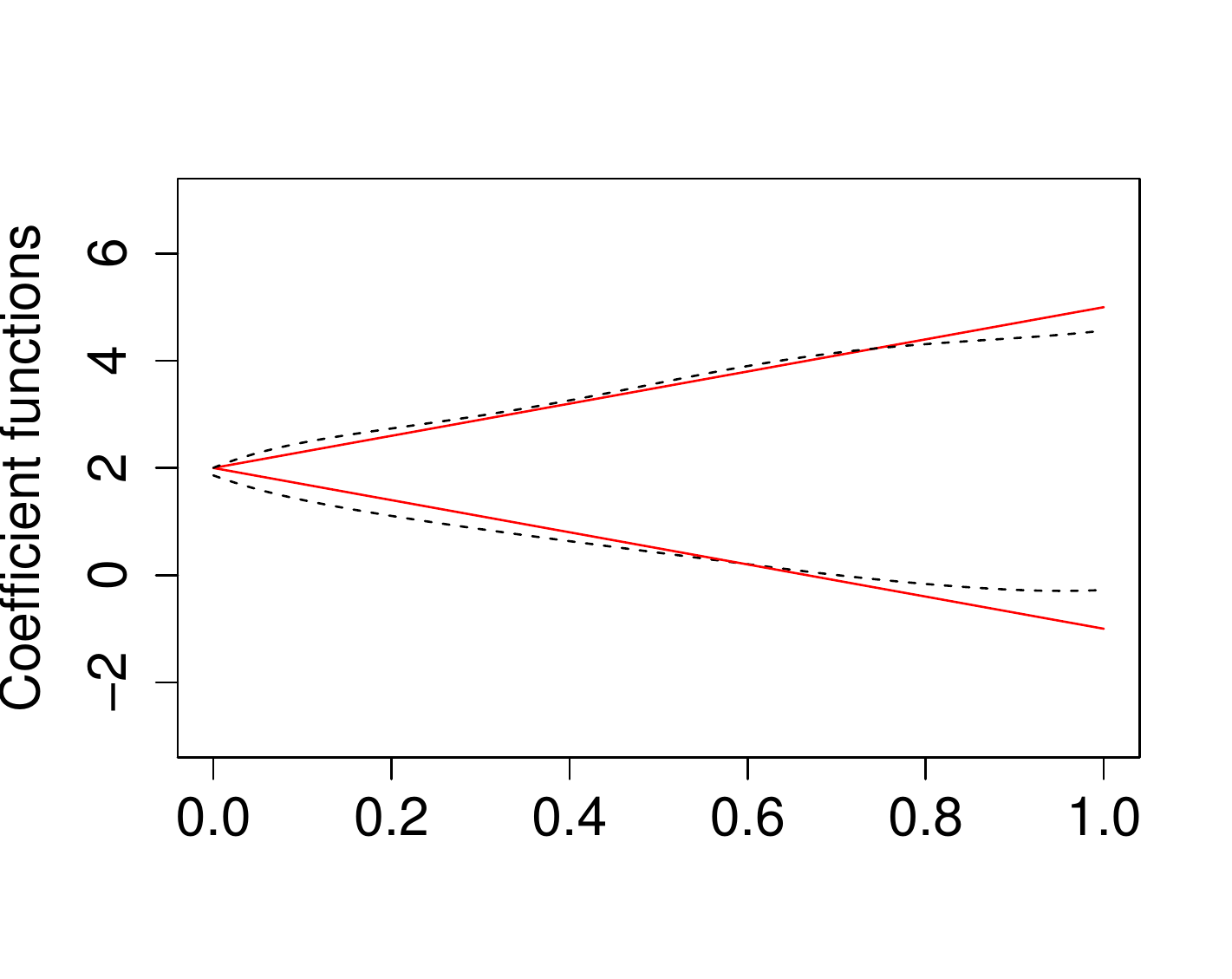} 
		}
		\subfigure[]
		{
			\includegraphics[width=3.8cm]{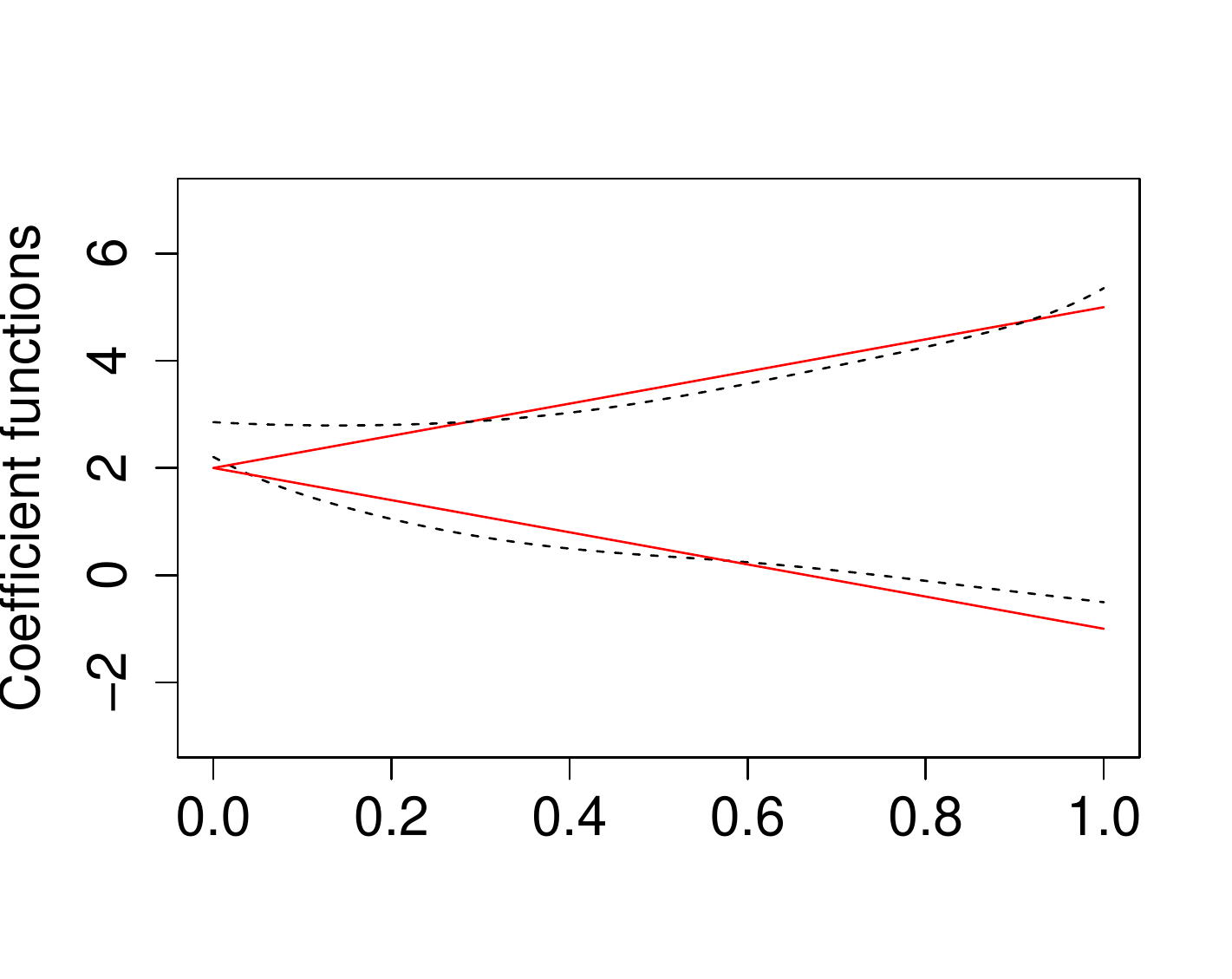}
		}
		\subfigure[]
		{
			\includegraphics[width=3.8cm]{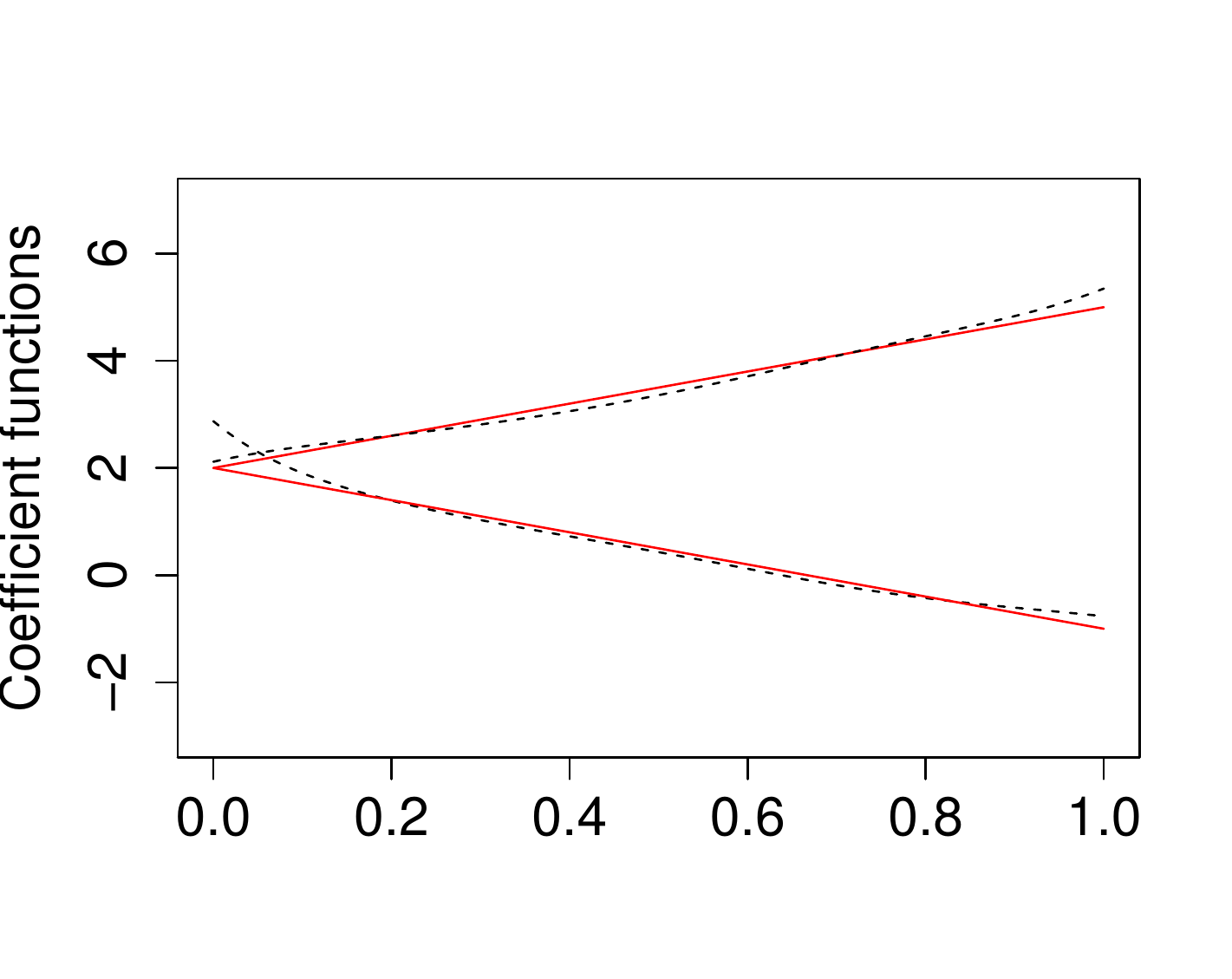} 
		}
		\subfigure[]
		{
			\includegraphics[width=3.8cm]{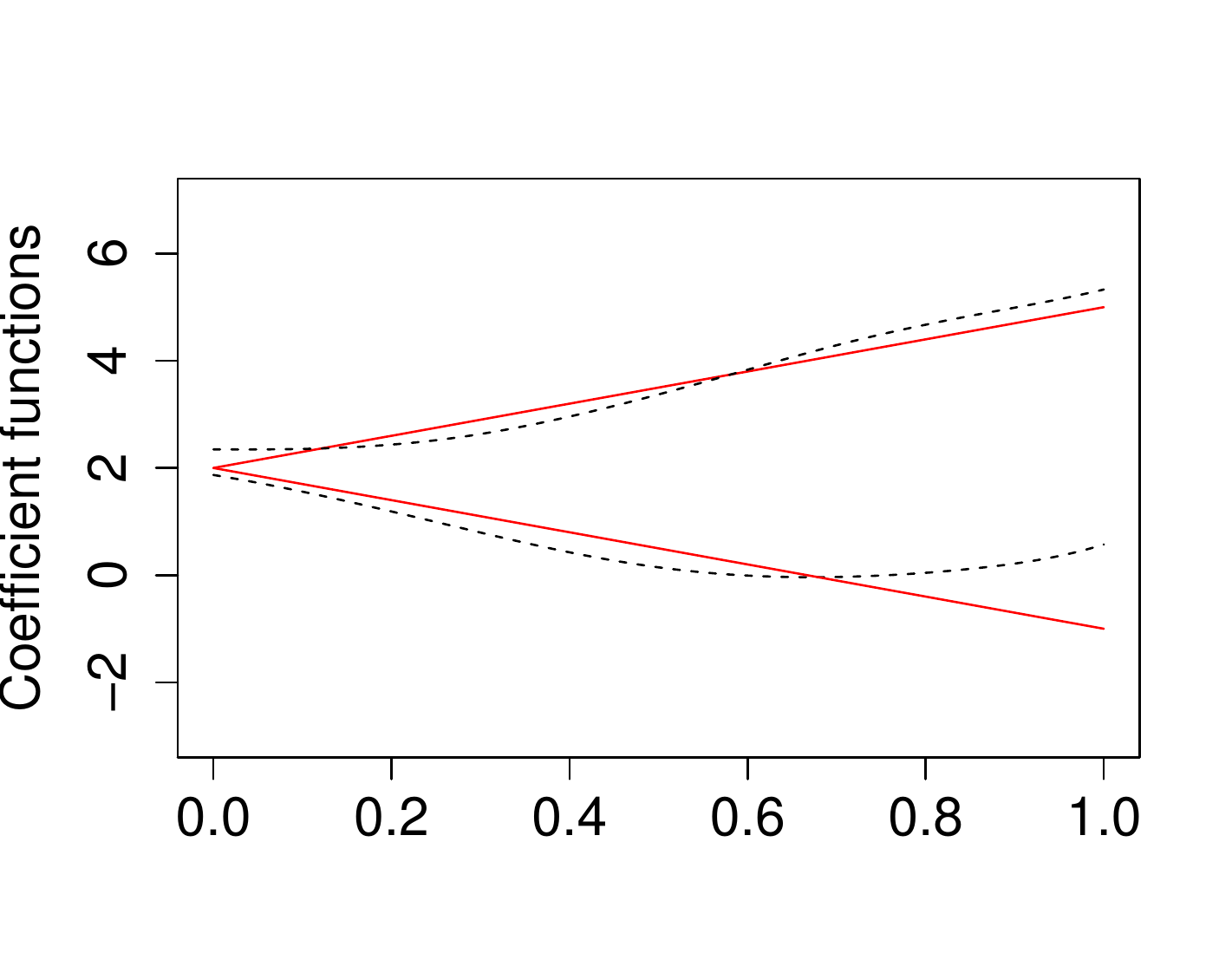}
		}
		\subfigure[]
		{
			\includegraphics[width=3.8cm]{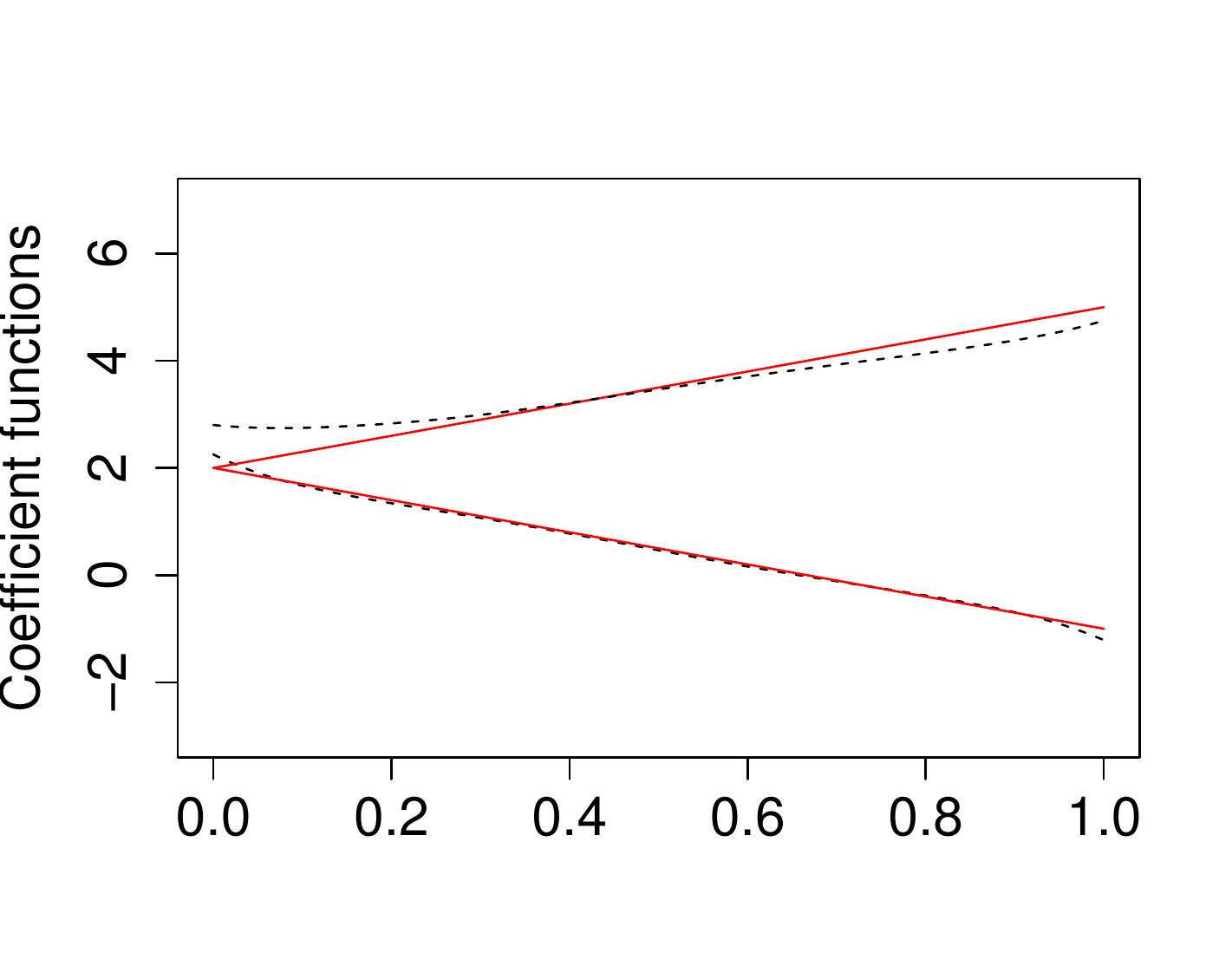} 
		}	
		\caption{Simulation results in Example 1: estimation of the coefficient functions under Scenario 2. (a) and (e) balanced structure with $n=40$, (b) and (f) balanced structure with $n=200$, (c) and (g) unbalanced structure with $n=40$, (d) and (h) unbalanced structure with $n=200$. (a)-(d) $a_{il}\sim N(2,1)$, and (b)-(h) $a_{il}\sim U(0,4)$. Red solid lines represent the true coefficient functions, and black dashed lines represent the estimated coefficient functions.} 
		\label{fig:coef2}
	\end{figure}
	
	Figure \ref{fig:coef2} presents estimations of the coefficient functions for Scenario 2 in Example 1. It can be observed that the underlying true coefficient function curves can be recovered well. 
	
	\begin{figure}[H]
		\centering    
		\subfigure[]
		{
			\includegraphics[width=3.8cm]{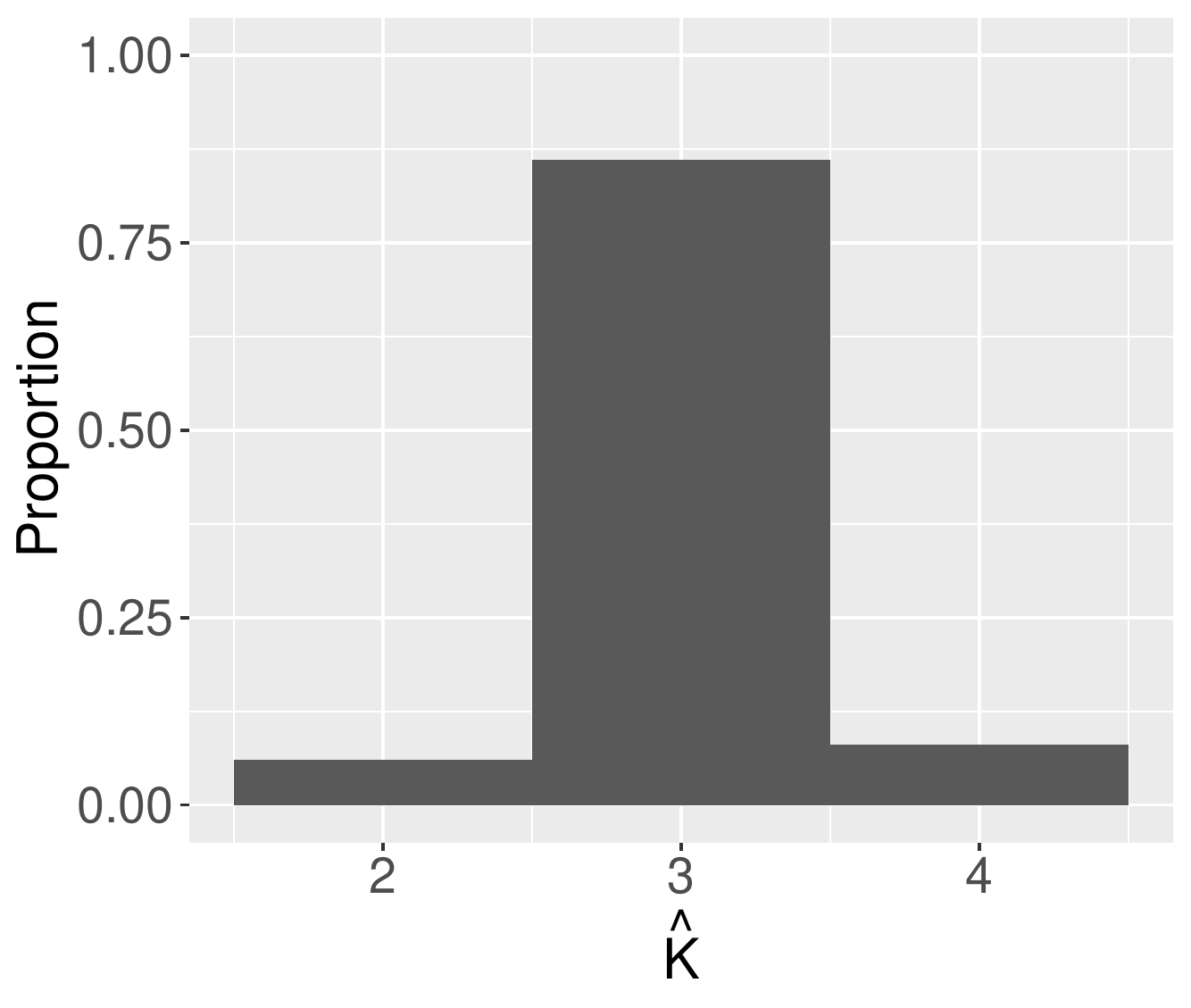}
		}
		\subfigure[]
		{
			\includegraphics[width=3.8cm]{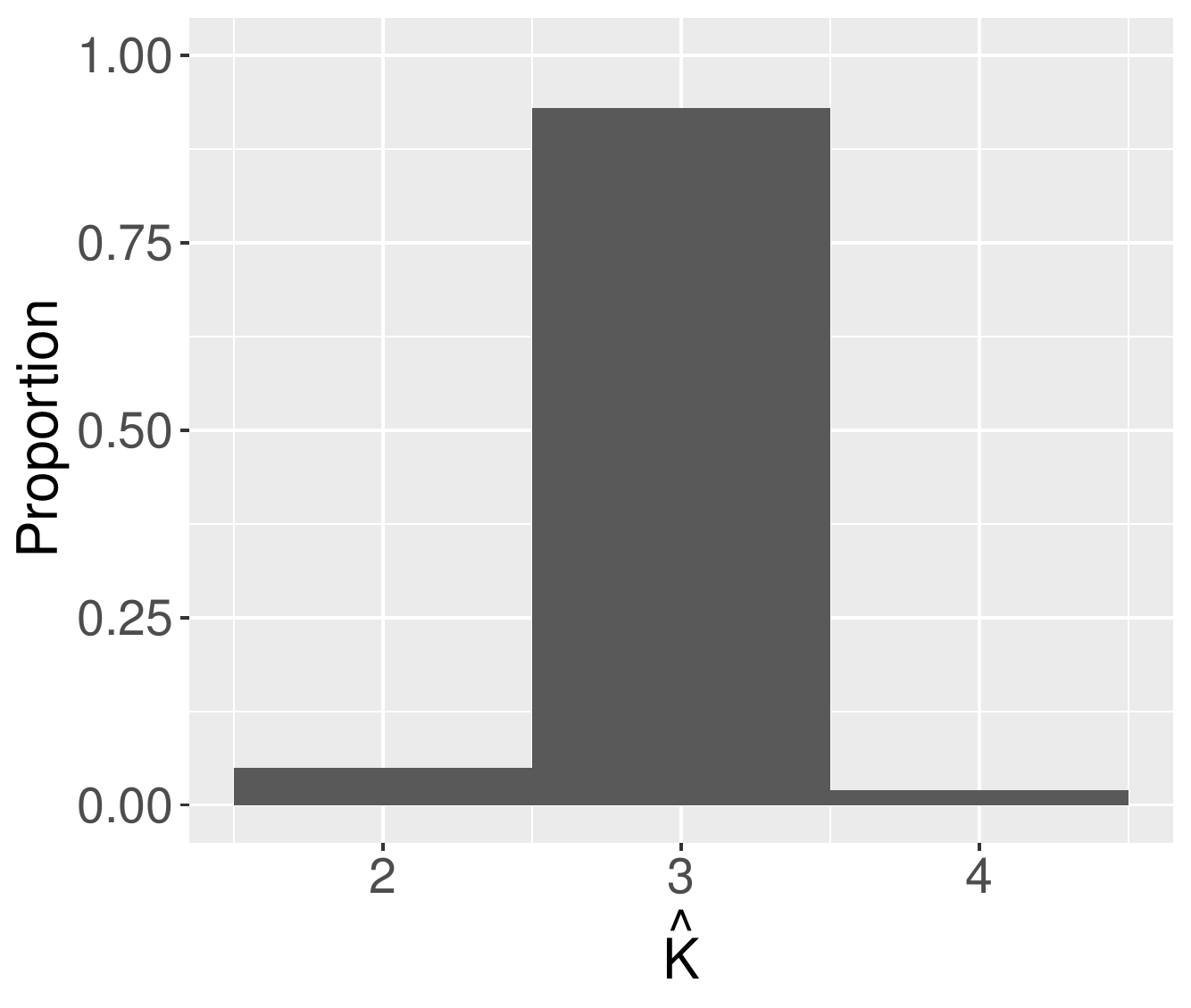} 
		}
		\subfigure[]
		{
			\includegraphics[width=3.8cm]{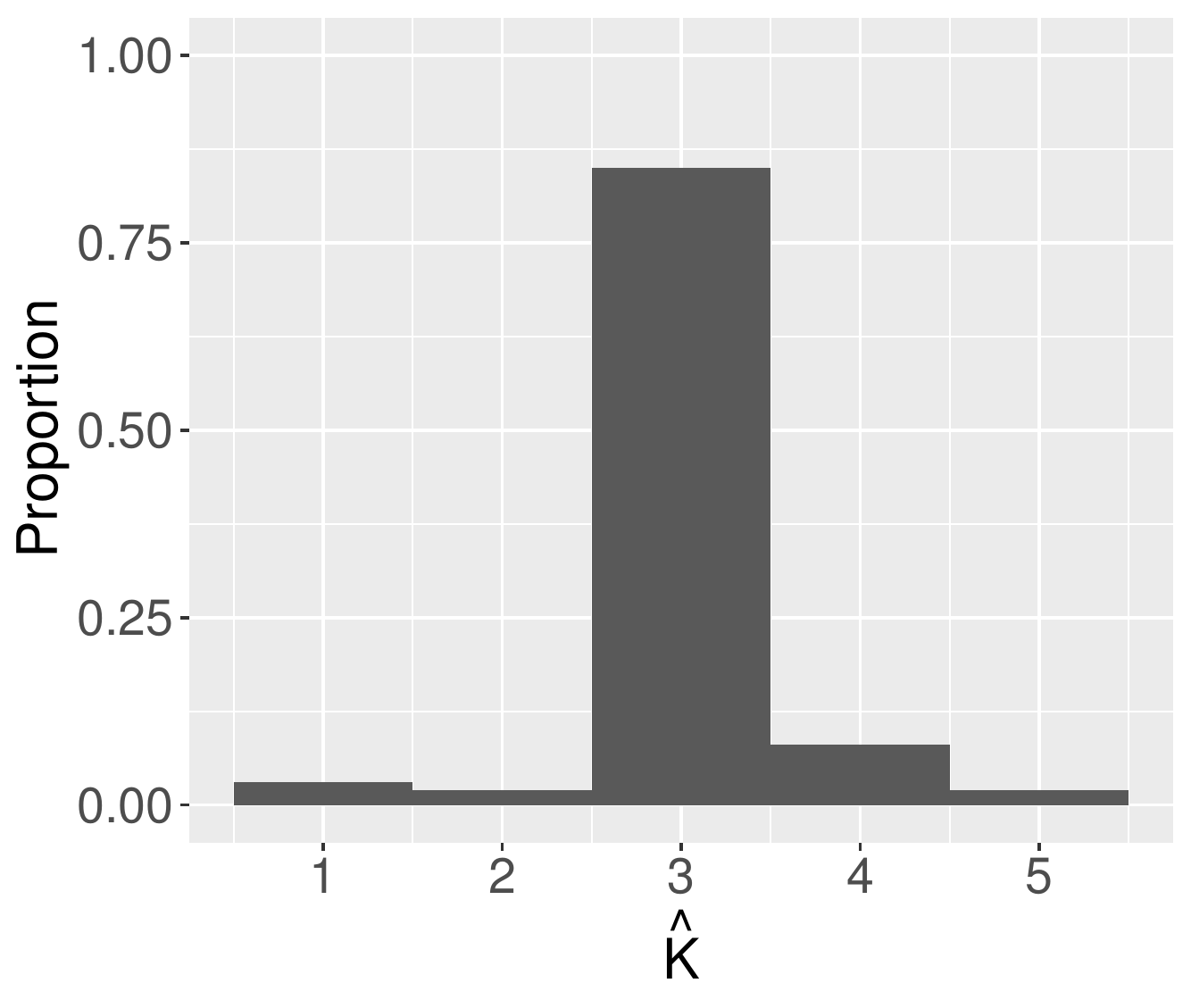}
		}
		\subfigure[] 
		{
			\includegraphics[width=3.8cm]{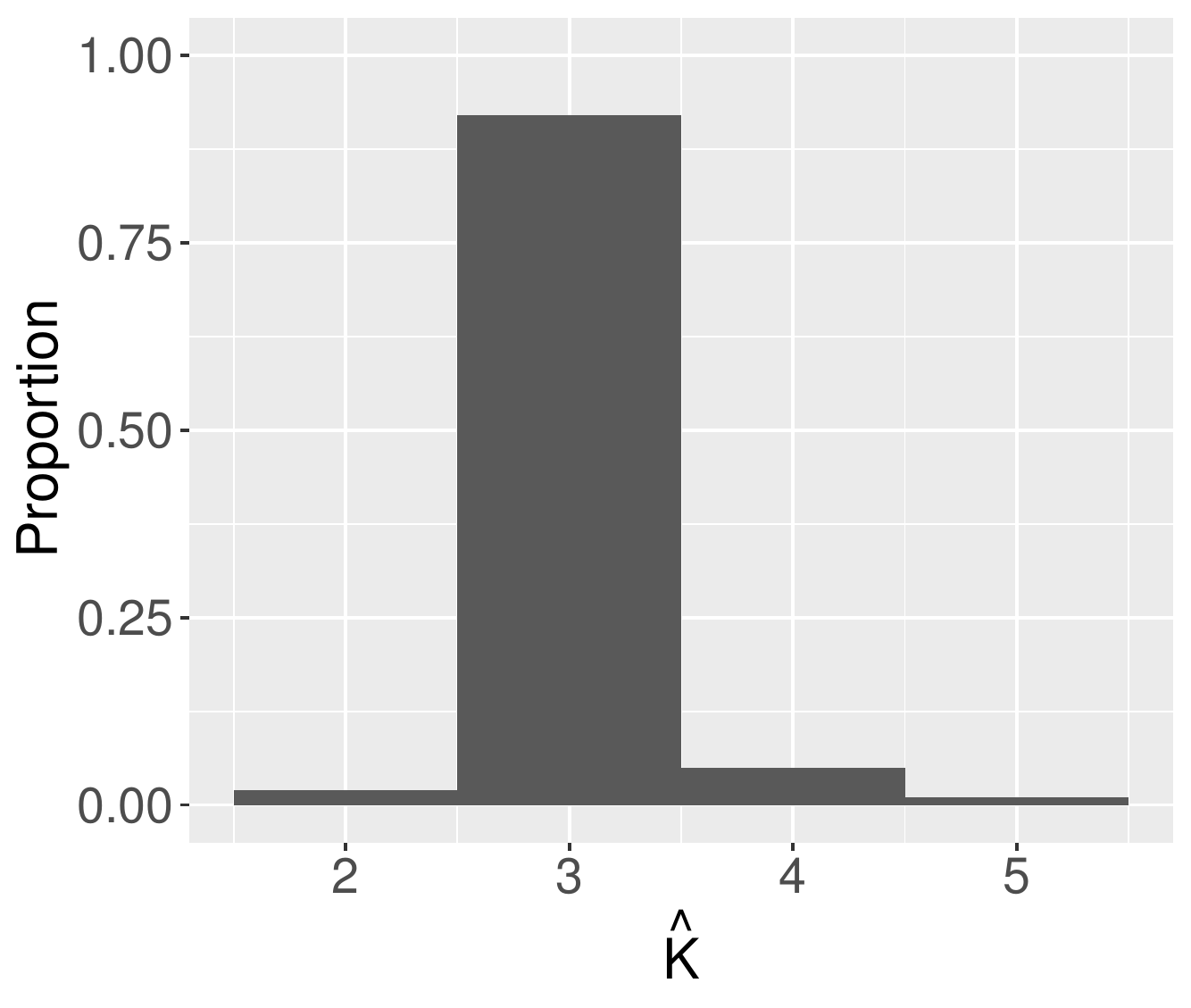} 
		}
		\subfigure[]
		{
			\includegraphics[width=3.8cm]{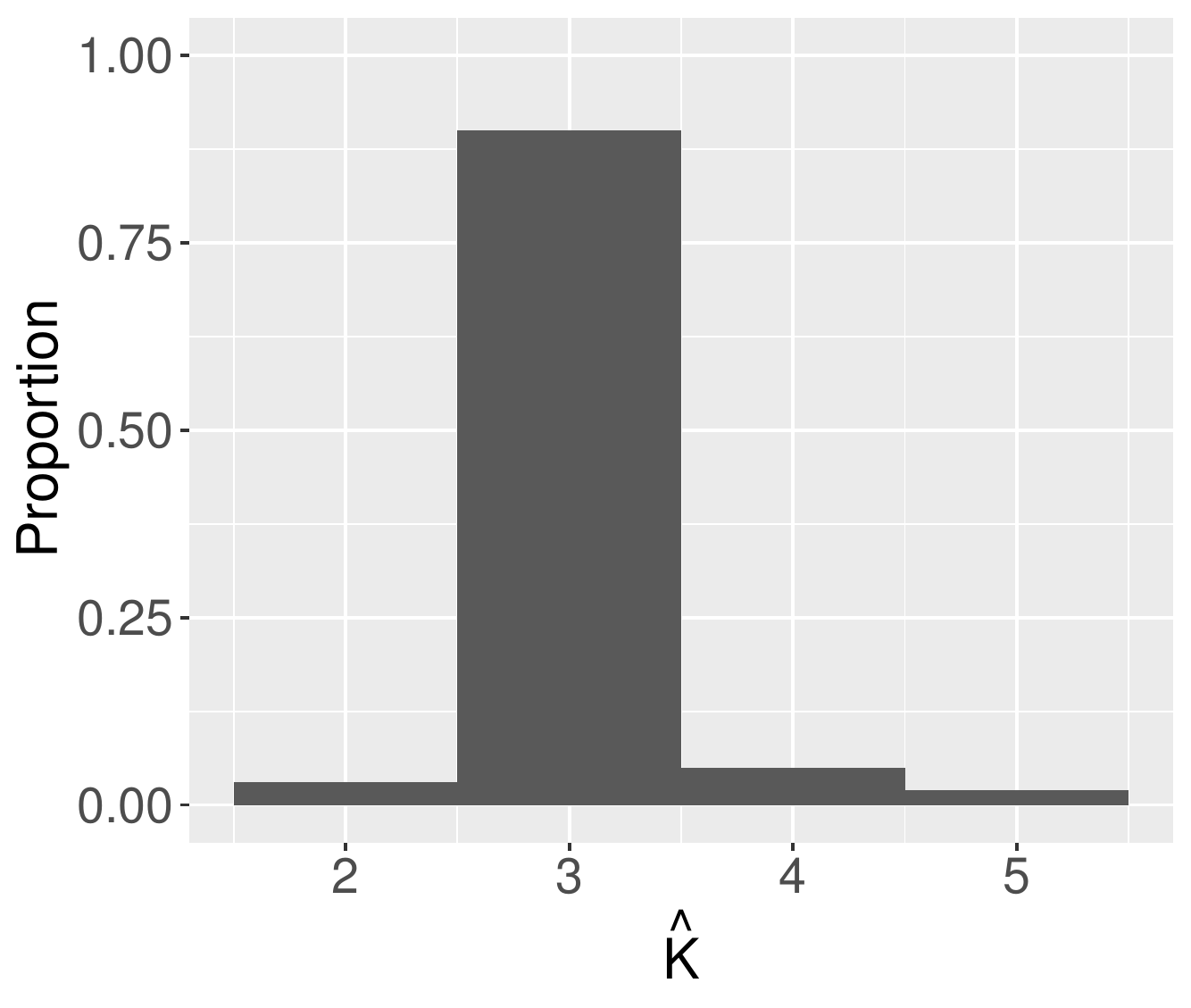}
		}
		\subfigure[]
		{
			\includegraphics[width=3.8cm]{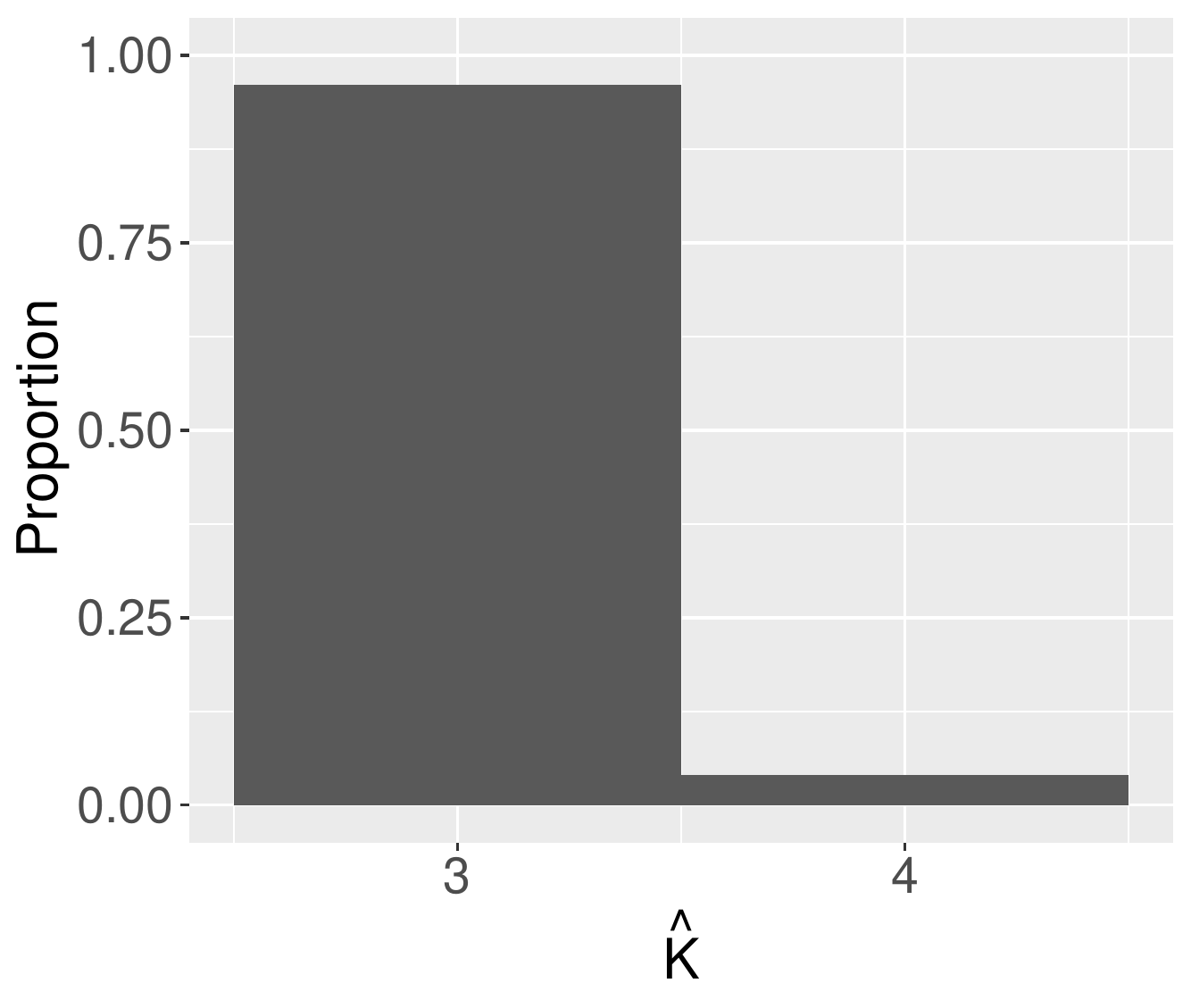} 
		}
		\subfigure[]
		{
			\includegraphics[width=3.8cm]{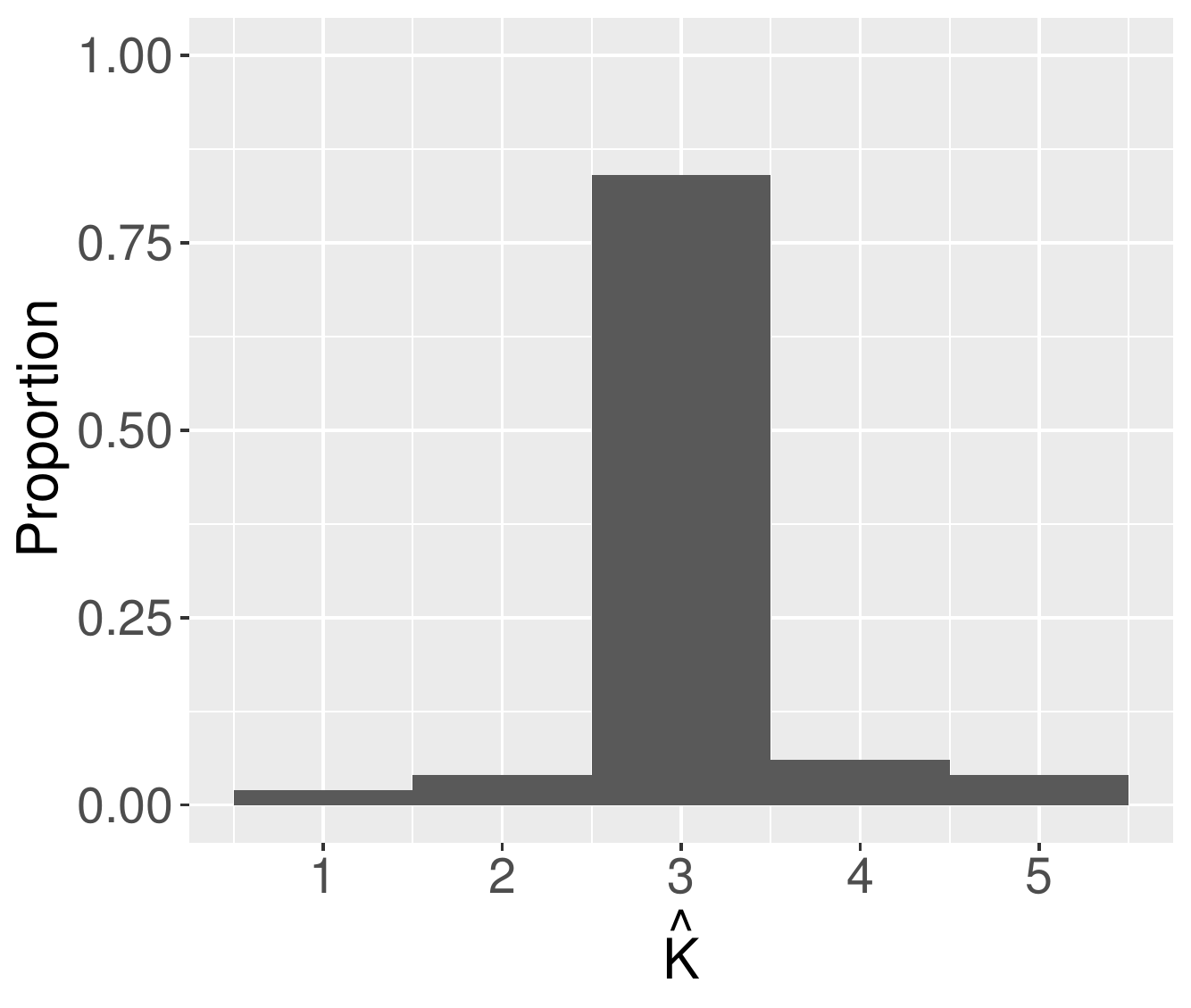}
		}
		\subfigure[]
		{
			\includegraphics[width=3.8cm]{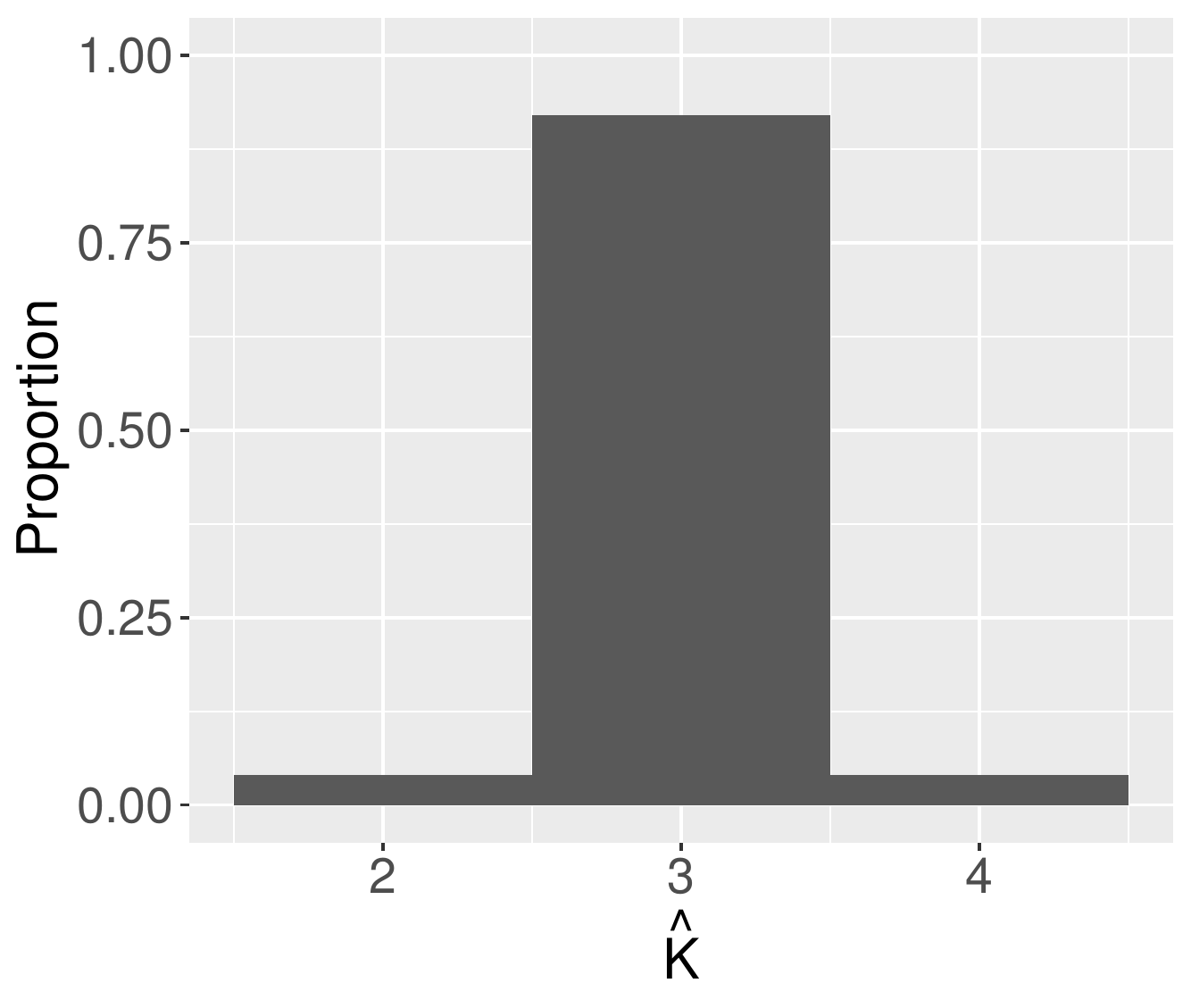} 
		}	
		\caption{Simulation results in Example 2: histograms of the estimated number of subgroups $\hat{K}$ by the proposed approach. (a) and (e) balanced structure with $n=60$, (b) and (f) balanced structure with $n=300$, (c) and (g) unbalanced structure with $n=60$, (d) and (h) unbalanced structure with $n=300$. (a)-(d) $a_{il}\sim N(2,1)$, and (b)-(h) $a_{il}\sim U(0,4)$.} 
		\label{fig:ex3}
	\end{figure}
	
	Figure \ref{fig:ex3} illustrates the distribution of the estimated number of subgroups by the proposed approach. We obtain similar patterns as those in Example 1. Again, the proposed approach can accurately identify the true number of subgroups.
	
	\begin{figure}[H]
		\centering
		\includegraphics[width=7.7cm]{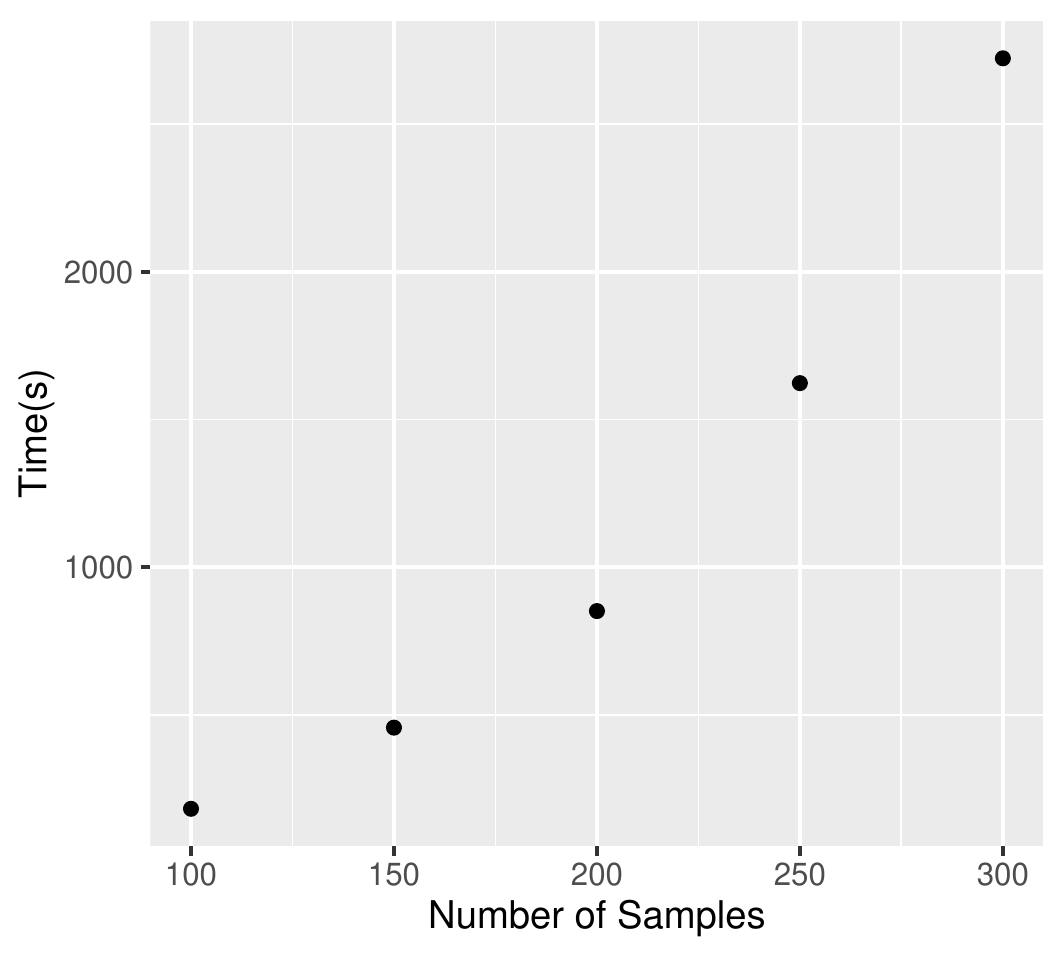}
		\caption{The consumed time (in seconds) of each experiment with different number of samples.}
		\label{timevsnum}
	\end{figure}
	
	Figure \ref{timevsnum} illustrates the consumed time (in seconds) running our algorithm with different sample sizes. It can be seen that the consumed time rapidly increases at a nearly $O(n^2)$ spped as the sample size increases. Therefore, in empirical studies, we randomly sampled 400 monitoring stations to reduce the computation cost.
	
	\begin{table}[]
		\centering
		\caption{Simulation results in Example 1: ARI and MSE based on 100 replications under Scenario 2.}
		\resizebox{\linewidth}{!}{\begin{tabular}{ccc|cc|cc}
				\hline
				\multicolumn{3}{c|}{\multirow{2}{*}{}}            & \multicolumn{2}{c|}{$n=40$} & \multicolumn{2}{c}{$n=200$} \\ \cline{4-7} 
				Structure & $a_{ij}$ & Method            & ARI                            & MSE            & ARI                             & MSE         \\ \hline
				B  & Norm   & Proposed & 0.959(0.003)     & 1.063(0.091)   & 0.991(0.001)        & 0.216(0.014)       \\
				&                       & Oracle   & & 0.911(0.082)     &                  & 0.205(0.012)                  \\
				&                       & FIMR &   0.633(0.032)       &39.492(5.284)                              &    0.844(0.011) &  7.877(2.298)                  \\                    
				&                       & Resp & 0.942(0.005)    & 1.617(0.122)     & 0.964(0.002)   & 0.462(0.024)        \\
				&                       & Resi & 0.893(0.008)   & 2.932(0.376)        & 0.967(0.004)   & 0.731(0.059)     \\
				
				& Unif & Proposed & 0.969(0.004)  & 0.977(0.084)    & 0.996(0.001)      & 0.254(0.015)       \\
				&                       & Oracle   & & 0.835(0.105)     &                  & 0.232(0.012)   \\
				&                       & FIMR &   0.664(0.022)      &                               47.027(7.523) &  0.893(0.023)& 4.295(2.010)                       \\            
				&                       & Resp  & 0.935(0.005)   & 1.542(0.096)       & 0.965(0.002) & 0.360(0.017)        \\
				&                       & Resi  & 0.909(0.007)  & 1.923(0.184)        & 0.974(0.003)  & 0.347(0.029)        \\ \hline
				UB & Norm  & Proposed & 0.964(0.002)& 1.124(0.096)        & 0.995(0.001)   & 0.341(0.014)        \\
				&                       & Oracle  & & 0.853(0.071)     &                  & 0.255(0.013)             \\
				&                       & FIMR &   0.749(0.016)                          &                               25.827(5.638) & 0.972(0.011) & 3.522(1.296)                         \\     
				&                       & Resp  & 0.938(0.007)  & 1.817(0.124)     & 0.973(0.003)   & 0.563(0.026)         \\
				&                       & Resi  & 0.923(0.004)  & 1.931(0.168)       & 0.967(0.002)  & 0.523(0.023)        \\
				& Unif         & Proposed & 0.968(0.007) & 1.230(0.102)      & 0.989(0.002)     & 0.381(0.015)       \\
				&                       & Oracle   & & 1.014(0.096)     &                  & 0.386(0.015)            \\ 
				&                       & FIMR &    0.788(0.013)                         & 27.352(4.985)                               & 0.974(0.009) & 3.369(1.074)                        \\     
				&                       & Resp & 0.946(0.009)  & 1.815(0.126)     & 0.952(0.004)    & 0.667(0.017)         \\
				&                       & Resi   & 0.919(0.010) & 2.043(0.198)       & 0.980(0.004)  & 0.422(0.024)        \\ \hline
				
		\end{tabular}}
		\label{table2}
	\end{table}
	
	Subgrouping and estimation results for Scenario 2 in Example 1 are summarized in Table \ref{table2}. We obtain similar results as in Scenario 1 in the main manuscript. Since the Oracle approach assumes the true subgroup structure, the proposed approach behaves inferior to the Oracle as expected. But it has significant advantages over the FIMR, Resp and Resi approaches. The values of ARI and MSE increase as $n$ grows for the proposed approach. This supports the estimation consistency established in Theorem 1 and 2.
	
	\begin{table}[]
		\centering
		\caption{Simulation results in Example 2: ARI and MSE. Each cell shows the mean(s.d.).}
		\resizebox{\linewidth}{!}{\begin{tabular}{ccc|cc|cc}
				\hline
				\multicolumn{3}{c|}{\multirow{2}{*}{}}            & \multicolumn{2}{c|}{$n=60$} & \multicolumn{2}{c}{$n=300$} \\ \cline{4-7} 
				Structure & $a_{ij}$ & Method            & ARI                            & MSE            & ARI                             & MSE         \\ \hline
				B  & Norm & Proposed  & 0.948(0.006)  & 1.969(0.508)    & 0.969(0.001)     & 0.877(0.021)      \\
				&                       & Oracle  &       & 1.096(0.104)                &     &                0.550(0.018)  \\
				&                      & FIMR &  0.533(0.028)         &395.569(38.380)                      &  0.789(0.013)                      &32.836(7.746) \\
				&                       & Resp   & 0.834(0.009)  & 4.230(0.344)   & 0.855(0.002)     & 2.317(0.052)       \\
				&                       & Resi    & 0.781(0.018)  & 7.304(1.103)    & 0.937(0.003)  & 1.606(0.122)       \\
				& Unif  & Proposed & 0.958(0.007)   & 1.332(0.115)     & 0.996(0.002)    & 0.724(0.027)       \\
				&                       & Oracle &   & 0.977(0.051)    & & 0.341(0.019)   \\
				&                      & FIMR &  0.588(0.025)                        &                     402.457(42.825) & 0.763(0.021)       & 24.294(5.283) \\ 
				&                       & Resp   & 0.772(0.016)  & 6.269(0.481)      & 0.795(0.009)  & 2.997(0.153)       \\
				&                       & Resi  & 0.752(0.015)   & 10.045(1.225)       & 0.903(0.003)   & 1.846(0.127)    \\\hline
				UB  & Norm & Proposed & 0.944(0.005)   & 1.627(0.178)     & 0.987(0.002)       & 0.906(0.022)    \\
				&                       & Orcale   & & 0.921(0.098)      &                 & 0.415(0.020)        \\
				&                      & FIMR &  0.611(0.030)         &284.565(28.210)                      &    0.876(0.012)            & 9.397(1.942)\\ 
				&                       & Resp & 0.608(0.025) & 6.386(0.553)      & 0.694(0.016)     & 4.142(0.279)      \\
				&                       & Resi     & 0.751(0.016)  & 7.336(0.751)     & 0.825(0.005)   & 2.618(0.197)     \\ 
				& Unif & Proposed   & 0.954(0.008)  & 1.511(0.285)         & 0.970(0.002) & 0.739(0.024)        \\
				&                       & Oracle  &  & 1.042(0.138)      &                 & 0.412(0.015)                    \\ 
				&                      & FIMR & 0.632(0.027)   & 269.922(25.686)                      &  0.843(0.021)       & 13.415(2.840) \\ 
				&                       & Resp   & 0.619(0.012)   & 7.991(0.368)    & 0.677(0.008) & 5.124(0.481)         \\
				&                       & Resi  & 0.787(0.017)  & 7.193(0.815)     & 0.863(0.005)     & 2.354(0.125)     \\ \hline
				
		\end{tabular}}\\
		\raggedright
		\label{table3}
	\end{table}
	
	Table \ref{table3} summarizes the results of subgrouping and estimation in Example 2. Similar to Example 1, aside from the Oracle approach, the ARI of the proposed approach is always the highest; therefore, it outperforms alternatives in the accuracy of identifying subgroup memberships. In terms of estimation accuracy, the proposed approach is again observed to have favorable performance. Compared to Example 1 where there are only two subgroups, here, the advantages of the proposed approach in terms of subgrouping and estimation are more significant. Consider for example $n=60$, balanced structure, and $a_{il}$'s are generated from norm distribution $N(2,1)$. The proposed approach has (ARI, MSE)=(0.948, 1.969), which is superior to the alternatives: (0.533, 395.569) for the FIMR approach, (0.834, 4.230) for the Resp approach, and (0.781, 7.304) for the Resi approach. FIMR has the worst performance.
	
	\begin{table}
		\caption{Data analysis: comparison of grouping results. In each cell, the value of NMI.}
		\begin{center}
			\setlength{\tabcolsep}{2.5pt}
			\renewcommand{\arraystretch}{1.0}
			\begin{tabular}{lcccc}
				\hline
				&  Proposed & FIMR &Resp & Resi \\ \hline
				Proposed &    1     &   0.092  &   0.203      &  0.150  \\
				FIMR         &         &   1       &    0.126     & 0.183 \\
				Resp        &          &          &     1     &  0.336  \\
				Resi        &           &          &          &   1 \\ 
				
				\hline
			\end{tabular}
		\end{center}
		\label{tab:NMI}
	\end{table}
	
	Table \ref{tab:NMI} presents the Normalized Mutual Information (NMI) which measures the similarity between the subgrouping results, with range $[0,1]$ and a larger value indicating a higher degree of similarity. We find that the proposed approach has low similarity with the alternatives because the NMIs are low.


\bibliographystyle{myjmva}
\bibliography{cita2}

\begin{thebibliography}{25}
\expandafter\ifx\csname natexlab\endcsname\relax\def\natexlab#1{#1}\fi
\providecommand{\bibinfo}[2]{#2}
\ifx\xfnm\relax \def\xfnm[#1]{\unskip,\space#1}\fi
\bibitem[{Blanchard and Hidy(2005)}]{Blanchard2005}
\bibinfo{author}{C.~Blanchard}, \bibinfo{author}{G.~Hidy},
  \bibinfo{title}{Effects of so2 and nox emission reductions on pm2.5 mass
  concentrations in the southeastern united states}, \bibinfo{journal}{Journal
  of the Air and Waste Management Association} \bibinfo{volume}{55}
  (\bibinfo{year}{2005}) \bibinfo{pages}{265--272}.
\bibitem[{Boor(1978)}]{Boor1978}
\bibinfo{author}{C.~Boor}, \bibinfo{title}{A practical guide to splines},
  Applied Mathematical Sciences, Vol. 27, \bibinfo{publisher}{Springer},
  \bibinfo{year}{1978}.
\bibitem[{Bosq(2000)}]{Bosq2000}
\bibinfo{author}{D.~Bosq}, \bibinfo{title}{Linear processes in function spaces
  : theory and applications}, \bibinfo{journal}{Lecture Notes in Statistics}
  \bibinfo{volume}{149} (\bibinfo{year}{2000}) \bibinfo{pages}{181--202}.
\bibitem[{Boyd et~al.(2011)Boyd, Parikh, Chu, Peleato and EcKstein}]{Boyd2011}
\bibinfo{author}{S.~Boyd}, \bibinfo{author}{N.~Parikh},
  \bibinfo{author}{E.~Chu}, \bibinfo{author}{B.~Peleato},
  \bibinfo{author}{J.~EcKstein}, \bibinfo{title}{Distributed optimization and
  statistical learning via the alternating direction method of multipliers},
  \bibinfo{journal}{Foundations and Trends in Machine Learning}
  \bibinfo{volume}{3} (\bibinfo{year}{2011}) \bibinfo{pages}{1--122}.
\bibitem[{Cardot et~al.(1999)Cardot, Ferraty and Sarda}]{cardot99}
\bibinfo{author}{H.~Cardot}, \bibinfo{author}{F.~Ferraty},
  \bibinfo{author}{P.~Sarda}, \bibinfo{title}{Functional linear model},
  \bibinfo{journal}{Statistics and Probability Letters} \bibinfo{volume}{45}
  (\bibinfo{year}{1999}) \bibinfo{pages}{11--22}.
\bibitem[{Cardot et~al.(2003)Cardot, Ferraty and Sarda}]{cardot03}
\bibinfo{author}{H.~Cardot}, \bibinfo{author}{F.~Ferraty},
  \bibinfo{author}{P.~Sarda}, \bibinfo{title}{Spline estimators for the
  functional linear model}, \bibinfo{journal}{Statistica Sinica}
  \bibinfo{volume}{13} (\bibinfo{year}{2003}) \bibinfo{pages}{571--591}.
\bibitem[{Claeskens et~al.(2009)Claeskens, Tatyana and Jean}]{Claeskens2009}
\bibinfo{author}{G.~Claeskens}, \bibinfo{author}{K.~Tatyana},
  \bibinfo{author}{D.~O. Jean}, \bibinfo{title}{Asymptotic properties of
  penalized spline estimators}, \bibinfo{journal}{Biometrika}
  \bibinfo{volume}{96} (\bibinfo{year}{2009}) \bibinfo{pages}{529--544}.
\bibitem[{Fan and Li(2001)}]{Fan2001}
\bibinfo{author}{J.~Fan}, \bibinfo{author}{R.~Li}, \bibinfo{title}{Variable
  selection via nonconcave penalized likelihood and its oracle properties},
  \bibinfo{journal}{Journal of the American Statistical Association}
  \bibinfo{volume}{96} (\bibinfo{year}{2001}) \bibinfo{pages}{1348--1360}.
\bibitem[{Fang et~al.(2011)Fang, Fu and C.M.~Lee}]{yao12}
\bibinfo{author}{Y.~Fang}, \bibinfo{author}{T.~Fu},
  \bibinfo{author}{T.~C.M.~Lee}, \bibinfo{title}{Functional mixture
  regression}, \bibinfo{journal}{Biostatistics} \bibinfo{volume}{12}
  (\bibinfo{year}{2011}) \bibinfo{pages}{341--53}.
\bibitem[{Garoni et~al.(2014)Garoni, Manni, Pelosi, Serra-Capizzano and
  Speleers}]{Garoni2014}
\bibinfo{author}{C.~Garoni}, \bibinfo{author}{C.~Manni},
  \bibinfo{author}{F.~Pelosi}, \bibinfo{author}{S.~Serra-Capizzano},
  \bibinfo{author}{H.~Speleers}, \bibinfo{title}{On the spectrum of stiffness
  matrices arising from isogeometric analysis}, \bibinfo{journal}{Numerische
  Mathematik} \bibinfo{volume}{127} (\bibinfo{year}{2014})
  \bibinfo{pages}{751--799}.
\bibitem[{Hilgert et~al.(2013)Hilgert, Mas and Verzelen}]{Hilgert2013}
\bibinfo{author}{N.~Hilgert}, \bibinfo{author}{A.~Mas},
  \bibinfo{author}{N.~Verzelen}, \bibinfo{title}{Minimax adaptive tests for the
  functional linear model}, \bibinfo{journal}{The Annals of Statistics}
  \bibinfo{volume}{41} (\bibinfo{year}{2013}) \bibinfo{pages}{838--869}.
\bibitem[{Hu et~al.(????)Hu, Huang, Liu, Sun and Zhao}]{Hu2021}
\bibinfo{author}{X.~Hu}, \bibinfo{author}{J.~Huang}, \bibinfo{author}{L.~Liu},
  \bibinfo{author}{D.~Sun}, \bibinfo{author}{X.~Zhao}, \bibinfo{title}{Subgroup
  analysis in the heterogeneous cox model}, \bibinfo{journal}{Statistics in
  Medicine} \bibinfo{volume}{40} (????) \bibinfo{pages}{739--757}.
\bibitem[{Hua et~al.(2021)Hua, Zhang, Foy, Mei, Shang and Feng}]{Hua2020}
\bibinfo{author}{J.~Hua}, \bibinfo{author}{Y.~Zhang}, \bibinfo{author}{B.~D.
  Foy}, \bibinfo{author}{X.~Mei}, \bibinfo{author}{J.~Shang},
  \bibinfo{author}{C.~Feng}, \bibinfo{title}{Competing pm2.5 and no2 holiday
  effects in the beijing area vary locally due to differences in residential
  coal burning and traffic patterns}, \bibinfo{journal}{Science of The Total
  Environment} \bibinfo{volume}{750} (\bibinfo{year}{2021})
  \bibinfo{pages}{141575}.
\bibitem[{Lan et~al.(2010)Lan, Qu and Zhou}]{Lan2010}
\bibinfo{author}{X.~Lan}, \bibinfo{author}{A.~Qu}, \bibinfo{author}{J.~Zhou},
  \bibinfo{title}{Consistent model selection for marginal generalized additive
  model for correlated data}, \bibinfo{journal}{Journal of the American
  Statistical Association} \bibinfo{volume}{105} (\bibinfo{year}{2010})
  \bibinfo{pages}{1518--1530}.
\bibitem[{Lyche et~al.(2018)Lyche, Manni and Speleers}]{Lyche2018}
\bibinfo{author}{T.~Lyche}, \bibinfo{author}{C.~Manni},
  \bibinfo{author}{H.~Speleers}, \bibinfo{title}{Foundations of spline theory:
  B-splines, spline approximation, and hierarchical refinement},
  \bibinfo{title}{Foundations of Spline Theory: B-Splines, Spline
  Approximation, and Hierarchical Refinement}, \bibinfo{publisher}{Springer},
  \bibinfo{year}{2018}, pp. \bibinfo{pages}{34--56}.
\bibitem[{Ma and Huang(2017)}]{Ma2017}
\bibinfo{author}{S.~Ma}, \bibinfo{author}{J.~Huang}, \bibinfo{title}{A concave
  pairwise fusion approach to subgroup analysis}, \bibinfo{journal}{Journal of
  the American Statistical Association} \bibinfo{volume}{112}
  (\bibinfo{year}{2017}) \bibinfo{pages}{410--423}.
\bibitem[{Ma et~al.(2020)Ma, Zhang, Huang and Liu}]{Ma2020}
\bibinfo{author}{S.~Ma}, \bibinfo{author}{Z.~Zhang},
  \bibinfo{author}{J.~Huang}, \bibinfo{author}{M.~Liu},
  \bibinfo{title}{Exploration of heterogeneous treatment effects via concave
  fusion}, \bibinfo{journal}{The International Journal of Biostatistics}
  \bibinfo{volume}{16} (\bibinfo{year}{2020}). \bibinfo{note}{Online}.
\bibitem[{Preda and Saporta(2005)}]{preda05}
\bibinfo{author}{C.~Preda}, \bibinfo{author}{G.~Saporta},
  \bibinfo{title}{Clusterwise pls regression on a stochastic process},
  \bibinfo{journal}{Computational Statistics and Data Analysis}
  \bibinfo{volume}{48} (\bibinfo{year}{2005}) \bibinfo{pages}{149--158}.
\bibitem[{Ramsay and Silverman(2002)}]{Ramsay02}
\bibinfo{author}{J.~O. Ramsay}, \bibinfo{author}{B.~W. Silverman},
  \bibinfo{title}{Applied Functional data analysis},
  \bibinfo{publisher}{Vol.77. Springer}, \bibinfo{address}{New York},
  \bibinfo{year}{2002}.
\bibitem[{Ramsay and Silverman(2005)}]{Ramsay05}
\bibinfo{author}{J.~O. Ramsay}, \bibinfo{author}{B.~W. Silverman},
  \bibinfo{title}{Functional data analysis}, \bibinfo{publisher}{Springer},
  \bibinfo{address}{New York}, \bibinfo{year}{2005}.
\bibitem[{Tibshirani et~al.(2005)Tibshirani, Saunders, Rosset, Zhu and
  Knight}]{Tibshirani05}
\bibinfo{author}{R.~Tibshirani}, \bibinfo{author}{M.~Saunders},
  \bibinfo{author}{S.~Rosset}, \bibinfo{author}{J.~Zhu},
  \bibinfo{author}{K.~Knight}, \bibinfo{title}{Sparsity and smoothness via the
  fused lasso}, \bibinfo{journal}{Journal of the Royal Statistical Society}
  \bibinfo{volume}{67} (\bibinfo{year}{2005}) \bibinfo{pages}{91--108}.
\bibitem[{Wang et~al.(2009)Wang, Li and Leng}]{Wang2009}
\bibinfo{author}{H.~Wang}, \bibinfo{author}{B.~Li}, \bibinfo{author}{C.~Leng},
  \bibinfo{title}{Shrinkage tuning parameter selection with a diverging number
  of parameters}, \bibinfo{journal}{Journal of the Royal Statistical Society:
  Series B (Statistical Methodology)} \bibinfo{volume}{71}
  (\bibinfo{year}{2009}) \bibinfo{pages}{671--683}.
\bibitem[{Zhang(2010)}]{Zhang2010}
\bibinfo{author}{C.~Zhang}, \bibinfo{title}{Nearly unbiased variable selection
  under minimax concave penalty}, \bibinfo{journal}{The Annals of Statistics}
  \bibinfo{volume}{38} (\bibinfo{year}{2010}) \bibinfo{pages}{894--942}.
\bibitem[{Zhou et~al.(1998)Zhou, Shen and Wolfe}]{Zhou1998}
\bibinfo{author}{S.~Zhou}, \bibinfo{author}{X.~Shen}, \bibinfo{author}{D.~A.
  Wolfe}, \bibinfo{title}{Local asymptotics for regression splines and
  confidence regions}, \bibinfo{journal}{The Annals of Statistics}
  \bibinfo{volume}{26} (\bibinfo{year}{1998}) \bibinfo{pages}{1760--1782}.
\bibitem[{Zhu and Qu(2018)}]{Zhu2018}
\bibinfo{author}{X.~Zhu}, \bibinfo{author}{A.~Qu}, \bibinfo{title}{Cluster
  analysis of longitudinal profiles with subgroups},
  \bibinfo{journal}{Electronic Journal of Statistics} \bibinfo{volume}{12}
  (\bibinfo{year}{2018}) \bibinfo{pages}{171--193}.

\end{thebibliography}

\end{document}